\documentclass{article}
\pdfoutput=1
\usepackage[T1]{fontenc}
\usepackage[latin9]{inputenc}
\usepackage[english]{babel}
\usepackage{algorithm2e}
\usepackage{amsthm}
\usepackage{amsmath}
\usepackage{amssymb}
\usepackage{tikz}
\usetikzlibrary{positioning,shadows,arrows,shapes,arrows,calc,decorations.pathmorphing}
\usepackage{url}
\usepackage{setspace}
\usepackage{graphicx}
\usepackage[unicode=true,pdfusetitle,
 bookmarks=true,bookmarksnumbered=false,bookmarksopen=false,
 breaklinks=false,pdfborder={0 0 1},backref=false,colorlinks=false]
 {hyperref}
 \usepackage{enumerate}

\definecolor{mygray1}{gray}{0.8}
\definecolor{mygray2}{gray}{0.6}

\newcommand{\contr}{\leadsto}
\newcommand{\intcontr}{\contr_{\text{int}}}
\newcommand{\clcontr}{\mapstochar\mathrel{\mspace{0.45mu}}\contr}

\makeatletter

\theoremstyle{plain}
\newtheorem{thm}{\protect\theoremname}
\theoremstyle{definition}
\newtheorem{defn}[thm]{\protect\definitionname}
\theoremstyle{plain}
\newtheorem{prop}[thm]{\protect\propositionname}
\theoremstyle{plain}
\newtheorem{lem}[thm]{\protect\lemmaname}
\theoremstyle{plain}
\RestyleAlgo{ruled}

\makeatother

  \providecommand{\definitionname}{Definition}
  \providecommand{\lemmaname}{Lemma}
  \providecommand{\propositionname}{Proposition}
\providecommand{\theoremname}{Theorem}

\begin{document}

\title{Finding AND-OR Hierarchies in Workflow Nets}

\author{Jacek Sroka\\
Institute of Informatics\\
University of Warsaw, Poland\\
\small \url{sroka@mimuw.edu.pl}
\and Jan Hidders\\
Vrije Universiteit Brussel, Belgium\\
\small \url{jan.hidders@vub.ac.be}
} 

\date{}

\maketitle

\begin{abstract}
This paper presents the notion of AND-OR reduction,
which reduces a WF net to a smaller net by iteratively contracting  certain
well-formed subnets into single nodes until no more such contractions
are possible. This reduction can reveal the hierarchical structure of a WF net,
and since it preserves certain semantical properties such as soundness, it can
help with analysing and understanding why a WF net is sound or not. The reduction
can also be used to verify if a WF net is an AND-OR net. This class of WF nets was introduced
in earlier work, and arguably describes nets that follow good hierarchical design principles.
It is shown that the AND-OR reduction is confluent up to isomorphism, which
means that despite the inherent non-determinism that comes from
the choice of subnets that are contracted, the final result of the reduction is always
the same up to the choice of the identity of the nodes. Based on this result, a polynomial-time algorithm
is presented that computes this unique result of the AND-OR reduction. Finally, it is shown how this
algorithm can be used to verify if a WF net is an AND-OR net.
\end{abstract}

\section{Introduction\label{sect:Introduction}}

Petri nets \cite{Reisig:2008} are one of the most popular and well
studied formalisms for modeling processes. Their graphical notation
is easy to understand, but at the same time concrete and formal, which
allows for reasoning over the complex systems that are being modeled.
Petri nets are especially useful for business processes and business
workflows for which a specific class of Petri nets, called \emph{workflow
nets}, was introduced \cite{Aalst1998workflow,VanDerAalst:2003}.
Even though other notations are used in most industrial process modeling
tools like Business Process Modeling Notation (BPMN) \cite{BPMN},
Business Process Execution Language (BPEL) or Event-driven Process
Chain (EPC) \cite{keller:1992}, the control flow aspect
of the models expressed in those notations can be translated to workflow
nets. At the same time workflow nets are considered to be the goto
formalism for workflow analysis, like detecting possible problems,
e.g., existence of deadlocks or livelocks, and for investigating the
principles of workflow modeling without focusing on a particular language.

Workflow models that lack those problems are called \emph{sound}, 
and the first definition of workflow-net soundness was proposed by van der Aalst
in \cite{Aalst1998workflow}. Quickly several alternative definitions
of soundness, varying in strictness and verification difficulty, emerged. Examples of these are weak soundness \cite{Martens_oncompatibility,Martens:2005:AWS:2136587.2136592},
relaxed soundness \cite{relaxed_soundness}, lazy soundness \cite{puhlmann_BPM2006},
$k$-soundness and generalised soundness \cite{DBLP:conf/apn/HeeSV03,Hee04}, up-to-$k$-soundness
\cite{toornThesis} and substitution soundness \cite{DBLP:journals/is/SrokaH14}.
Informally, the original notion of soundness guarantees two properties of the net. First,
that if we initiate the workflow net correctly, then no matter how
the execution proceeds, we can always end up in a proper final state.
Second, that every subtask can be potentially executed in some run
of the workflow. An overview of the research on the different types of soundness of workflow
nets and their decidability can be found in \cite{journals/fac/AalstHHSVVW11}.

In earlier research \cite{DBLP:journals/is/SrokaH14} we have proposed
a new notion of soundness, namely the substitution soundness or sub-soundness
for short. It is similar to $k$- and {*}-soundness studied in
\cite{DBLP:conf/apn/HeeSV03}, but captures exactly the conditions necessary for building
complex workflow nets by following a structured approach where subsystems
with multiple inputs and outputs are used as building blocks of larger systems. As was shown in that
research, it is not enough for such subsystems to be classically sound
by themselves. It may be the case, for example, that if a sound WF net is used inside
another sound WF net, that the nested WF net is used to execute
several simultaneous computations which can interfere and cause the whole WF net
to become unsound. Or it can be that partial
results, represented by tokens in the output places of the nested WF net, are consumed
prematurely by the containing WF net before the nested WF net has finished properly. 
This is prevented by the notion of substitution soundness (or sub-soundness),
which is informally defined as follows: a WF net is sub-sound iff after initiating it with $k$ tokens in every input place and letting it execute it will always be able to finish by producing $k$ tokens for every output place even if during the run the output tokens are removed by some external transitions.

Although stronger soundness properties may be desirable, they are often also more difficult
to verify.  For this reason, a method is introduced in \cite{DBLP:journals/is/SrokaH14} 
for systematically constructing workflow nets so that they
are guaranteed to satisfy the sub-soundness property. This method is
in principle, and in effect, similar to methods employed in software
engineering, where complexity is tackled by separation of concerns
and encapsulation, and systems are divided into building blocks such as
modules, objects and functions, which in turn can be decomposed further.

We follow those good practices in the context of workflow nets where
they, like in software engineering, allow to avoid common pitfalls. Similarly to general programming
languages, also for workflow nets, patterns and anti-patterns have
been published \cite{VanDerAalst:2003,Trcka:2009:DAD:1573487.1573532}.
Also similarly to general programming, it is beneficial to organise
the workflow models in a structured way. In programming the ideas
of using macros, subroutines, procedures, functions, and later on,
classes, proved that even extremely complex systems can be programmed
and maintained in a practical and effective manner. Such structurisation was successfully
applied to designing complex Petri nets 
\cite{WangWei2009,Suzuki198351,Devillers:1997me,conf/stacs/BestDE93,conf/apn/EsparzaS90,Polyvyanyy2012518}
and workflow nets \cite{wachtel2003,reviewer2b,DBLP:journals/is/SrokaH14}.
As with general programming, the system is composed of small, separated
fragments, which are easier to understand and maintain. Fragments
can include invocations of other fragments, which can include other
nested fragments, and so on. 
% Non-recursive invocations are semantically
% equivalent to folding fragments of the program/system.

The class of nets we introduced in \cite{DBLP:journals/is/SrokaH14}
is called the class of AND-OR nets (see Section~\ref{sect:AND-OR-nets}).
This class is larger and more general than other classes of workflow nets
generated with a similar type of structural approach, as presented for example
in \cite{DBLP:conf/apn/HeeSV03,wachtel2003}. Apart from studying conditions necessary
for structured workflow systems to be {*}-sound, it was shown in  \cite{DBLP:journals/is/SrokaH14} that all
AND-OR nets indeed are sub-sound. In this paper we continue this line of research
and introduce a method to determine the hierarchical structure of a WF net, or parts of it, that was
not necessarily designed in such a structured way. In \cite{DBLP:journals/is/SrokaH14}
the AND-OR nets were defined as all the nets that can be constructed
with a top-down refinement procedure, by using nets of certain basic
classes similar to S/T systems. In this paper we show that at the same time
AND-OR nets that were not necessarily constructed in such a way, can be analysed to
determine a refinement hierarchy with a bottom-up
reduction procedure that contracts subnets of the basic AND-OR classes. Moreover, it is shown
that finding occurrences of such subnets can be done in polynomial
time.

A key result in this paper is that the procedure of contracting
subnets of the basic AND-OR classes is confluent and therefore
the reduction will always return the same result, independent of how
the subnets where selected for contraction. It is shown in this paper that this can be used to turn the procedure
into a polynomial algorithm and therefore a tractable method for 
determining an AND-OR refinement tree. Next to that, it can also be determined 
if a net is an AND-OR by checking if the reduction procedure reduces it to 
a one-node WF net. If a
net is positively identified as an AND-OR net, it is consequently also guaranteed
to be {*}-sound and sub-sound\footnote{It
follows straightforwardly from the definitions of *-soundness and sub-soundness
that the latter implies the first.}, i.e., can be used as a building block
of larger systems. A first example of an application of this result would
be a scenario where a process modeller constructs a complex model from
submodels published in some repository. He or she may want to make sure that
the submodels follow good design principles and are sub-sound, which means that they can
be safely used as building blocks of a composite model. The repository
can contain models for subunits in some organisational structure,
e.g., models for faculties of an university or departments of a company
or even models from some global repository of socially shared workflows,
which appear in e-science~\cite{myExperiment}.

The reduction algorithm can not only be used for AND-OR nets, but also for 
the analysis of {*}-soundness and sub-soundness of general workflow nets. It 
can help the user with finding problems causing unsoundness. More concretely,
if the result of AND-OR verification is negative, then the reduction algorithm
stops without reaching a one-node WF net. This resulting net can
serve as a condensed version of the original net and point the user to the
source of the problem in the design. Note that a WF net
may be not an AND-OR net, but still be {*}-sound or sub-sound. We conjecture,
but have not proven, that to verify {*}-soundness or sub-soundness
of an arbitrary net, it is enough to verify {*}-soundness or sub-soundness
of the net resulting from AND-OR verification procedure. The contractions
used in our algorithm would have to be proven to preserve {*}-soundness
or sub-soundness, similarly as for example rules of \cite{murata_reductions}
preserve liveness and boundedness. This would give a symmetric and
probably similarly laborious result to \cite{DBLP:journals/is/SrokaH14},
where it was shown that substitutions of AND-OR nets into AND-OR nets
preserve sub-soundness,
from which it follows that they also preserve {*}-soundness.
The reduced net, resulting from AND-OR net verification procedure,
could then undergo a proper soundness verification with similar methods
as in \cite{verification_thesis,Verbeek01a}. Furthermore, limiting
the size of the verified net with hierarchical methods can be helpful
for users struggling with understanding the reasons for unsoundness
of workflow nets. That this is often a problem, even when using automated
verification tools, is for example reported in \cite{Flender_visualisation_of_soundness}.

Finally, as a byproduct of a successful reduction, a tree structure
describing the nesting of the fragments of the net can be determined. As with similar
methods \cite{wachtel2003,wachtel2006,PChPGBL13}, which deal with
workflow net class which is a proper subclass of AND-OR nets, such
a tree structure can be used for modeling recovery regions or determining
sound markings, or just for better understanding the structure of the workflow
net and its properties. The latter can for example help with determining
how parts of the workflow can best be distributed to independent organisational
units or to different servers in case of workflows representing computations,
e.g., as in scientific workflows.

% Structure tree refinements were also studied in~\cite{soundness_checking}, where single-input-single-output nets were considered with respect to soundness and the authors shows how this could help with the state explosion problem. 

In related work \cite{graph_decomposition} a set of heuristics was proposed to find appropriate decomposition boundaries, which results in a refinement tree for a given graph. Our approach, however concentrates specifically on workflow nets which are generalised to allow multiple inputs and outputs, and it is closely tied in a well understood manner to their semantics and soundness properties.

The outline of the remainder of this paper is as follows. In Section~\ref{sect:Basic-terminology-and} the basic terminology of WF nets and their semantics is introduced. In Section~\ref{sect:AND-OR-nets} the class of AND-OR nets is introduced, based on the notions of place and transition substitution, where a node is replaced with a WF net. In Section~\ref{sect:AND-OR-reduction} the notion of AND-OR reduction is introduced, which is based on the notion of contraction, where certain well-formed subnets of WF nets are contracted into single nodes. It is discussed here how this reduction process is confluent in that it returns a unique result up to the choice of the identity of the nodes. This is based on the observation that the process is locally confluent, but since the proof of this observation is quite involved, it is presented separately in Section~\ref{sect:confluence-proof}. In Section~\ref{sect:Algorithm-AND-OR-nets} a concrete polynomial algorithm for computing the result of the AND-OR reduction is presented, and it is shown how it can be used to verify if a WF net is an AND-OR net. Finally, in Section~\ref{sect:summary} a summary of the results is given, and potential future research directions are discussed.

\section{Basic terminology and definitions\label{sect:Basic-terminology-and}}

Let $S$ be a set. A bag (multiset) $m$ over $S$ is a function $m:S\rightarrow\mathbb{N}$.
We use $+$ and $-$ for the sum and the difference of two bags and
$=$, $<$, $>$, $\le$, $\ge$ for comparisons of bags, which are
defined in the standard way. We overload the set notation, writing
$\emptyset$ for the empty bag and $\in$ for the element inclusion.
We list elements of bags between brackets, e.g. $m=[p^{2},q]$ for
a bag $m$ with $m(p)=2$, $m(q)=1$, and $m(x)=0$ for all $x\notin\{p,q\}$.
The shorthand notation $k.m$ is used to denote the sum of $k$ bags
$m$. The size of a bag $m$ over $S$ is defined as $|m|=\Sigma_{s\in S}m(s)$.

\begin{defn}
[Petri net] A \emph{Petri net} is a tuple $N=(P,T,F)$ where $P$ is
a finite set of places, $T$ is a finite set of transitions such that
$P\cap T=\emptyset$ and $F\subseteq(T\times P)\cup(P\times T)$ the
set of flow edges.
\end{defn}

We will refer to the elements of $P \cup T$ also as \emph{nodes} in Petri net. We say that the \emph{type} of a node is \emph{place} or \emph{transition} if it is in $P$ or $T$, respectively.

A path in a net is a non-empty sequence $(n_{1},...,n_{m})$ of nodes
where for all $i$ such that $1\leq i\leq n-1$ it holds that $(n_{i},n_{i+1})\in F$.
Markings are states (configurations) of a net and the set of markings
of $N=(P,T,F)$ is the set of all bags over $P$ and is denoted as $\mathbf{M_{N}}$.
Given a transition $t\in T$, the preset $\bullet t$ and the postset
$t\bullet$ of $t$ are the sets $\{p\mid(p,t)\in F\}$ and $\{p\mid(t,p)\in F\}$,
respectively. In a similar fashion we write $\bullet p$ and $p\bullet$ for
pre- and postsets of places, respectively. To emphasise the fact that
the preset (postset) is considered within some net $N$, we write $\bullet_{N}a$,
$a\bullet_{N}$. 
% We overload this notation further allowing to apply
% preset and postset operations to a set $B$ of places/transitions,
% which is defined as the union of pre-/postsets of elements of $B$.
We overload this notation by letting $\bullet a$ ($a \bullet$) also 
denote the bags of nodes that (1) contain all nodes in the preset (postset) of $a$ 
exactly once and (2) contains no other nodes.
A transition $t\in T$ is said to be enabled at marking $m$ iff $\bullet t \leq m$.
For a net $N=(P,T,F)$ with markings $m_{1}$ and $m_{2}$ and a transition
$t\in T$ we write $m_{1}\stackrel{t}{\longrightarrow}_{N}m_{2}$,
if $t$ is enabled at $m_{1}$ and $m_{2}=m_{1}-\bullet t+t\bullet$.
For a sequence of transitions $\sigma=(t_{1},\ldots,t_{n})$ we write
$m_{1}\stackrel{\sigma}{\longrightarrow}_{N}m_{n+1}$, if $m_{1}\stackrel{t_{1}}{\longrightarrow}_{N}m_{2}\stackrel{t_{2}}{\longrightarrow}_{N}\ldots\stackrel{t_{n}}{\longrightarrow}_{N}m_{n+1}$,
and we write $m_{1}\stackrel{*}{\longrightarrow}_{N}m_{n+1}$, if there
exists such a sequence $\sigma\in T^{*}$. We will write $m_{1}\stackrel{t}{\longrightarrow}m_{2}$,
$m_{1}\stackrel{\sigma}{\longrightarrow}m_{n+1}$ and $m_{1}\stackrel{*}{\longrightarrow}m_{n+1}$,
if $N$ is clear from the context.

We now introduce the notion of Workflow net, which is a Petri net where certain places and transitions are marked as input and output nodes.

\begin{defn}
[I/O net] An \emph{I/O net} is a tuple $N=(P,T,F,I,O)$
where $(P,T,F)$ is a Petri net with a non-empty set
$I\subseteq P \cup T$ of \emph{input places} and a non-empty set $O\subseteq P \cup T$
of \emph{output places}. 
\end{defn}

In our setting we will restrict ourselves to I/O nets where input and output nodes are either all places, or all transitions.

\begin{defn}
[I/O consistent I/O net] An I/O net $N=(P,T,F,I,O)$ is called \emph{I/O consistent} if
$I \cup O \subseteq P$ or $I \cup O \subseteq T$.
\end{defn}

As is usual for Petri nets that model workflows, we will also require that all nodes in the net can be reached from an input node, and from all nodes in the net an output node can be reached.

\begin{defn}
[Well-connected I/O net] An I/O net $N=(P,T,F,I,O)$
is called \emph{well-connected} if (1) every node in $P \cup T$ is
reachable by a path from at least one node in $I$ and (2) from every
node in $P \cup T$ we can reach at least one node in $O$. 
\end{defn}

The formal definition of Workflow net is then as follows.

\begin{defn}
[Workflow net] A \emph{workflow net} (WF net) is a I/O net $N=( P,T,F,I,O)$ that is I/O consistent
and well-connected.
\end{defn}

If $I \cup O \subseteq P$, then
we call $N$ a \emph{place workflow net} (pWF net), and if $I \cup O \subseteq T	$,  then
a \emph{transition workflow net} (tWF net). The \emph{I/O type} of a WF net is
the type of its input and output nodes, i.e., it is \emph{place} if it is pWF net, and \emph{transition} if it is a tWF net.

Note that input places can have incoming edges in a workflow net, and that output places can have outgoing edges. We will refer to the nodes in $I \cup O$  as the \emph{interface nodes} of the net. We will call a workflow net a \emph{one-input} workflow net if
$I$ contains one element, and a \emph{one-output} workflow net if
$O$ contains one element. Often, as in \cite{Aalst1998workflow},
workflow nets are restricted to one-input one-output place workflow
nets. We generalise this in two ways: first by allowing also nets with input
and output transitions rather than input and output places, and second by allowing multiple input and
output places/transitions. For these generalised workflow nets we
define the corresponding one-input one-output pWF net as follows.
The \emph{place-completion} of a tWF net $N=(P,T,F,I,O)$
is denoted as $\mathbf{pc}(N)$ and is a one-input one-output pWF
net that is constructed from $N$ by adding places $p_{i}$ and $p_{o}$
such that $p_{i}\bullet=I$ and $\bullet p_{o}=O$ and setting the
input set and output set as $\{p_{i}\}$ and $\{p_{o}\}$, respectively.
This is illustrated in Figure~\ref{fig:place-compl} (a). In such
diagrams we will indicate nodes in $I$ with an unconnected incoming
arrow and nodes in $O$ with an unconnected outgoing arrow. The \emph{transition-completion}
of a pWF net $N=( P,T,F,I,O)$ is denoted as $\mathbf{tc}(N)$
and is a one-input one-output tWF net that is constructed from $N$
by adding transitions $t_{i}$ and $t_{o}$ such that $t_{i}\bullet=I$
and $\bullet t_{o}=O$ and setting the input set and output set as
$\{t_{i}\}$ and $\{t_{o}\}$, respectively. This is illustrated in
Figure~\ref{fig:place-compl} (b). 

\begin{figure}[htb]
\begin{center}
\resizebox{0.8\textwidth}{!}{%
\begin{tikzpicture}
    \tikzstyle{transition} = [rectangle,draw,minimum width=0.47cm, minimum height=0.47cm,fill=white]
    \tikzstyle{place} = [circle,draw,minimum width=0.5cm, minimum height=0.5cm,fill=white,inner sep=0.08cm]

\node (ps-res-cl2) [cloud, draw,cloud puffs=20,cloud puff arc=100, aspect=2,fill=white,minimum width=3cm, minimum height=2.5cm,fill=lightgray] {\Large $N$};

\node (ps-res-t1) [transition,above left=-0.5cm and -0.6cm of ps-res-cl2,fill=white] {};
\node (ps-res-t2) [transition,below left=-0.5cm and -0.6cm of ps-res-cl2,fill=white] {};
\node (ps-res-t3) [transition,above right=-0.4cm and -0.7cm of ps-res-cl2,fill=white] {};
\node (ps-res-t4) [transition,right=-0.7cm of ps-res-cl2,fill=white] {};
\node (ps-res-t5) [transition,below right=-0.4cm and -0.7cm of ps-res-cl2,fill=white] {};

\node (ps-res-pi) [place,left=0.7cm of ps-res-cl2,fill=white] {$p_i$};
\node (ps-res-po) [place,right=0.7cm of ps-res-cl2,fill=white] {$p_o$};

\path (ps-res-pi) edge[-latex] (ps-res-t1) ; 
\path (ps-res-pi) edge[-latex] (ps-res-t2) ; 
\path (ps-res-t3) edge[-latex] (ps-res-po) ; 
\path (ps-res-t4) edge[-latex] (ps-res-po) ; 
\path (ps-res-t5) edge[-latex] (ps-res-po) ; 

\node (ppi) [left=0.3cm of ps-res-pi] {} ;
\path (ppi) edge[-latex] (ps-res-pi);

\node (ppo) [right=0.3cm of ps-res-po] {} ;
\path (ppo) edge[latex-] (ps-res-po);

\node [above=0.2cm of ps-res-cl2] {$\textbf{pc}(N)$ if $N$ is a tWF net} ;
\node [below=0.3cm of ps-res-cl2] {(a)} ;

\node (ts-res-cl2) [cloud, draw,cloud puffs=20,cloud puff arc=100, aspect=2,fill=white,minimum width=3cm,minimum height=2.5cm,fill=lightgray, right=4cm of ps-res-cl2] {\Large $N$};

\node (ts-res-t1) [place,above left=-0.5cm and -0.6cm of ts-res-cl2,fill=white] {};
\node (ts-res-t2) [place,below left=-0.5cm and -0.6cm of ts-res-cl2,fill=white] {};
\node (ts-res-t3) [place,above right=-0.4cm and -0.7cm of ts-res-cl2,fill=white] {};
\node (ts-res-t4) [place,right=-0.7cm of ts-res-cl2,fill=white] {};
\node (ts-res-t5) [place,below right=-0.4cm and -0.7cm of ts-res-cl2,fill=white] {};

\node (ts-res-ti) [transition,left=0.7cm of ts-res-cl2,fill=white] {$t_i$};
\node (ts-res-to) [transition,right=0.7cm of ts-res-cl2,fill=white] {$t_o$};

\path (ts-res-ti) edge[-latex] (ts-res-t1) ; 
\path (ts-res-ti) edge[-latex] (ts-res-t2) ; 
\path (ts-res-t3) edge[-latex] (ts-res-to) ; 
\path (ts-res-t4) edge[-latex] (ts-res-to) ; 
\path (ts-res-t5) edge[-latex] (ts-res-to) ; 

\node (pti) [left=0.3cm of ts-res-ti] {} ;
\path (pti) edge[-latex] (ts-res-ti);

\node (pto) [right=0.3cm of ts-res-to] {} ;
\path (pto) edge[latex-] (ts-res-to);

\node [above=0.2cm of ts-res-cl2] {$\textbf{tc}(N)$ if $N$ is a pWF net} ;
\node [below=0.3cm of ts-res-cl2] {(b)} ;

\end{tikzpicture}
}
\end{center}
\caption{\label{fig:place-compl}A place-completed tWF net and transition-completion pWF net (from \cite{DBLP:journals/is/SrokaH14})}
\end{figure}
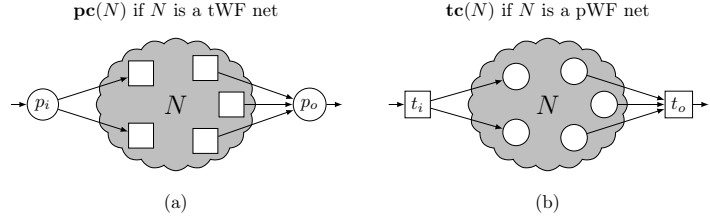

In this paper we discuss a particular kind of soundness, namely the
soundness that guarantees the reachability of a proper final state
\cite{DBLP:journals/is/SrokaH14}. We generalise this for the case
where (1) there can be more than one input place and (2) these contain
one or more tokens in the initial marking. We also provide a generalisation
of soundness for tWF nets, which intuitively states that after $k$
firings of all input transitions the computation can end in an empty
marking while firing each of the output transitions exactly $k$ times.

\begin{defn}
[$k$ and *-soundness] A pWF net $N=( P,T,F,I,O)$ is
said to be \emph{$k$-sound} if for each marking $m$ such that $k.I\stackrel{*}{\longrightarrow}m$
it holds that $m\stackrel{*}{\longrightarrow}k.O$. We call $N$ \emph{{*}-sound}
if it is $k$-sound for all $k\geq1$. We say that these properties
hold for tWF net $N$ if they hold for $\mathbf{pc}(N)$.
\end{defn}

By definition place-completion does not affect the {*}-soundness.
However, as we observed in \cite{DBLP:journals/is/SrokaH14,10.1109/ACSD.2011.26},
for transition-completion this is only true in one direction as every
pWF net $N$ is {*}-sound if $\mathbf{tc}(N)$ is {*}-sound but not
vice versa.

\section{Definition of AND-OR nets\label{sect:AND-OR-nets}}

In this section we introduce the AND-OR nets and show that they are
{*}-sound. For that we recall some definitions and a result from \cite{DBLP:journals/is/SrokaH14,10.1109/ACSD.2011.26}
where we defined a method for constructing complex nets by substituting
their nodes with other nets - places by pWF nets and transitions by
tWF nets.

\begin{defn}[Place substitution, Transition substitution]
Consider two \emph{disjoint} WF nets $N$ and $M$, i.e., if $N=(P,T,F,I,O)$
and $M=(P',T',F',I',O')$, then $(P\cup T)\cap(P'\cup T')=\emptyset$.

\emph{Place substitution:} If $p$ is a place in $N$ and $M$ is
a pWF net, then we define the result of substituting $p$ in $N$
with $M$, denoted as $N\otimes_{p}M$, as the net that is obtained
if in $N$ we remove $p$ and the edges in which it participates and
replace it with the net $M$ and edges such that $\bullet p'=\bullet p$
for each input place $p'\in I'$ of $M$ and $p'\bullet=p\bullet$
for each output place $p'\in O'$ of $M$. If $p\in I$ then $p$
is replaced in the set of input nodes of the resulting net with $I'$,
i.e., the input set of $N\otimes_{p}M$ is $\left(I\setminus\{p\}\right)\cup I'$,
and if $p\in O$ then $p$ is replaced in the set of output nodes
of the resulting net with $O'$, i.e., the output set of $N\otimes_{p}M$
is $\left(O\setminus\{p\}\right)\cup O'$. Otherwise, the input and
output sets of $N\otimes_{p}M$ are the same as the respective sets
for $N$.

\emph{Transition substitution:} Likewise, if $t$ is a transition
in $N$ and $M$ is a tWF net, then we define the result of substituting
$t$ in $N$ with $M$, denoted as $N\otimes_{t}M$, as the net that
is obtained if in $N$ we remove $t$ and the edges in which it participates
and replace it with the net $M$ and edges such that $\bullet t'=\bullet t$
for each input transition $t'\in I'$ of $M$ and $t'\bullet=t\bullet$
for each output transition $t'\in O'$ of $M$. If $t\in I$ then
$t$ is replaced in the set of input nodes of the resulting net with
$I'$, i.e., the input set of $N\otimes_{t}M$ is $\left(I\setminus\{t\}\right)\cup I'$,
and if $t\in O$ then $t$ is replaced in the set of output nodes
of the resulting net with $O'$, i.e., the output set of $N\otimes_{t}M$
is $\left(O\setminus\{t\}\right)\cup O'$. Otherwise, the input and
output sets of $N\otimes_{t}M$ are the same as the respective sets
for $N$. 
\end{defn}

The results of a place substitution and transition substitution are illustrated in Figure~\ref{fig:subst-illustr} (a) and (b), respectively. It is not hard to see that for all WF nets $N$ and $M$ and $n$ a node in $N$ such that $N\otimes_{n}M$ is defined, it is again a WF net. It also holds for all WF nets $A$, $B$ and $C$ that $(A \otimes_{n} B) \otimes_{m} C = A \otimes_{n} (B \otimes_{m} C)$, and $(A \otimes_{n} B)\otimes_{m}C = (A\otimes_{n}C) \otimes_{m} B$ if $n$ and $m$ are different nodes in $A$.

\begin{figure}[htb]
\begin{center}
\resizebox{\textwidth}{!}{%
\begin{tikzpicture}
    \tikzstyle{transition} = [rectangle,draw,minimum width=0.47cm, minimum height=0.47cm,fill=white]
    \tikzstyle{place} = [circle,draw,minimum width=0.5cm, minimum height=0.5cm,fill=white,inner sep=0.08cm]

\node (ps-res-cl1) [cloud, draw,cloud puffs=20,cloud puff arc=100, fill=mygray1, fill opacity=0.5,
  minimum width=6.5cm,minimum height=4cm] {};    
\node (ps-res-cl2) [cloud, draw,cloud puffs=20,cloud puff arc=-100, fill=white,
  minimum width=3.7cm,minimum height=2cm] at (ps-res-cl1.center) {};
\node (ps-res-cl3) [cloud, draw,cloud puffs=15,cloud puff arc=120, aspect=1.5,fill=mygray2,  fill opacity=0.5, text opacity=1,
  minimum width=2.3cm,minimum height=1.8cm] at (ps-res-cl2.center) {$M$};

\node (ps-res-p1) [place,above left=-0.5cm and -0.5cm of ps-res-cl3,fill=white] {};
\node (ps-res-p2) [place,below left=-0.5cm and -0.5cm of ps-res-cl3,fill=white] {};
\node (ps-res-p3) [place,above right=-0.5cm and -0.5cm of ps-res-cl3,fill=white] {};
\node (ps-res-p4) [place,below right=-0.5cm and -0.5cm of ps-res-cl3,fill=white] {};

\node (ps-res-t1) [transition,above left=-0.4cm and 0.5cm of ps-res-cl2,fill=white] {};
\node (ps-res-t2) [transition,below left=-0.4cm and 0.5cm of ps-res-cl2,fill=white] {};
\node (ps-res-t3) [transition,above right=0.1cm and 0.1cm of ps-res-cl2,fill=white] {};
\node (ps-res-t4) [transition,right=0.2cm of ps-res-cl2,fill=white] {};
\node (ps-res-t5) [transition,below right=0.1cm and 0.1cm of ps-res-cl2,fill=white] {};

\path (ps-res-t1) edge[-latex] (ps-res-p1) ; 
\path (ps-res-t1) edge[-latex] (ps-res-p2) ; 
\path (ps-res-t2) edge[-latex] (ps-res-p1) ; 
\path (ps-res-t2) edge[-latex] (ps-res-p2) ; 
\path (ps-res-p3) edge[-latex] (ps-res-t3) ; 
\path (ps-res-p3) edge[-latex] (ps-res-t4) ; 
\path (ps-res-p3) edge[-latex] (ps-res-t5) ; 
\path (ps-res-p4) edge[-latex] (ps-res-t3) ; 
\path (ps-res-p4) edge[-latex] (ps-res-t4) ; 
\path (ps-res-p4) edge[-latex] (ps-res-t5) ; 

\node [above=0cm of ps-res-cl1] {$N \otimes_{p} M$} ;

\node (ps-in-cl) [cloud, draw,cloud puffs=15,cloud puff arc=120, aspect=1.5,fill=mygray2,  fill opacity=0.5, text opacity=1,
  minimum width=2.3cm, minimum height=1.8cm] at (-3.5cm, 2.5cm) {$M$};
\node (ps-in-p1) [place,above left=-0.5cm and -0.5cm of ps-in-cl,fill=white] {};
\node (ps-in-p2) [place,below left=-0.5cm and -0.5cm of ps-in-cl,fill=white] {};
\node (ps-in-p3) [place,above right=-0.5cm and -0.5cm of ps-in-cl,fill=white] {};
\node (ps-in-p4) [place,below right=-0.5cm and -0.5cm of ps-in-cl,fill=white] {};

\foreach \n in {p1, p2}{
  \node (ps-pre-\n) [left=0.3cm of ps-in-\n] {};
  \path (ps-pre-\n) edge[-latex] (ps-in-\n) ; 
}

\foreach \n in {p3, p4}{
  \node (ps-post-\n) [right=0.3cm of ps-in-\n] {};
  \path (ps-post-\n) edge[latex-] (ps-in-\n) ; 
}

\node [below=0.6cm of ps-res-cl1] {(a)} ;

\node (ts-res-cl1) [cloud, draw,cloud puffs=20,cloud puff arc=100, fill=mygray1, fill opacity=0.5,
  minimum width=6.5cm,minimum height=4cm, right=2cm of ps-res-cl1] {};    
\node (ts-res-cl2) [cloud, draw,cloud puffs=20,cloud puff arc=-100, fill=white,
  minimum width=3.7cm, minimum height=2cm] 
  at (ts-res-cl1.center) {};
\node (ts-res-cl3) [cloud, draw,cloud puffs=15,cloud puff arc=120, fill=mygray2,  fill opacity=0.5, text opacity=1,
  minimum width=2.3cm,minimum height=1.8cm] 
  at (ts-res-cl2.center) {$M$};

\node (ts-res-p1) [transition,above left=-0.5cm and -0.5cm of ts-res-cl3,fill=white] {};
\node (ts-res-p2) [transition,below left=-0.5cm and -0.5cm of ts-res-cl3,fill=white] {};
\node (ts-res-p3) [transition,above right=-0.5cm and -0.5cm of ts-res-cl3,fill=white] {};
\node (ts-res-p4) [transition,below right=-0.5cm and -0.5cm of ts-res-cl3,fill=white] {};

\node (ts-res-t1) [place,above left=-0.4cm and 0.5cm of ts-res-cl2,fill=white] {};
\node (ts-res-t2) [place,below left=-0.4cm and 0.5cm of ts-res-cl2,fill=white] {};
\node (ts-res-t3) [place,above right=0.1cm and 0.1cm of ts-res-cl2,fill=white] {};
\node (ts-res-t4) [place,right=0.2cm of ts-res-cl2,fill=white] {};
\node (ts-res-t5) [place,below right=0.1cm and 0.1cm of ts-res-cl2,fill=white] {};

\path (ts-res-t1) edge[-latex] (ts-res-p1) ; 
\path (ts-res-t1) edge[-latex] (ts-res-p2) ; 
\path (ts-res-t2) edge[-latex] (ts-res-p1) ; 
\path (ts-res-t2) edge[-latex] (ts-res-p2) ; 
\path (ts-res-p3) edge[-latex] (ts-res-t3) ; 
\path (ts-res-p3) edge[-latex] (ts-res-t4) ; 
\path (ts-res-p3) edge[-latex] (ts-res-t5) ; 
\path (ts-res-p4) edge[-latex] (ts-res-t3) ; 
\path (ts-res-p4) edge[-latex] (ts-res-t4) ; 
\path (ts-res-p4) edge[-latex] (ts-res-t5) ; 

\node [above=0cm of ts-res-cl1] {$N \otimes_{t} M$} ;

\node (ts-in-cl) [cloud, draw,cloud puffs=15,cloud puff arc=120, aspect=1.5,fill=mygray2,  fill opacity=0.5, text opacity=1,
  minimum width=2.3cm, minimum height=1.8cm] at (12cm, 2.5cm) {$M$};
\node (ts-in-t1) [transition,above left=-0.5cm and -0.5cm of ts-in-cl,fill=white] {};
\node (ts-in-t2) [transition,below left=-0.5cm and -0.5cm of ts-in-cl,fill=white] {};
\node (ts-in-t3) [transition,above right=-0.5cm and -0.5cm of ts-in-cl,fill=white] {};
\node (ts-in-t4) [transition,below right=-0.5cm and -0.5cm of ts-in-cl,fill=white] {};

\foreach \n in {t1, t2}{
  \node (ts-pre-\n) [left=0.3cm of ts-in-\n] {};
  \path (ts-pre-\n) edge[-latex] (ts-in-\n) ; 
}

\foreach \n in {t3, t4}{
  \node (ts-post-\n) [right=0.3cm of ts-in-\n] {};
  \path (ts-post-\n) edge[latex-] (ts-in-\n) ; 
}

\node [below=0.6cm of ts-res-cl1] {(b)} ;

\node (ts-inp-cl1) [cloud, draw,cloud puffs=20,cloud puff arc=100, aspect=2,fill=mygray1,  fill opacity=0.5,
  minimum width=6.5cm,minimum height=4cm, above=1.5cm of ts-res-cl1] {};    
\node (ts-inp-cl2) [cloud, draw,cloud puffs=20,cloud puff arc=-100, aspect=2,fill=white,
  minimum width=3.7cm,minimum height=2cm] at (ts-inp-cl1.center) {};

\node (ts-inp-p) [transition,fill=white] at (ts-inp-cl1.center) {$t$};

\node (ts-inp-t1) [place,above left=-0.4cm and 0.5cm of ts-inp-cl2,fill=white] {};
\node (ts-inp-t2) [place,below left=-0.4cm and 0.5cm of ts-inp-cl2,fill=white] {};
\node (ts-inp-t3) [place,above right=0.1cm and 0.1cm of ts-inp-cl2,fill=white] {};
\node (ts-inp-t4) [place,right=0.2cm of ts-inp-cl2,fill=white] {};
\node (ts-inp-t5) [place,below right=0.1cm and 0.1cm of ts-inp-cl2,fill=white] {};

\path (ts-inp-t1) edge[-latex] (ts-inp-p) ; 
\path (ts-inp-t1) edge[-latex] (ts-inp-p) ; 
\path (ts-inp-t2) edge[-latex] (ts-inp-p) ; 
\path (ts-inp-t2) edge[-latex] (ts-inp-p) ; 
\path (ts-inp-p) edge[-latex] (ts-inp-t3) ; 
\path (ts-inp-p) edge[-latex] (ts-inp-t4) ; 
\path (ts-inp-p) edge[-latex] (ts-inp-t5) ; 
\path (ts-inp-p) edge[-latex] (ts-inp-t3) ; 
\path (ts-inp-p) edge[-latex] (ts-inp-t4) ; 
\path (ts-inp-p) edge[-latex] (ts-inp-t5) ; 

\node [above=0cm of ts-inp-cl1] {$N$} ;

\node (ps-inp-cl1) [cloud, draw,cloud puffs=20,cloud puff arc=100, aspect=2,fill=mygray1,  fill opacity=0.5,
  minimum width=6.5cm,minimum height=4cm, above=1.5cm of ps-res-cl1] {};    
\node (ps-inp-cl2) [cloud, draw,cloud puffs=20,cloud puff arc=-100, aspect=2,fill=white,
  minimum width=3.7cm,minimum height=2cm] at (ps-inp-cl1.center) {};

\node (ps-inp-p) [place,fill=white] at (ps-inp-cl1.center) {$p$};

\node (ps-inp-t1) [transition,above left=-0.4cm and 0.5cm of ps-inp-cl2,fill=white] {};
\node (ps-inp-t2) [transition,below left=-0.4cm and 0.5cm of ps-inp-cl2,fill=white] {};
\node (ps-inp-t3) [transition,above right=0.1cm and 0.1cm of ps-inp-cl2,fill=white] {};
\node (ps-inp-t4) [transition,right=0.2cm of ps-inp-cl2,fill=white] {};
\node (ps-inp-t5) [transition,below right=0.1cm and 0.1cm of ps-inp-cl2,fill=white] {};

\path (ps-inp-t1) edge[-latex] (ps-inp-p) ; 
\path (ps-inp-t1) edge[-latex] (ps-inp-p) ; 
\path (ps-inp-t2) edge[-latex] (ps-inp-p) ; 
\path (ps-inp-t2) edge[-latex] (ps-inp-p) ; 
\path (ps-inp-p) edge[-latex] (ps-inp-t3) ; 
\path (ps-inp-p) edge[-latex] (ps-inp-t4) ; 
\path (ps-inp-p) edge[-latex] (ps-inp-t5) ; 
\path (ps-inp-p) edge[-latex] (ps-inp-t3) ; 
\path (ps-inp-p) edge[-latex] (ps-inp-t4) ; 
\path (ps-inp-p) edge[-latex] (ps-inp-t5) ; 

\node [above=0cm of ps-inp-cl1] {$N$} ;

\end{tikzpicture}
}
\end{center}
\caption{\label{fig:subst-illustr}Illustration of place substitution and transition substitution (adapted from \cite{DBLP:journals/is/SrokaH14})}
\end{figure}
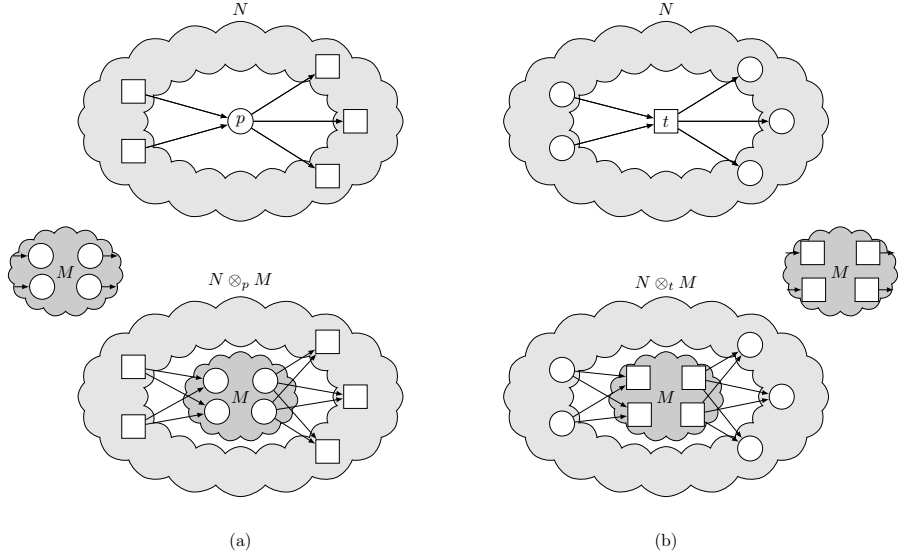

We will generate nets by starting from some basic classes of nets
and allowing substitutions of places with pWF nets and transitions
with tWF nets. 

\begin{defn}[Substitution closure]
Given a class $C$ of nets we define the substitution closure of
$C$, denoted as $\mathbf{S}(C)$, as the closure of $C$
under transition substitution and place substitution,
i.e., the smallest superclass $\mathbf{S}(C)$ of $C$ that satisfies the following two rules for every two disjoints nets $N$
and $M$ in $\mathbf{S}(C)$: 

(1) if $M$ is a pWF net and $p$ a place in $N$ then $N\otimes_{p}M$
is a net in $\mathbf{S}(C)$ and 

(2) if $M$ is a tWF net and $t$ a transition in $N$ then $N\otimes_{t}M$
is a net in $\mathbf{S}(C)$. 
\end{defn}

The motivation for using the concept of substitution closure to define interesting classes of nets is that it makes it straightforward to show for such a class that its nets have a certain property, such as a certain kind of soundness, by showing that (1) this property holds for the nets in the initial class and (2) this property is preserved by substitution. As was argued in earlier work~\cite{DBLP:journals/is/SrokaH14}, *-soundness is not preserved by substitution, but substitution does preserve a stronger property called \emph{substitution soundness}. The intuition behind this notion of soundness is that the execution of a net, when started with the same number of tokens in all input places, should always be able to finish properly, even if the execution is interfered with by removing at some step an identical number of tokens from all output places. More precisely, if we start the net if $k$ token in all its input places, and somewhere during the execution remove $k'$ tokens from all the output places, then the net should be able to finish with $k - k'$ tokens in its output places. 

\begin{defn}[substitution soundness] 
A pWF net $N=(P,T,F,I,O)$ is said to be \emph{substitution sound} if for all $0 \leq k' \leq k$ and every marking $m'$ of $N$ such that $k.I \stackrel{*}{\longrightarrow} m'+k'.O$ it holds that $m' \stackrel{*}{\longrightarrow} (k-k').O$. A tWF net is said to be substitution sound if its place-completion is substitution sound.
\end{defn}

The choice for the initial classes is made such that these are indeed substitution sound. Their intuition is based on that of acyclic marked graphs (T-nets) and state machines (S-nets),
as considered in \cite{DBLP:conf/apn/HeeSV03}. The basic idea of T-nets is that during executions no choices are made about who consumes a token and each transition will fire in the execution of a workflow. This is captured by a syntactic restriction that says that places have exactly one incoming edge and one outgoing edge. We will call this the \emph{AND-property} of a WF net.

\begin{defn}[AND property] 
A WF net $N=(P,T,F,I,O)$ is said to have the \emph{AND property} if for every place $p\in P$ it holds that (1) $p\in I\wedge|\bullet p|=0$
or $p\notin I\wedge|\bullet{p}|=1$ and (2) $p\in O\wedge|p\bullet|=0$ or $p\notin O\wedge|p\bullet|=1$. 
\end{defn}
Note that in this definition being an input node is considered equivalent to having an extra incoming edge, and being an output node equivalent to having an extra outgoing edge.

The basic idea of S-nets is that they represent a state machine, with the state represented by a single token that is in one of the places. This is captured by a restriction that says that transitions have exactly one incoming edge and one outgoing edge. We will call this the \emph{OR-property} of a net.

\begin{defn}[OR property] 
A WF net $N=(P,T,F,I,O)$ is said to have the \emph{OR property} if for every transition $t \in T$ it holds that (1) $t \in I\wedge|\bullet p|=0$
or $t \notin I \wedge |\bullet t| = 1$ and (2) $t \in O \wedge |t \bullet| = 0$ or $t \notin O \wedge |t \bullet| = 1$. 
\end{defn}

Although all WF nets with the OR property are substitution sound, it is unfortunately not the case that all WF nets with the AND property are substitution sound. Consider for example an AND net containing a transition $t$ with a loop containing a place $p$. Note that because of the AND property place $p$ can have no other edges except those of the loop. Therefore, the transition $t$ can never fire, since that would require a token in $p$, but such a token can only be generated by firing $t$. To remedy this, we will define AND nets as WF nets that not only satisfy the AND property but are also acyclic, which leads to the following definitions of AND and OR nets.

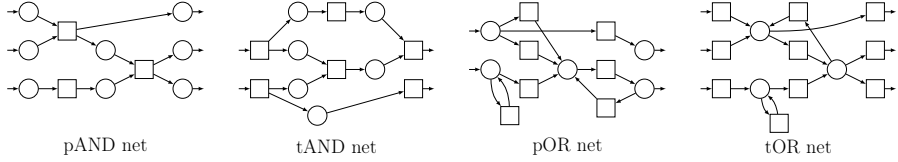
\begin{figure}[htb]
\begin{center}
\resizebox{\textwidth}{!}{%
\begin{tikzpicture}
    \tikzstyle{transition} = [rectangle,draw,minimum width=0.47cm, minimum height=0.47cm,fill=white]
    \tikzstyle{place} = [circle,draw,minimum width=0.5cm, minimum height=0.5cm,fill=white,inner sep=0.08cm]

% --- starting pAND net ---

% drawing transitions of pAND net
\foreach \num / \pos in {1/{(1, 0)}, 2/{(1, 1.5)}, 3/{(3, 0.5)}}{
  \node (pand-t\num) [transition] at \pos {};
}

% drawing places of pAND net
\foreach \num / \pos in {1/{(0, 0)}, 2/{(0, 1)}, 3/{(0, 2)}, 4/{(2, 0)}, 5/{(2, 1)}, 6/{(4, 0)}, 7/{(4, 1)}, 8/{(4, 2)}}{
  \node (pand-p\num) [place] at \pos {};
}

% marking input nodes of pAND net
\foreach \n in {p1, p2, p3}{
  \node (pand-pre-\n) [left=0.3cm of pand-\n] {};
  \path (pand-pre-\n) edge[-latex] (pand-\n) ; 
}

% marking output nodes of pAND net
\foreach \n in {p6, p7, p8}{
  \node (pand-post-\n) [right=0.3cm of pand-\n] {};
  \path (pand-post-\n) edge[latex-] (pand-\n) ; 
}

% normal edges pAND net
\foreach \x / \y in {
  p3/t2, p2/t2, p1/t1, t2/p8, t2/p5, t1/p4, p5/t3,
  p4/t3, t3/p7, t3/p6} {
  \path (pand-\x) edge[-latex] (pand-\y);
}

% label below pAND net
\node at (2, -1.5) {\Large pAND net} ;

% --- starting tAND net ---

% drawing transitions of tAND net
\foreach \num / \pos in {1/{(6, 0)},2/{(6, 1)},3/{(8, 0.5)},4/{(8, 2)},5/{(10, 0)},6/{(10, 1)}}{
  \node (tand-t\num) [transition] at \pos {};
}

% drawing places of tAND net
\foreach \num / \pos in {1/{(7, 0)}, 2/{(7, 1)}, 3/{(7, 2)}, 4/{(7.5, -0.7)}, 5/{(9, 0.5)}, 6/{(9, 2)}}{
  \node (tand-p\num) [place] at \pos {};
}

% marking input nodes of tAND net
\foreach \n in {t1, t2}{
  \node (tand-pre-\n) [left=0.3cm of tand-\n] {};
  \path (tand-pre-\n) edge[-latex] (tand-\n) ; 
}

% marking output nodes of tAND net
\foreach \n in {t5, t6}{
  \node (tand-post-\n) [right=0.3cm of tand-\n] {};
  \path (tand-post-\n) edge[latex-] (tand-\n) ; 
}

% normal edges tAND net
\foreach \x / \y in {t1/p4, t2/p2, t2/p3, t1/p1, p1/t3, p3/t4, p2/t3, p4/t5, t4/p6, t3/p5, p6/t6, p5/t6}{
  \path (tand-\x) edge[-latex] (tand-\y);
}

% label below tAND net
\node at (8, -1.5) {\Large {tAND} net} ;

% --- starting pOR net ---

% drawing transitions of pOR net
\foreach \num / \pos in {1/{(12.5, -0.7)}, 2/{(13, 0)}, 3/{(13, 1)}, 4/{(13, 2)}, 5/{(15, -0.5)}, 6/{(15, 0.5)}, 7/{(15, 1.5)}}{
  \node (por-t\num) [transition] at \pos {};
}

% drawing places of pOR net
\foreach \num / \pos in {1/{(12, 0.5)}, 2/{(12, 1.5)}, 3/{(14, 0.5)}, 4/{(16, 0)}, 5/{(16, 1)}}{
  \node (por-p\num) [place] at \pos {};
}

% marking input nodes of pOR net
\foreach \n in {p1, p2}{
  \node (por-pre-\n) [left=0.3cm of por-\n] {};
  \path (por-pre-\n) edge[-latex] (por-\n) ; 
}

% marking output nodes of pOR net
\foreach \n in {p5, p4}{
  \node (por-post-\n) [right=0.3cm of por-\n] {};
  \path (por-post-\n) edge[latex-] (por-\n) ; 
}

% special edges pOR net
\path (por-p1) edge[-latex, bend right=20] (por-t1);
\path (por-t1) edge[-latex, bend right=20] (por-p1);

% normal edges pOR net
\foreach \x / \y in {p1/t2, p2/t3, p2/t7, t4/p3, t3/p3, t2/p3, p2/t4, t7/p5, t6/p4, t5/p3, p3/t6, p4/t5} {
  \path (por-\x) edge[-latex] (por-\y);
}

% label below pOR net
\node at (14, -1.5) {\Large {pOR} net} ;

% --- starting tOR net ---

% drawing transitions of tOR net
\foreach \num / \pos in {1/{(18,2)}, 2/{(18,1)}, 3/{(18,0)}, 4/{(20,2)}, 5/{(20,1)}, 6/{(20,0)}, 7/{(22,2)}, 8/{(22,1)}, 9/{(22,0)}, 10/{(19.5, -0.9)}}{
  \node (tor-t\num) [transition] at \pos {};
}

% drawing places of tOR net
\foreach \num / \pos in {1/{(19,1.5)}, 2/{(19,0)}, 3/{(21,0.5)}}{
  \node (tor-p\num) [place] at \pos {};
}

% marking input nodes of tOR net
\foreach \n in {t1, t2, t3}{
  \node (tor-pre-\n) [left=0.3cm of tor-\n] {};
  \path (tor-pre-\n) edge[-latex] (tor-\n) ; 
}

% marking output nodes of tOR net
\foreach \n in {t7, t8, t9}{
  \node (tor-post-\n) [right=0.3cm of tor-\n] {};
  \path (tor-post-\n) edge[latex-] (tor-\n) ; 
}

% drawing edges for tOR net
\foreach \x / \y in {t1/p1, t2/p1, t3/p2, p1/t5, p2/t6, t4/p1, t5/p3, t6/p3, p3/t4, p3/t8, p3/t9} {
  \path (tor-\x) edge[-latex] (tor-\y);
}

% special edges for tOR net
\path (tor-p1) edge[-latex, bend right=10] (tor-t7);
\path (tor-p2) edge[-latex, bend right=20] (tor-t10);
\path (tor-t10) edge[-latex, bend right=20] (tor-p2);

% label below tOR net
\node at (20, -1.5) {\Large tOR net} ;

\end{tikzpicture}
}
\end{center}
\caption{\label{fig:And-or-net}Examples of a pAND, tAND, pOR and tOR nets
(adapted from \cite{DBLP:journals/is/SrokaH14})}
\end{figure}

\begin{defn}[AND net]
An \emph{AND net} is an acyclic WF net that satisfies the AND property. An AND net that is a pWF net is
called a pAND net, and if it is a tWF net it is called a tAND net. 
\end{defn}

\begin{defn}[OR net]
An \emph{OR net} is a WF net that satisfies the OR property.
An OR net that is a pWF net is called a pOR net, and if it is a tWF
net it is called a tOR net. 
\end{defn}

%Note that OR nets can contain cycles where AND nets by definition
%cannot, but otherwise they are each others dual. Also note that for
%the requirements over the edges, being an input node counts as having
%an input edge, and being an output node counts as having an output
%edge.

For some examples of AND and OR nets see Figure~\ref{fig:And-or-net}. 
We will use names to identify the classes of nets we have just defined. The
classes of pOR nets, pAND nets, tOR and tAND nets will be denoted as 
$\mathbf{pOR}$, $\mathbf{pAND}$, $\mathbf{tOR}$ and $\mathbf{tAND}$.
In addition we will prefix the name with $\mathbf{11}$ if the class contains
only nets with one input node and one output node. So $\mathbf{11pOR}$ is the class
of one-input one-output pOR nets, and $\mathbf{11tAND}$ is the class
of one-input one-output tAND nets.

Since our aim is to generate substitution sound nets, which as we will show are also {*}-sound, we will restrict the tAND nets 
to one-input one-output nets, so the class $\mathbf{11tAND}$. The reason is that
although all pAND nets are substitution sound, it is not true that all tAND nets are substitution sound, as
is illustrated by the tAND net in Figure~\ref{fig:And-or-net}. Recall that the soundness
of a tWF net is defined as the soundness of its place-completion which adds a single
place before the input transitions and a single place after the output transitions.
So if we put a single token in the first place in this place-completion, then only one of the first
transitions can fire, after which the net can no longer correctly finish.  A similar problem occurs
in the pOR net in Figure~\ref{fig:And-or-net}. If it starts with a token in each of its input places, 
then it cannot finish correctly if it moves the token of the
upper input place to the place in the middle of the net. Both types of problems are solved if we
only consider one-input one-output versions of tAND and pOR nets, which are all
substitution sound nets. Hence we only consider the classes illustrated in Figure~\ref{fig:And-or-net11}, 
which we will refer to as the  \emph{basic
AND-OR classes}, and the class of nets that is obtained by combining them by substitution 
we will call the class of \emph{AND-OR nets}.

\begin{figure}[htb]
\begin{center}
\resizebox{\textwidth}{!}{%
\begin{tikzpicture}
    \tikzstyle{transition} = [rectangle,draw,minimum width=0.47cm, minimum height=0.47cm,fill=white]
    \tikzstyle{place} = [circle,draw,minimum width=0.5cm, minimum height=0.5cm,fill=white,inner sep=0.08cm]

% --- starting pAND net ---

% drawing transitions of pAND net
\foreach \num / \pos in {1/{(1, 0)}, 2/{(1, 1.5)}, 3/{(3, 0.5)}}{
  \node (pand-t\num) [transition] at \pos {};
}

% drawing places of pAND net
\foreach \num / \pos in {1/{(0, 0)}, 2/{(0, 1)}, 3/{(0, 2)}, 4/{(2, 0)}, 5/{(2, 1)}, 6/{(4, 0)}, 7/{(4, 1)}, 8/{(4, 2)}}{
  \node (pand-p\num) [place] at \pos {};	
}

% marking input nodes of pAND net
\foreach \n in {p1, p2, p3}{
  \node (pand-pre-\n) [left=0.3cm of pand-\n] {};
  \path (pand-pre-\n) edge[-latex] (pand-\n) ; 
}

% marking output nodes of pAND net
\foreach \n in {p6, p7, p8}{
  \node (pand-post-\n) [right=0.3cm of pand-\n] {};
  \path (pand-post-\n) edge[latex-] (pand-\n) ; 
}

% normal edges pAND net
\foreach \x / \y in {
  p3/t2, p2/t2, p1/t1, t2/p8, t2/p5, t1/p4, p5/t3,
  p4/t3, t3/p7, t3/p6} {
  \path (pand-\x) edge[-latex] (pand-\y);
}

% label below pAND net
\node at (2, -1.5) {\Large pAND net} ;

% --- starting tAND net ---

% drawing transitions of tAND net
\foreach \num / \pos in {2/{(6, 1)},3/{(8, 0.5)},4/{(8, 2)},6/{(10, 1)}}{
  \node (tand-t\num) [transition] at \pos {};
}

% drawing places of tAND net
\foreach \num / \pos in {1/{(7, 0)}, 2/{(7, 1)}, 3/{(7, 2)}, 5/{(9, 0.5)}, 6/{(9, 2)}}{
  \node (tand-p\num) [place] at \pos {};
}

% marking input nodes of tAND net
\foreach \n in {t2}{
  \node (tand-pre-\n) [left=0.3cm of tand-\n] {};
  \path (tand-pre-\n) edge[-latex] (tand-\n) ; 
}

% marking output nodes of tAND net
\foreach \n in {t6}{
  \node (tand-post-\n) [right=0.3cm of tand-\n] {};
  \path (tand-post-\n) edge[latex-] (tand-\n) ; 
}

% normal edges tAND net
\foreach \x / \y in {t2/p2, t2/p3, t2/p1, p1/t3, p3/t4, p2/t3, t4/p6, t3/p5, p6/t6, p5/t6}{
  \path (tand-\x) edge[-latex] (tand-\y);
}

% label below tAND net
\node at (8, -1.5) {\Large 11tAND net} ;

% --- starting pOR net ---

% drawing transitions of pOR net
\foreach \num / \pos in {1/{(12.5, -0.7)}, 2/{(13, 0)}, 3/{(13, 1)}, 4/{(14, 2)}, 5/{(15, 0)}, 6/{(15, 1)}}{
  \node (por-t\num) [transition] at \pos {};
}

% drawing places of pOR net
\foreach \num / \pos in {1/{(12, 0.5)}, 3/{(14, 0.5)}, 4/{(16, 0.5)}}{
  \node (por-p\num) [place] at \pos {};
}

% marking input nodes of pOR net
\foreach \n in {p1}{
  \node (por-pre-\n) [left=0.3cm of por-\n] {};
  \path (por-pre-\n) edge[-latex] (por-\n) ; 
}

% marking output nodes of pOR net
\foreach \n in {p4}{
  \node (por-post-\n) [right=0.3cm of por-\n] {};
  \path (por-post-\n) edge[latex-] (por-\n) ; 
}

% special edges pOR net
\path (por-p1) edge[-latex, bend right=20] (por-t1);
\path (por-t1) edge[-latex, bend right=20] (por-p1);

% normal edges pOR net
\foreach \x / \y in {p1/t2, p1/t3, t3/p3, t2/p3, t6/p4, t5/p3, p3/t6, p4/t5} {
  \path (por-\x) edge[-latex] (por-\y);
}

% special edges for pOR net
\path (por-t4) edge[-latex, bend right=20] (por-p1);
\path (por-p4) edge[-latex, bend right=20] (por-t4);

% label below pOR net
\node at (14, -1.5) {\Large 11pOR net} ;

% --- starting tOR net ---

% drawing transitions of tOR net
\foreach \num / \pos in {1/{(18,2)}, 2/{(18,1)}, 3/{(18,0)}, 4/{(20,2)}, 5/{(20,1)}, 6/{(20,0)}, 7/{(22,2)}, 8/{(22,1)}, 9/{(22,0)}, 10/{(19.5, -0.9)}}{
  \node (tor-t\num) [transition] at \pos {};
}

% drawing places of tOR net
\foreach \num / \pos in {1/{(19,1.5)}, 2/{(19,0)}, 3/{(21,0.5)}}{
  \node (tor-p\num) [place] at \pos {};
}

% marking input nodes of tOR net
\foreach \n in {t1, t2, t3}{
  \node (tor-pre-\n) [left=0.3cm of tor-\n] {};
  \path (tor-pre-\n) edge[-latex] (tor-\n) ; 
}

% marking output nodes of tOR net
\foreach \n in {t7, t8, t9}{
  \node (tor-post-\n) [right=0.3cm of tor-\n] {};
  \path (tor-post-\n) edge[latex-] (tor-\n) ; 
}

% drawing edges for tOR net
\foreach \x / \y in {t1/p1, t2/p1, t3/p2, p1/t5, p2/t6, t4/p1, t5/p3, t6/p3, p3/t4, p3/t8, p3/t9} {
  \path (tor-\x) edge[-latex] (tor-\y);
}

% special edges for tOR net
\path (tor-p1) edge[-latex, bend right=10] (tor-t7);
\path (tor-p2) edge[-latex, bend right=20] (tor-t10);
\path (tor-t10) edge[-latex, bend right=20] (tor-p2);

% label below tOR net
\node at (20, -1.5) {\Large tOR net} ;

\end{tikzpicture}
}

\end{center}
\caption{\label{fig:And-or-net11}Examples from classes \textbf{pAND},
\textbf{11tAND}, \textbf{11pOR} and \textbf{tOR} (from \cite{DBLP:journals/is/SrokaH14})}
\end{figure}
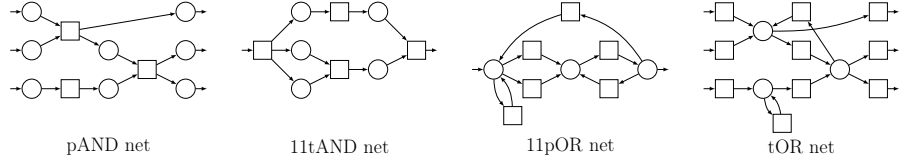

\begin{defn}[AND-OR net]
The class $\mathbf{S}(\mathbf{pAND}\cup\mathbf{11tAND}\cup\mathbf{11pOR}\cup\mathbf{tOR})$
is called the class of AND-OR nets. 
\end{defn}

An example of the generation of an AND-OR net is shown in Figure~\ref{fig:And-or-net-example},
with on the left the hierarchical decomposition and on the right the
resulting net.

\begin{figure}[htb]
\begin{center}
\resizebox{\textwidth}{!}{%
\begin{tikzpicture}
    \tikzstyle{transition} = [rectangle,draw,minimum width=0.47cm, minimum height=0.47cm,fill=white]
    \tikzstyle{place} = [circle,draw,minimum width=0.5cm, minimum height=0.5cm,fill=white,inner sep=0.08cm]

% --- first figure ---

% draw boxes
\node (b1) [cloud, draw, thick, fill=lightgray, fill opacity=0.5, minimum height=3cm, minimum width=4cm, cloud puffs=20,cloud puff arc=110] at (2.5, 3) {};
\node (b2) [cloud, draw, thick, fill=lightgray, fill opacity=0.5, minimum height=3.5cm, minimum width=7cm, cloud puffs=30,cloud puff arc=110] at (9, 1.5) {};
\node (b3) [cloud, draw, thick, fill=white, minimum height=2.7cm, minimum width=3.8cm, cloud puffs=20,cloud puff arc=110] at (9, 1.5) {};

% draw places first figure
\foreach \num / \pos in {1/{(0,0)}, 2/{(0,2)}, 3/{(0,4)}, 4/{(2.5,4)}, 5/{(2.5, 2)}, 6/{(5, 2)}, 7/{(8, 2)}, 8/{(8, 1)}, 9/{(10, 2)}, 10/{(10, 1)}, 11/{(13,0)}, 12/{(13,2)}, 13/{(13,4)}}{
  \node (p\num) [place, fill=white] at \pos {};
}

% draw transitions first figure
\foreach \num / \pos in {1/{(1.5, 3.5)}, 2/{(1.5, 2.5)}, 3/{(2.5, 3)}, 4/{(3.5, 3.5)}, 5/{(3.5, 2.5)}, 6/{(6.5, 2)}, 7/{(6.5, 1)}, 8/{(9, 2)}, 9/{(9, 1)}, 10/{(11.3, 0.8)}, 11/{(11.5, 2)}}{
  \node (t\num) [transition, fill=white] at \pos {};
}

% marking input nodes first figure
\foreach \n in {p1, p2, p3, t1, t2, t6, t7, p7, p8}{
  \node (pre-\n) [left=0.3cm of \n] {};
  \path (pre-\n) edge[-latex] (\n) ; 
}

% marking output nodes first figure
\foreach \n in {t4, t5, p9, p10, t11, p13, p12, p11}{
  \node (post-\n) [right=0.3cm of \n] {};
  \path (post-\n) edge[latex-] (\n) ; 
}

% straight edges first figure
\foreach \x / \y in {p3/b1, p2/b1, t1/p4, p4/t3, t3/p5, t2/p5, p4/t4, p5/t5, b1/p6, p1/b2, p6/b2, t6/b3, t7/b3, p7/t8, p8/t9, t8/p9, t9/p10, b3/t11, b2/p12, b2/p11, b1/p13}{
  \path (\x) edge[-latex] (\y);
}

% special edges in first figure
\path (t10.north) edge[latex-, bend right=22] (b3.east);
\path (b3) edge[latex-, bend right=22] (t10.south);

% --- second full figure ---

% draw places second figure
\foreach \num / \pos in {1/{(16,0)}, 2/{(16,2)}, 3/{(16,4)}, 4/{(18.5, 4)}, 5/{(18.5, 2)}, 6/{(21, 2)}, 7/{(24, 2)}, 8/{(24, 1)}, 9/{(26, 2)}, 10/{(26, 1)}, 11/{(29,0)}, 12/{(29,2)}, 13/{(29,4)}}{
  \node (fp\num) [place, fill=white] at \pos {};
}

% draw transitions second figure
\foreach \num / \pos in {1/{(17.5, 3.5)}, 2/{(17.5, 2.5)}, 3/{(18.5, 3)}, 4/{(19.5, 3.5)}, 5/{(19.5, 2.5)}, 6/{(22.5, 2)}, 7/{(22.5, 1)}, 8/{(25, 2)}, 9/{(25, 1)}, 10/{(27.3, 0.8)}, 11/{(27.5, 2)}}{
  \node (ft\num) [transition, fill=white] at \pos {};
}

% marking input nodes second figure
\foreach \n in {p1, p2, p3}{
  \node (fpre-\n) [left=0.3cm of f\n] {};
  \path (fpre-\n) edge[-latex] (f\n) ; 
}

% marking output nodes second figure
\foreach \n in {p13, p12, p11}{
  \node (fpost-\n) [right=0.3cm of f\n] {};
  \path (fpost-\n) edge[latex-] (f\n) ; 
}

% straight edges second figure
\foreach \x / \y in {p3/t1, p3/t2, p2/t1, p2/t2, t1/p4, p4/t3, t3/p5, t2/p5, p4/t4, p5/t5, t4/p6, t5/p6, p1/t6, p1/t7, p6/t6, p6/t7, t6/p7, t6/p8, t7/p7, t7/p8, p7/t8, p8/t9, t8/p9, t9/p10, p9/t11, p10/t11, t11/p12, t11/p11, t4/p13, t5/p13}{
  \path (f\x) edge[-latex] (f\y);
}

% special edges in second figure
\path (ft10) edge[-latex, bend left=56] (fp7);
\path (ft10) edge[-latex, bend left=50] (fp8);
\path (fp9) edge[-latex, bend left] (ft10);
\path (fp10) edge[-latex, bend left] (ft10);

\end{tikzpicture}
}

\end{center}
\caption{\label{fig:And-or-net-example}An example of the generation of an
AND-OR net (adapted from \cite{DBLP:journals/is/SrokaH14})}
\end{figure}
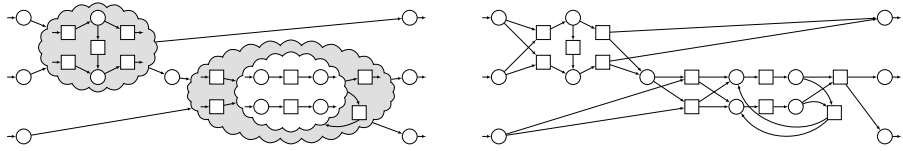

% JH: the following was removed since it does not add much and seems a bit too obvious to present as a theorem
% JS: I've put it back (but added "easily") as *-soundness of AND-OR nets was never proven in the previous papers
% JH: And I've removed it again, and made it a footnote in the first place we talk about this.

%Finally we recall from \cite{DBLP:journals/is/SrokaH14} Corollary 30:
%
%\begin{prop}
%\label{prop:AND_OR_are_subsound} Every AND-OR net is substitution sound.
%\end{prop}
%
%Based on this result it can be easily shown that all AND-OR are {*}-sound.
%\begin{thm}
%\label{thm:AND_OR_are_*-sound}Every AND-OR net is {*}-sound.\end{thm}
%\begin{proof}
%From Proposition~\ref{prop:AND_OR_are_subsound} we know that
%all AND-OR nets are substitution sound. For pWF nets substitution
%soundness means that if for all $k \geq k' \geq 0$ and every marking
%$m'$ it holds that if $k.I\stackrel{*}{\longrightarrow}(m'+k'.O)$
%then $m'\stackrel{*}{\longrightarrow}(k-k').O$ and for tWF nets this
%property is generalised such that a tWF net $N$ is sub-sound if $pc(N)$ is.
%Now it is enough to observe that for $k'=0$ this gives {*}-soundness
%by definition.
%\end{proof}

\section{The AND-OR reduction \label{sect:AND-OR-reduction}}

We now proceed with presenting the AND-OR net verification procedure.
Informally, the procedure can be described as reversing the generation
process. This means that we try to find subnets that might have been
generated by a substitution and reverse the substitution. This process
is repeated until we can find no more such subnets, and if then the
resulting net is a single node, the original net is concluded
to be an AND-OR net.

\begin{defn}[Subnet]
A subnet in workflow net $M=( P_{M},T_{M},F_{M},I_{M},O_{M})$ is identified by a non-empty set of nodes $S \subseteq P_M \cup T_M$. With $S$ we associate a net $M[S] = ( P_{S},T_{S},F_{S},I_{S},O_{S})$ that is the restriction of $M$ to the nodes in $S$, i.e.,
\begin{itemize}
  \item $P_S = P_M \cap S$,
  \item $T_S = T_M \cap S$,
  \item $F_S = F_M \cap (S \times S)$,
  \item $I_S = (I_M \cap S) \cup \{ n_2 \mid (n_1, n_2) \in F_M, n_1 \in ((P_M \cup T_M) \setminus S), n_2 \in S) \}$ and
  \item $O_S = (O_M \cap S) \cup  \{ n_1 \mid (n_1, n_2) \in F_M, n_1 \in S, n_2 \in ((P_M \cup T_M) \setminus S) \}$.
\end{itemize}
A subnet consisting of exactly one node will be called a \emph{trivial subnet}.
\end{defn}

Note that the input nodes of $M[S]$ are not just the nodes in $S$ that are input nodes of $M$ but also contains the nodes in $S$ that have in $M$ incoming edges from outside $S$. Analogously the output nodes are the nodes in $S$ that are output nodes of $M$ or have in $M$ an outgoing edge to a node outside $S$.

Not every subnet of a WF net is again itself a WF net, but it will always be an I/O net.

\begin{lem}[subnets are I/O nets] \label{lem:subnet-is-IOnet}
Let $S$ be a non-empty subset of nodes of a WF net $M$, then $M[S]$ is an I/O net.
\end{lem}

\begin{proof}
Let $M=( P_{M},T_{M},F_{M},I_{M},O_{M})$. For $M[S]=( P_{S},T_{S},F_{S},I_{S}, O_{S})$ to be an I/O net the following must be shown. First the defining properties of a Petri net: (i) $P_{S} \cap T_{S} = \emptyset$, (ii) $F_{S} \subseteq (T_{S} \times P_{S} ) \cup (P_{S} \times T_{S})$. Furthermore the additional defining properties of an I/O net: (iii) $I_{S} \neq \emptyset$ and $O_{S} \neq \emptyset$. We can show these as follows:
\begin{enumerate}[(i)]

  \item Since $P_M \cap T_M = \emptyset$, and by definition $P_S = P_M \cap N$ and $T_S = T_M \cap N$, it follows that $P_S \cap T_S = \emptyset$.

  \item Since $F_M \subseteq (P_M \times T_M) \cup (T_M \times P_M)$, it follows that 
  \begin{align*} 
     F_S &=  F_M \cap (N \times N) \\
             &\subseteq ((P_M \times T_M) \cup (T_M \times P_M)) \cap (N \times N) \\ 
             &=  ((P_M \times T_M) \cap (N \times N)) \cup ((T_M \times P_M) \cap (N \times N)) \\
             &=  ((P_M \cap N) \times (T_M \cap N)) \cup ((T_M \cap N) \times (P_M\cap N)) \\
             &= (P_S \times T_S) \cup (T_S \cup P_S)             
\end{align*}

  \item We show that $I_{S} \neq \emptyset$. In the case where $S \cap I_M \neq \emptyset$, this is trivially satisfied. In the case where $S \cap I_M = \emptyset$ there will be a node in $S$ that has in $M$ an incoming edge from outside $S$, since in $M$ all nodes must be reachable from an input node of $M$ and $S$ is non-empty. This node will hence be an input node of $M[S]$. In a similar fashion it can be shown that  $O_{S} \neq \emptyset$.

\end{enumerate}
\end{proof}

In addition, it can also be shown that every subnet is well-connected.

\begin{lem}[well-connectedness of subnets] \label{lem:subnet-is-well-connected}
Let $S$ be a non-empty subset of nodes of a WF net $M$, then $M[S]$ is well-connected.
\end{lem}

\begin{proof}
We first show that every node in $S$ has in $M[S]$ a path to an output node of $M[S]$. Since $M$ is well-connected, every node in $S$ will have a path to an output node of $M$. There either (i) is a node in that path that is not in $S$ or (ii) all nodes of the path are in $S$: 
\begin{enumerate}[(i)]

  \item Consider the longest prefix of the path that contains only nodes in $S$. The final node will be an output node of $M[S]$, since it is followed by a node not in $S$. Moreover, this path is present in $M[S]$ since it only contains nodes from $S$.
  
  \item The final node of the path is an output node of $M$, and therefore also an output node of $M[S]$.

\end{enumerate}

In a similar fashion it can be shown that to every node in $S$ there is in $M[S]$ a path from an input node of $M[S]$, except that the edges and the paths are reversed and the set of output nodes is replaced with the set of input nodes.
\end{proof}

Since the definition of a WF net is an I/O net that is well-connected and I/O consistent, it follows immediately from the previous two lemmas,  Lemma~\ref{lem:subnet-is-IOnet} and Lemma~\ref{lem:subnet-is-well-connected},  that a subnet of WF net is a WF net iff it is I/O consistent.

\begin{thm} \label{thm:subnet-wf-char}
Let $S$ be a non-empty set of nodes of a WF net $M$, then $M[S]$ is a WF net iff it is I/O consistent.
\end{thm}

We now proceed with defining the notion of \emph{contraction}, i.e., a contraction of a subnet $S$ of a WF net $M$ into a single new node $n$ such that (1) all edges that departed from outside $S$ and arrived in an input node of $S$ become edges arriving in $n$, (2) all edges that arrived in a node outside $S$ and left from an output node of $S$ become edges leaving from $n$. Moreover, if $S$ contained an input node of $M$ then $n$ becomes also an input node, and if $S$ contained an output node of $M$ then $n$ becomes also an output node. Formally, the definition is as follows.

\begin{defn}[Contraction]
Given a WF net $M=( P_{M},T_{M},F_{M},I_{M},O_{M})$
and a subnet $S$ such that $M[S] = ( P_{S},T_{S},F_{S},I_{S},O_{S})$ is a pWF net  (tWF net), we define the result of contracting $S$
into a new place (transition) node $n$ as $M'=( P_{M'},T_{M'},F_{M'},I_{M'},O_{M'})$
where
\begin{itemize}
  \item  $P_{M'}=(P_{M}\setminus S)\cup\{n\}$, ($P_{M'}=P_{M}\setminus S$),
  \item $T_{M'}  =T_{M}\setminus S$, ($T_{M'}=(T_{M}\setminus S)\cup\{n\}$),
  \item \begin{tabbing} 
$F_{M'} = $ \= $\{ (n_{1},n_{2}) \mid (n_{1},n_{2}) \in F_{M}, n_{1} \not\in S,n_{2}\not\in S \} \ \cup$ \\ 
  \>  $\{ (n_{1},n) \mid (n_{1}, n_{2}) \in F_{M}, n_1 \not\in S, n_{2}\in I_S \} \ \cup$ \\
  \> $\{ (n,n_{2}) \mid (n_{1}, n_{2}) \in F_{M}, n_{1}\in O_S, n_2 \not\in S \}$,
\end{tabbing}
  \item $I_{M'} = 
    \begin{cases}
  I_{M} & \text{ if } S \cap I_{M} = \emptyset \\      
  (I_{M}\setminus S)\cup\{n\} & \text{ otherwise, }   
    \end{cases}$
  \item $O_{M'} =
    \begin{cases}
      O_{M} & \text{ if } S \cap O_{M} = \emptyset \\
     (O_{M}\setminus S)\cup\{n\} & \text{ otherwise.} 
    \end{cases}$
\end{itemize}
\end{defn}

An important property of the contraction is that it preserves paths, as is witnessed by the following lemma.

\begin{lem}[preservation of paths by contraction]
\label{lem:contraction_keeps_paths}
Let $M'$ be the result of contracting in a WF net $M$ a subnet $S$ into $n$. Then, for every path in $M$ there is a similar path in $M'$ where all maximal substrings of nodes in $S$ are replaced with $n$.
\end{lem}

\begin{proof}
Clearly all edges between nodes outside $S$ are not affected by the contraction. What remains to consider are the following possibilities: (i) an edge from outside $S$ to inside $S$, (ii) an edge from inside $S$ to outside $S$ and (iii) an edge from inside $S$ to inside $S$.
\begin{enumerate}[(i)]

  \item If there is in $M$ an edge $(n_1, n_2)$ from outside $S$ to inside $S$, then by the definition of $M[S]$ the node $n_2$ is an input node of $M[S]$, and so by the definition of contraction there is in $M'$ an edge from $n_1$ to $n$.
  
  \item By a similar argument it holds that if there is in $M$ an edge $(n_1, n_2)$ from inside $S$ to outside $S$, then there is in $M'$ an edge from $n$ to $n_2$.
  
  \item For every edge $(n_1, n_2)$ from inside $S$ to inside $S$ there is the corresponding trivial one-node path $(n)$ in $M'$. 

\end{enumerate}
It follows that for each edge in $M$ there is a corresponding path in $M'$ such that if two edges in $M$ have the same begin node or end node, then the corresponding paths also have the same begin node or end node, respectively. Therefore, we can concatenate the corresponding paths in the order of the original edges in the path in $M$, and obtain a path in $M'$.
\end{proof}

With the help of this lemma it can be shown that the result of a contraction of a subnet in a pWF net (tWF net) is always a pWF net (tWF net).

\begin{thm}[correctness of the result of contraction] \label{thm:contraction-correctness}
If $M$ is a pWF net (tWF net), then the result $M'$ of contracting a subnet $S$ of $M$ into a node $n$ is again a a pWF net (tWF net).
\end{thm}

\begin{proof}
In order for $M'=( P_{M'},T_{M'},F_{M'},I_{M'}, O_{M'})$ to be a WF net the following must be shown. First, the defining properties of a Petri net: (i) $P_{M'} \cap T_{M'} = \emptyset$, (ii) $F_{M'} \subseteq (T_{M'} \times P_{M'} ) \cup (P_{M'} \times T_{M'})$. Furthermore, the defining properties of a WF net: (iii) $I_{M'} \neq \emptyset$ and $O_{M'} \neq \emptyset$, (iv) $I_{M'} \cup O_{M'} \subseteq P_{M'}$ or $I_{M'} \cup O_{M'} \subseteq T_{M'}$, (v) for every node in $P_{M'} \cup T_{M'}$ there is a path to that node from a node in $I_{M'}$, and (vi) for every node in $P_{M'} \cup T_{M'}$ there is a path from that node to a node in $O_{M'}$. Finally, it can concluded that (vii) the I/O type of $M$ is the same as that of $M'$.
\begin{enumerate}[(i)]

  \item It holds for nodes unequal to $n$, because it holds for $M$, and except for the new node $n$ it holds that $P_{M'}$ and $T_{M'}$ are subsets of $P_M$ and $T_M$, respectively. It also holds for the new node $n$, since it is either added to $P_{M'}$ or $T_{M'}$.

  \item Since this holds for $M$, and nodes outside $S$ and edges between them are not removed, it follows that edges from $M$ that are also present in $M'$ are between nodes in $P_{M'}$ and $T_{M'}$. For an edge in $M'$ between $n$ and another node $n'$ we can argue as follows. There must have been an edge in $M$ between a node $n''$ in $S$ and $n'$. In that case $n''$ must have been an input or output node of $M[S]$, and therefore of the same type as $n$, from which it follows that $n'$ is of a different type as $n$. 

  \item We first consider $I_{M'}$. In the case where $S \cap I_M = \emptyset$, the property holds since it holds for $M$. In the case where $S \cap I_M \neq \emptyset$ there will at least be $n$ in $I_{M'}$.
  
  The proof for $O_{M'}$ proceeds similarly.  

  \item We can show that if $M$ is pWF net, then it holds that $I_{M'}, O_{M'} \subseteq P_{M'}$, and if $M$ is a tWF net then it holds that $I_{M'}, O_{M'} \subseteq T_{M'}$. We consider the two cases:
   
     \begin{description}
     
       \item[$M$ is a pWF net:] We first derive that $I_{M'} \subseteq P_{M'}$ and for this we consider the cases (a) $M[S]$ is a pWF net and (b) $M[S]$ is a tWF net:
       
       \begin{enumerate}[(a)]
        
          \item If $M[S]$ is a pWF net, the argument goes as follows.
            \begin{itemize}

              \item  If $S \cap I_{M} = \emptyset$ then clearly $I_{M'} = I_M \subseteq (P_M \setminus S) \cup \{ n \} = P_{M'}$ since from $I_M \subseteq P_M$ and $S \cap I_{M} = \emptyset$ it follows $I_M \subseteq P_M \setminus S$.

              \item If, on the other hand,  $S \cap I_{M} \neq \emptyset$ then clearly $I_{M'} = (I_M \setminus S) \cup \{ n \} \subseteq (P_M \setminus S) \cup \{ n \} = P_{M'}$ since $I_M \subseteq P_M$.

            \end{itemize}
            
          \item If$M[S]$ is a tWF net, the argument goes as follows.    
            \begin{itemize}

              \item  If $S \cap I_{M} = \emptyset$ then clearly $I_{M'} = I_M \subseteq P_M \setminus S = P_{M'}$ since from $I_M \subseteq P_M$ and $S \cap I_{M} = \emptyset$ it follows $I_M \subseteq P_M \setminus S$. 

              \item If, on the other hand,  $S \cap I_{M} \neq \emptyset$, then we get a contradiction since a node $n' \in S \cap I_{M}$ would be by definition in $I_S$ and therefore a transition, since $M[S]$ is a tWF net, as well as a place in $M$, since $M$ is a pWF net. 

            \end{itemize}

        \end{enumerate}

       By a similar argument, where $I_M$ and $I_{M'}$ are replaced with $O_M$ and $O_{M'}$, respectively, it can be shown that $O_{M'} \subseteq P_{M'}$.
   
     \item[$M$ is a tWF net:] Proceeds similar as in the previous case, except that the cases (a) and (b) are swapped, i.e., the argument of (a) is used for the case where $M[S]$ is a tWF net, and the argument (b) is used where $M[S]$ is a pWF net.
   
   \end{description}
     
 \item A node in $M'$ is either a (a) a node in $M$ or (b) node $n$. In case (a) there would have been in $M$ a path from an input node of $M$ to that node. By Lemma~\ref{lem:contraction_keeps_paths} a similar path exists in $M'$. Moreover, the start node of the similar path will also be an input node of $M'$, since it either is a node in $M$, in which case it will still be an input node of $M'$, or it is replaced with $n$, in which case $n$ will be an input node of $M'$. In case (b) there would have been in $M$ a path from an input node of $M$ to an input node of $M[S]$. By Lemma~\ref{lem:contraction_keeps_paths} a similar path exists in $M'$ which ends in $n$. Also in this case it holds that the begin node of this path is an input node of $M'$.

 \item Is similar to case (v), but with the directions of paths and edges reversed, and the roles of set of input nodes and output nodes interchanged.
 
 % JS: I've added this extra case and strenghtened the thm (as there was a comment before the thm and this will be needed later)
 % JH: Ok
 
 \item This case follows immediately by the reasoning in (iv).
  
\end{enumerate}
\end{proof}

During the AND-OR reduction we will search for non-empty subset of nodes $S$ in a workflow net $M$ such that their associated net $M[S]$ is in a basic AND-OR class and contract them into a single new node. Such a contraction should be the inverse of a substitution of this new node with $M[S]$ in the sense that if we apply the contraction and then the substitution we should again have the same WF net. However, applying that substitution may not give back $M$ if, for example, not all input nodes of $M[S]$ have in $M$ the same incoming edges from outside $S$, since that is a property that always holds after a substitution. Therefore we define the notion of \emph{well-nestedness} for subnets that captures the properties of input and output nodes that guarantee applying the substitution after contraction always gives back $M$.

\begin{defn}[Well-nested subnet]
A subnet $S$ is said to be well-nested in $M$ if it holds that:
\begin{enumerate}
  \item for any two input nodes $n_1, n_2$ in $M[S]$, i.e., any two nodes $n_1, n_2 \in I_S$, it holds that
    \begin{enumerate}
      \item for every $n_3 \not\in N$ it holds that $(n_3, n_1) \in F_M$ if $(n_3, n_2) \in F_M$ 
      \item $n_1 \in I_M$ if $n_2 \in I_M$
    \end{enumerate}
  \item for any two output nodes $n_1, n_2$ in $M[S]$, i.e., any two nodes $n_1, n_2 \in O_S$, it holds that
    \begin{enumerate}
      \item for every $n_3 \not\in N$ it holds that $(n_1, n_3) \in F_M$ if $(n_2, n_3) \in F_M$ 
      \item $n_1 \in O_M$ if $n_2 \in O_M$
    \end{enumerate}    
\end{enumerate}
\end{defn}
A subnet that is well-nested, and where the associated net $M[S]$ is WF net (pWF, tWF) net will be called a \emph{well-nested WF (pWF, tWF) net}.

%\begin{prop}[Path-closure of well-nested WF nets]
%\label{lem:path-closure-of-well-nested-wf}
%For every well-nested WF net $N$ and a path in 
%the larger WF net between nodes in $N$ it holds
%that all the nodes of the path are in $N$ if (1) no nodes of the path are in $I_{S}$ except perhaps the first
%node or (2) no nodes of the path are in $O_{S}$ except perhaps the
%last node.
%\end{prop}
%
%\begin{proof}
%We consider the two types of paths. We first regard the case for paths
%satisfying (1). Consider the last internal node of the path (i.e., a node that is not a begin or an end node of the path) that is not in $N$. The node after it is therefore a node in $N$ that has
%an incoming edge arriving from outside $N$, and it must therefore
%be in $I_{S}$. However, this contradicts (1). For the case for paths
%of type (2), consider the first node of the path that is not in $N$.
%The node before it then has an outgoing edge leaving $N$ and must
%therefor be in $O_{S}$. However, this contradicts (2).
%\end{proof}

It is then not hard to see that contractions of well-nested subnets are the inverse of substitutions. To be more precise it holds that if a contraction is followed by a substitution that replaces the new node with the WF net associated with the contracted well-nested subnet,
then the result is the original WF net. Vice versa, it holds that if a substitution is followed by a contraction of the substituted well-nested net, we again obtain the original WF net, up to the choice of the identity of the new node.

\begin{defn}[AND-OR contractible]
We will call a well-nested WF net in $M$ \emph{AND-OR contractible} if (1) it is a well-nested WF
net and (2) is of a basic AND-OR class.
\end{defn}

Recall from Section~\ref{sect:AND-OR-nets} that the basic AND-OR classes are pAND,
11tAND, 11pOR and tOR. Based on this we define the AND-OR contraction relation which will form the basis of the reduction process.

\begin{defn}[AND-OR contraction relation]
The \emph{AND-OR contraction relation} is the binary relation $\contr$ over WF nets such that $M \contr M'$ expresses that there is an AND-OR contractible WF net $N$ in $M$ and that the contraction of $N$ in $M$ can result in $M'$. The reflexive and transitive closure of $\contr$ is denoted as $\contr^*$.  
\end{defn}

An interesting property of the contraction relation is that it has the \emph{local confluence} property, as stated by the following theorem.

\begin{thm}[local confluence of the AND-OR contraction relation] \label{thm:local-confluence}
 For all WF nets $M$, $M_1$ and $M_2$ it holds that if $M \contr M_1$ and $M \contr M_2$, then there is a WF net $M_3$ such that $M_1 \contr^* M_3$ and $M_2 \contr^* M_3$.
\end{thm}
Proof of this theorem will be given in Section~\ref{sect:confluence-proof}.

However, if we consider $\contr$ as the basis of a rewriting procedure where we contract contractible subnets until no more such nets can be found, we have the problem that this procedure will never terminate. This is because every node in a WF net is by itself a contractible subnet. Of course, in that case the contraction does not actually change the WF net, except that it replaces the node with a new node. This could be solved by restricting ourselves to contractions of non-trivial subnets. However, in that case we lose the property of local (and global) confluence, since the rewriting procedure could pick different node identifiers when it contracts a certain subnet and thereby end up with different nets. However, except for the choice of the node identifiers, these nets would be the same, and in that sense the procedure would in fact be confluent. To show this formally we define another contraction relation that holds between equivalence classes of WF nets containing isomorphic WF nets. Here we define WF nets as isomorphic if they are identical up to the identity of the nodes.

\begin{defn}[net isomorphism] 
Given WF nets $M = (P_{M},T_{M},F_{M},I_{M},O_{M})$ and $M' = (P_{M'},T_{M'},F_{M'},I_{M'},O_{M'})$, we say that a function 
$h : (P_M \cup T_M) \to (P_{M'} \cup T_{M'})$  is an \emph{isomorphism} from $M$ to $M'$, denoted as $M \sim_h M'$, if it holds that
\begin{itemize}

\item $P_{M'} = \{ h(n) \mid n \in P_M \}$, $T_{M'} = \{ h(n) \mid n \in T_M \}$, 

\item $F_{M'} = \{ (h(n_1), h(n_2) ) \mid (n_1, n_2) \in F_M \}$, 

\item $I_{M'} = \{ h(n) \mid n \in I_M \}$ and $O_{M'} = \{ h(n) \mid n \in O_M \}$. 

\end{itemize}
We say that $M_1$ and $M_2$ are \emph{isomorphic}, written as $M_1 \sim M_2$, if for some $h$ it holds that $M_1 {\sim_h} M_2$.
\end{defn}

It is easy to verify that the isomorphism relation is an equivalence relation, and so can be used to define equivalence classes. We will let $[M]$ denote the class of WF nets isomorphic to a WF net $M$. The AND-OR contraction relation can be adapted to the equivalence classes of an isomorphism.

\begin{defn}[class-based AND-OR contraction relation]
The \emph{class-based AND-OR contraction relation} is the binary relation $\clcontr$ over WF nets such that $[M] \clcontr [M']$ expresses that (1) $[M] \neq [M']$ and (2) there are WF nets $M_1 \in [M]$ and $M_2 \in [M']$, such that $M_1 \contr M_2$.  The reflexive and transitive closure of $\clcontr$ is denoted as $\clcontr^*$.  
\end{defn}

Note that the class-based contraction relation is irreflexive, and therefore does not consider contractions of one-node WF nets, which always result in an isomorphic net. So only contractions of subnets with two or more nodes are considered, which necessarily produces a WF net that is not isomorphic, since it strictly reduces the number of nodes. It follows that $\clcontr$ is \emph{well-founded}.

\begin{prop}[well-foundedness of class-based contraction] \label{prop:well-founded}
	The relation $\clcontr$ has no infinite sequences $[M_1] \clcontr [M_2] \clcontr [M_3]  \clcontr \ldots$.
\end{prop}

Based on the class-based contraction relation we can then define the notion of AND-OR reduction. The well-foundedness of class-based contraction ensures that a procedure that tries to determine a class-based AND-OR reduction by contracting contractible non-trivial subnets of basic AND-OR types, will always terminate.

\begin{defn}[class-based AND-OR reduction]
 A \emph{class-based AND-OR reduction} of a WF net $M$ is a WF net $M'$ such that (1) $[M] \clcontr^* [M']$ and (2) there is no WF net $M''$ such that $[M'] \clcontr [M'']$.
\end{defn}

Another observation that can be made is that the definition of contraction is generic with respect to node identity, i.e., it compares nodes to each other but never to a specific value. It follows that the contraction relation commutes with the isomorphism relation.

\begin{lem}[genericity of contraction] \label{lem:contract-genericity}
If $M_1 \contr M_2$ and $M_1 \sim M'_1$, then there is a WF net $M'_2$ such that $M_2 \sim M'_2$ and $M'_1 \contr M'_2$.
\end{lem}

\begin{proof}
Let us assume that $M_1 \contr M_2$ and $M_1 \sim_h M'_1$. Then it is not hard to see that if a subnet $S$ is contractible in $M_1$, the subset $M'_1[S']$ where $S' = \{ h(n) \mid n \in S \}$ is contractible in $M'_1$. Moreover, the contraction of this subnet will be isomorphic to the result of the contraction of $S$, where the homomorphism is identical to $h$, except that it maps the node that $M_1[S]$ was contracted into, to the node that $M'_1[S']$ is contracted into. 
\end{proof}

A consequence of this is that the local confluence property is inherited from the normal contraction relation.

\begin{thm}[local confluence of the class-based AND-OR contraction relation] \label{thm:class-based-local-confluence}
 For all WF nets $M$, $M_1$ and $M_2$ it holds that if $[M] \contr [M_1]$ and $[M] \contr [M_2]$, then there is a WF net $M_3$ such that $[M_1] \contr^* [M_3]$ and $[M_2] \contr^* [M_3]$.
\end{thm}

\begin{proof}
Let us assume that $[M] \contr [M_1]$ and $[M] \contr [M_2]$. Then by the definition of class-based contraction there are WF nets $M'$ and $M'_1$ such that $M' \sim M$, $M'_1 \sim M_1$ and $M' \contr M'_1$, and that there  are WF nets $M''$ and $M''_2$ such that $M'' \sim M$, $M''_2 \sim M_2$ and $M'' \contr M''_2$. Since $M' \sim M$ and $M'' \sim M$, it follows that $M' \sim M''$. It then follows by Lemma~\ref{lem:contract-genericity}, from $M' \sim M''$ and $M'' \contr M''_2$, that there is a WF net $M'_2$ such that $M'_2 \sim M''_2$ and $M' \contr M'_2$. It then follows by Theorem~\ref{thm:local-confluence}, from $M' \contr M'_1$ and $M' \contr M'_2$,  that there is a WF net $M_3$ such that $M'_1 \contr^* M_3$ and $M'_2 \contr^* M_3$. By induction on the number of contractions, it follows from the definition of class-based contractions, that $[M'_1] \contr^* [M_3]$ and $[M'_2] \contr^* [M_3]$. Since $M_1 \sim M'_1$ and $M_2 \sim M'_2$, it holds that $[M'_1] = [M_1]$ and $[M'_2] = [M_2]$, and so $[M_1] \contr^* [M_3]$ and $[M_2] \contr^* [M_3]$. 
\end{proof}

Since the class-based contraction relation is well-founded, as stated by Property~\ref{prop:well-founded},
it follows by \emph{Newman's Lemma} (sometimes also called the \emph{diamond Lemma}, see \cite{Huet:1980}) that the procedure is globally confluent.

\begin{thm}[global confluence of the class-based AND-OR contraction relation] \label{thm:class-based-global-confluence}
For all WF nets $M$, $M_1$ and $M_2$ it holds that if $[M] \contr^* [M_1]$ and $[M] \contr^* [M_2]$, then there is a WF net $M_3$ such that $[M_1] \contr^* [M_3]$ and $[M_2] \contr^* [M_3]$.
\end{thm}

An immediate consequence of this is that the result of the class-based AND-OR reduction is unique.

\begin{thm}[unique result of class-based AND-OR reduction] \label{thm:unique-result-reduction}
Let $[M_1]$ and $[M_2]$ be two results of the class-based AND-OR reduction of a WF net $M$, then $[M_1] = [M_2]$.
\end{thm}

\begin{proof}
From global confluence, Theorem~\ref{thm:class-based-global-confluence}, it follows that there is a WF net $M_3$ such that $[M_1] \clcontr^* [M_3]$ and $[M_2] \clcontr^* [M_3]$. However, by the definition of the class-based AND-OR reduction, it also holds that there is no WF net $M_4$ such that $[M_1] \clcontr [M_4]$, and also no WF net $M_4$ such that  $[M_2] \clcontr [M_4]$. It then follows that $[M_1] = [M_2]$.
\end{proof}

In an actual implementation of the reduction, the classes would be represented by actual WF nets. Therefore we define an alternative notion of AND-OR reduction that operates on WF nets.

\begin{defn}[net-based AND-OR reduction]
 A \emph{net-based AND-OR reduction} of a WF net $M$ is a WF net $M'$ such that (1) $M$ can be reduced to $M'$ via contractions that contract non-trivial subnets and (2) there is no contractible non-trivial subnet in $M''$.
\end{defn}

We can demonstrate a close relationship between the two types of AND-OR reductions.

\begin{thm}[similarity of AND-OR reductions] \label{thm:similarity-reduction}
Let $M$ and $M'$ be two WF nets, then $[M']$ is the result of applying the class-based AND-OR reduction to $M$ iff there is a WF net $M''$ such that $M''$ is the result of applying the net-based AND-OR reduction to $M$ and $M'' \sim M'$.
\end{thm}

\begin{proof}
We first consider the \emph{if}-part. It is clear that if there is a series of contractions that leads from $M$ to $M''$, then $[M] \clcontr^* [M'']$.  Since there are no contractible subnets in $M''$ of two nodes or more, it follows that only contractions of one-node subnets are possible, which leave the net isomorphic. And so, there is no WF net $M_2$ such that $[M''] \clcontr [M_2]$, since the class-based AND-OR contraction is by definition irreflexive.

We now consider the \emph{only-if} part. Assume that $[M] \clcontr^* [M']$. This means that $[M_0] \clcontr [M_1] \clcontr \ldots \clcontr [M_n]$ for some sequence of nets $M_0$, $M_1$, \ldots, $M_n$ where $M_0 \sim M$ and $M_n \sim M'$. Then by using Lemma~\ref{lem:contract-genericity} on the genericity of contraction, for each step in this sequence of class-based contractions it follows that there is a series of net-based contractions that reduces some net $M'_0$ which is isomorphic to $M$ to a WF net $M'_n$ which is isomorphic to $M'$. Moreover, each contraction between two classes will contract two nodes or more, and so the same will hold for this series of net contractions. It remains to observe that as there are no more non-trivial net contractions possible in $M'_n$ which follows from the fact that the class-based reduction stopped at $[M']$ and $M'_n \sim M'$. This completes the proof with $M'_n$ being the postulated $M''$.
\end{proof}

A direct consequence of Theorem~\ref{thm:similarity-reduction} on the similarity of net-based and class-based AND-OR reduction, and Theorem~\ref{thm:unique-result-reduction}  on the uniqueness of the result of the class-based AND-OR reduction, is that the result of the net-based AND-OR reduction is unique up to isomorphism.

\begin{thm}[unique result of net-based AND-OR reduction] \label{thm:unique-result-net-reduction}
Let $M_1$ and $M_2$ be two results of the net-based AND-OR reduction of a WF net $M$, then $M_1 \sim M_2$.
\end{thm}

\begin{proof}
Assume that $M_1$ and $M_2$ are both results of the net-based AND-OR reduction of a WF net $M$. By Theorem~\ref{thm:similarity-reduction} it follows that $[M_1]$ and $[M_2]$ are both results of the class-based AND-OR reduction of $M$. From the uniqueness of the result of the class-based reduction, as stated by Theorem~\ref{thm:unique-result-reduction}, it follows that $[M_1] = [M_2]$, which implies $M_1 \sim M_2$.
\end{proof}

\section{Local confluence of the AND-OR contraction relation
\label{sect:confluence-proof}}

In this section we show the local confluence of the AND-OR contraction relation, as stated by Theorem~\ref{thm:local-confluence}.
To simplify the presentation of the proof, we will at first restrict
ourselves to contractions that do not involve input and output nodes
of the larger net, i.e., where $S \cap I_{M}=S \cap O_{M}=\emptyset$. It is easy to see that for such contractions 
$I_{M'}=I_{M}$ and $O_{M'}=O_{M}$ if $M'$ is the result of the contraction, which simplifies the definition of the result of the contraction.
We will call such contractions \emph{internal contractions}, and a
well-nested WF net that does not share input nor output nodes with
the larger WF net will be called an \emph{internal well-nested WF
net.} We will call an internal well-nested net that is of one of the basic AND-OR classes \emph{internal
AND-OR contractible}. We let $\intcontr$ denote the restriction of $\contr$ where only internal contractions are used.

The restriction to internal contractions is interesting because the general reduction process can be simulated by (1) taking the place or transition completion of the input net, (2) applying only internal contractions and (3) removing the input an output place while making the nodes in the postset of the input node the new input places and the nodes in the preset of the output place the new output nodes. Each contraction in the reduction process on the original WF net will then correspond to an internal contraction on the completed WF net, and vice versa. It follows that the result is then the place completion of a single transition or a transition completion of a single place iff the result of the reduction process on the original input WF net is a single place or a single transaction.

Since in the remainder of this section we will talk mostly about internal AND-OR contractions, we will refer to internal contractions, internal well-nested subnets, and internal contractible nets simply as \emph{contractions}, \emph{well-nested subnets} and \emph{contractible nets}, respectively, unless stated otherwise.

%Since the formal definition of an internal contraction can be somewhat
%simplified with respect to the one for contractions that was given
%earlier, we give here that simplified definition:
%
%\begin{defn}[Internal contraction]
%Given a WF net $M=( P_{M},T_{M},F_{M},I_{M},O_{M})$
%and a subnet $S$ such that $M[S]$ is a pWF net  (tWF net), we define the result of contracting $N$
%into a new place (transition) node $n$ as $M'=( P_{M'},T_{M'},F_{M'},I_{M'}O_{M'})$
%where
%\begin{itemize}
%\item  $P_{M'}=(P_{M}\setminus N)\cup\{n\}$, ($P_{M'}=P_{M}\setminus N$),
%\item $T_{M'}  =T_{M}\setminus N$, ($T_{M'}=(T_{M}\setminus N)\cup\{n\}$),
%\item \begin{tabbing} 
%$F_{M'} = $ \= $\{ (n_{1},n_{2}) \mid (n_{1},n_{2}) \in F_{M}, n_{1} \not\in N,n_{2}\not\in N \} \ \cup$ \\ 
%\>  $\{ (n_{1},n) \mid (n_{1}, n_{2}) \in F_{M}, n_{2}\in I, n_1 \not\in N \} \ \cup$ \\
%\> $\{ (n,n_{2}) \mid (n_{1}, n_{2}) \in F_{M}, n_{1}\in O, n_1 \not\in N \}$,
%\end{tabbing}
%\item $I_{M'} = I_M$
%\item $O_{M'} = O_M$
%\end{itemize}
%\end{defn}
%}

To show local confluence, we will set out to prove that if there are two
subnets that are both contractible, then (1) the contraction
of one does not change the contractibility of the
other one (although its structure might change somewhat because of
the contraction) and (2) the result of contracting the two nets in
either order results in the same WF net. This is straightforward to see if the subnets are completely independent, i.e., they do not overlap and there are no edges that connect a node in one subnet to a node in the other subnet. In this case the contraction of one subnet does not change the other subnet or the edges connected to it, and so does not change the contractibility. If the two subnets are not completely independent we can distinguish four cases, which are illustrated in Figure~\ref{fig:cases-of-proof}. In this figure we use clouds to indicate subnets and rounded rectangles to represent nodes that can be either places or transitions.

\begin{figure}[htb]
\begin{center}
\resizebox{0.8\textwidth}{!}{%
\begin{tikzpicture}
    \tikzstyle{transition} = [rectangle,draw,minimum width=0.47cm, minimum height=0.47cm,fill=white]
    \tikzstyle{place} = [circle,draw,minimum width=0.5cm, minimum height=0.5cm,fill=white,inner sep=0.08cm]
    \tikzstyle{place-transition} = [rounded rectangle,draw,minimum width=1cm, minimum height=0.5cm,fill=white,inner sep=0.08cm]

% Case A

   \node (case-A-top-left) [cloud,draw,cloud puffs=10,cloud puff arc=100, aspect=1.5,
          fill opacity=0.5, text opacity=1,
          minimum width=2cm,minimum height=1.7cm,fill=mygray1] 
     at (0cm, 0cm) {$S_1$};
   \node (case-A-top-right) [cloud,draw,cloud puffs=10,cloud puff arc=100, aspect=1.5,
          fill opacity=0.5, text opacity=1,
          minimum width=2cm,minimum height=1.7cm,fill=mygray2] 
     at (2.5cm, 0cm) {$S_2$};

  \path (case-A-top-left) edge[-] (case-A-top-right) ; 
  \path (case-A-top-left) edge[-, bend left=20] (case-A-top-right) ; 
  \path (case-A-top-left) edge[-, bend right=20] (case-A-top-right) ; 

  \node (case-A-go-left) [rotate=-135] at (0.25cm, -1.25cm) {\Large $\contr$};

   \node (case-A-left-left) [place-transition,draw,fill=mygray1,fill opacity=0.5, text opacity=1] 
     at (-2cm, -2.5cm) {$n_1$};
   \node (case-A-left-right) [cloud,draw,cloud puffs=10,cloud puff arc=100, aspect=1.5,
          fill opacity=0.5, text opacity=1,
          minimum width=2cm,minimum height=1.7cm,fill=mygray2]
     at (0cm, -2.5cm) {$S_2$};

  \path (case-A-left-left) edge[-, bend left=15] (case-A-left-right) ; 
  \path (case-A-left-left) edge[-, bend right=15] (case-A-left-right) ; 

  \node (case-A-go-right) [rotate=-45] at (2.25cm, -1.25cm) {\Large $\contr$};

   \node (case-A-right-left) [cloud,draw,cloud puffs=10,cloud puff arc=100, aspect=1.5,
          fill opacity=0.5, text opacity=1,
          minimum width=2cm,minimum height=1.7cm,fill=mygray1] 
     at (2.5cm, -2.5cm) {$S_1$};
   \node (case-A-right-right) [place-transition,draw,fill=mygray2,fill opacity=0.5, text opacity=1] 
     at (4.5,-2.5cm) {$n_2$};

  \path (case-A-right-left) edge[-, bend left=15] (case-A-right-right) ; 
  \path (case-A-right-left) edge[-, bend right=15] (case-A-right-right) ; 

  \node (case-A-go-right2) [rotate=-55] at (0.25cm, -3.75cm) {\Large $\contr$};
  \node (case-A-go-left2) [rotate=-135] at (2.25cm, -3.75cm) {\Large $\contr$};

   \node (case-A-bottom-left) [place-transition,draw,fill=mygray1,fill opacity=0.5, text opacity=1] 
     at (0.5cm,-4.5cm) {$n_1$};
   \node (case-A-bottom-right) [place-transition,draw,fill=mygray2,fill opacity=0.5, text opacity=1] 
     at (2cm,-4.5cm) {$n_2$};

  \path (case-A-bottom-left) edge[-] (case-A-bottom-right) ; 

  \node (case-A) at (1.25cm,-5.5cm) {(A) No overlap, but connected};

% Case B

   \node (case-B-top-left) [cloud,cloud puffs=10,cloud puff arc=100, aspect=1.5,
       minimum width=2cm,minimum height=1.7cm,fill=mygray1, fill opacity=0.5, text opacity=1,draw] 
     at (9cm, 0cm) {$S_1$};
   \node (case-B-top-right) [cloud,draw,cloud puffs=10,cloud puff arc=100, aspect=1.5,
       fill opacity=0.5, text opacity=1,
       minimum width=2cm,minimum height=1.7cm,fill=mygray2] 
     at (10.5cm, 0cm) {$S_2$};

  \node (case-B-go-left) [rotate=-135] at (8.75cm, -1.25cm) {\Large $\contr$};

   \node (case-B-left-right) [cloud,draw,cloud puffs=10,cloud puff arc=100, aspect=1.5,
          fill opacity=0.5, text opacity=1,
          minimum width=2cm,minimum height=1.7cm,fill=mygray2]
     at (8.5cm, -2.5cm) {};
     \node at (8.7cm, -2.5cm) {$S_2 \setminus S_1$};
% the following clips the cloud of S2     
   \node (case-B-left-right-cut) [draw,cloud,cloud puffs=10,cloud puff arc=100, aspect=1.5,
       minimum width=2cm,minimum height=1.7cm,fill=white,draw=white] 
     at (7cm, -2.5cm) {};
   \node (case-B-left-left) [place-transition,draw,fill=mygray1,fill opacity=0.5, text opacity=1] 
     at (7cm, -2.5cm) {$n_1$};     

% the following puts back the edge of S2 that is missing because of the clip
  \begin{scope}[]
    \pgftransformshift{\pgfpoint{8.5cm}{-2.5cm}}
    \pgfset{cloud puffs=10,cloud puff arc=100, aspect=1.5,
       minimum width=2cm,minimum height=1.7cm}
    \pgfnode{cloud}{center}{}{nodename}{\pgfusepath{clip}}

   \node [draw,cloud,cloud puffs=10,cloud puff arc=100, aspect=1.5,
       minimum width=2cm,minimum height=1.7cm] 
     at (case-B-left-left.center) {};
  \end{scope}

% add edges from n1 to S2, but clip so it does not enter S2
  \begin{scope}[]
    \pgftransformshift{case-B-left-left.center}
    \pgfset{cloud puffs=10,cloud puff arc=100, aspect=1.5,
       minimum width=2cm,minimum height=1.7cm}
    \pgfnode{cloud}{center}{}{nodename}{\pgfusepath{clip}}
    
    \path (case-B-left-left) edge[-, bend left=30] (case-B-left-right.center) ; 
    \path (case-B-left-left) edge[-] (case-B-left-right.center) ;     
    \path (case-B-left-left) edge[-, bend right=30] (case-B-left-right.center) ; 
  \end{scope}
  
  \node (case-B-go-right) [rotate=-45] at (10.75cm, -1.25cm) {\Large $\contr$};

   \node (case-B-right-left) [cloud,draw,cloud puffs=10,cloud puff arc=100, aspect=1.5,
          fill opacity=0.5, text opacity=1,
          minimum width=2cm,minimum height=1.7cm,fill=mygray1] 
     at (11cm, -2.5cm) {};
   \node at (10.8cm, -2.5cm) {$S_1 \setminus S_2$};
% the following clips the cloud of S1     
   \node (case-B-right-left-cut) [draw,cloud,cloud puffs=10,cloud puff arc=100, aspect=1.5,
       minimum width=2cm,minimum height=1.7cm,fill=white,draw=white] 
     at (12.5cm, -2.5cm) {};
   \node (case-B-right-right) [place-transition,draw,fill=mygray2,fill opacity=0.5, text opacity=1] 
     at (case-B-right-left-cut.center) {$n_2$};

% the following puts back the edge of S2 that is missing because of the clip
  \begin{scope}[]
    \pgftransformshift{\pgfpoint{11cm}{-2.5cm}}
    \pgfset{cloud puffs=10,cloud puff arc=100, aspect=1.5,
       minimum width=2cm,minimum height=1.7cm}
    \pgfnode{cloud}{center}{}{nodename}{\pgfusepath{clip}}

   \node [draw,cloud,cloud puffs=10,cloud puff arc=100, aspect=1.5,
       minimum width=2cm,minimum height=1.7cm] 
     at (case-B-right-right.center) {};
  \end{scope}

% add edges from n2 to S1, but clip so it does not enter clipped S1
  \begin{scope}[]
    \pgftransformshift{\pgfpoint{12.5cm}{-2.5cm}}
    \pgfset{cloud puffs=10,cloud puff arc=100, aspect=1.5,
       minimum width=2cm,minimum height=1.7cm}
    \pgfnode{cloud}{center}{}{nodename}{\pgfusepath{clip}}
    
    \path (case-B-right-right) edge[-, bend left=30] (case-B-right-left.center) ; 
    \path (case-B-right-right) edge[-] (case-B-right-left.center) ;     
    \path (case-B-right-right) edge[-, bend right=30] (case-B-right-left.center) ; 
  \end{scope}

  \node (case-B-go-right2) [rotate=-55] at (8.75cm, -3.75cm) {\Large $\contr$};
  \node (case-B-go-left2) [rotate=-135] at (10.75cm, -3.75cm) {\Large $\contr$};

   \node (case-B-bottom-left) [place-transition,draw,fill=mygray1,fill opacity=0.5, text opacity=1] 
     at (9cm,-4.5cm) {$n_1$};
   \node (case-B-bottom-right) [place-transition,draw,fill=mygray2,fill opacity=0.5, text opacity=1] 
     at (10.5cm,-4.5cm) {$n_2$};

  \path (case-B-bottom-left) edge[-] (case-B-bottom-right) ; 

  \node (case-B) at (10cm,-5.5cm) {(B) Overlap, but not nested, different I/O type};

% Case C

   \node (case-C-top-left) [cloud,cloud puffs=10,cloud puff arc=100, aspect=1.5,
       minimum width=2cm,minimum height=1.7cm,fill=mygray1, fill opacity=0.5, text opacity=1,draw] 
     at (0.5cm, -7.5cm) {$S_1$};
   \node (case-C-top-right) [cloud,draw,cloud puffs=10,cloud puff arc=100, aspect=1.5,
       fill opacity=0.5, text opacity=1,
       minimum width=2cm,minimum height=1.7cm,fill=mygray2] 
     at (2cm, -7.5cm) {$S_2$};

  \node (case-C-go-left) [rotate=-135] at (0.3cm, -8.75cm) {\Large $\contr$};

   \node (case-C-left-right) [cloud,draw,cloud puffs=10,cloud puff arc=100, aspect=1.5,
          fill opacity=0.5, text opacity=1,
          minimum width=2cm,minimum height=1.7cm,fill=mygray2]
     at (0cm, -10cm) {};
   \node at (0.2cm, -10cm) {$S_2 \setminus S_1$};
% the following clips the cloud of S2     
   \node (case-C-left-right-cut) [draw,cloud,cloud puffs=10,cloud puff arc=100, aspect=1.5,
       minimum width=2cm,minimum height=1.7cm,fill=white,draw=white] 
     at (-1.5cm, -10cm) {};
   \node (case-C-left-left) [place-transition,draw,fill=mygray1,fill opacity=0.5, text opacity=1] 
     at (-1.5cm, -10cm) {$n_1$};     

% the following puts back the edge of S2 that is missing because of the clip
  \begin{scope}[]
    \pgftransformshift{\pgfpoint{0cm}{-10cm}}
    \pgfset{cloud puffs=10,cloud puff arc=100, aspect=1.5,
       minimum width=2cm,minimum height=1.7cm}
    \pgfnode{cloud}{center}{}{nodename}{\pgfusepath{clip}}

   \node [draw,cloud,cloud puffs=10,cloud puff arc=100, aspect=1.5,
       minimum width=2cm,minimum height=1.7cm] 
     at (case-C-left-left.center) {};
  \end{scope}

% add edges from n1 to S2, but clip so it does not enter S2
  \begin{scope}[]
    \pgftransformshift{case-C-left-left.center}
    \pgfset{cloud puffs=10,cloud puff arc=100, aspect=1.5,
       minimum width=2cm,minimum height=1.7cm}
    \pgfnode{cloud}{center}{}{nodename}{\pgfusepath{clip}}
    
    \path (case-C-left-left) edge[-, bend left=30] (case-C-left-right.center) ; 
    \path (case-C-left-left) edge[-] (case-C-left-right.center) ;     
    \path (case-C-left-left) edge[-, bend right=30] (case-C-left-right.center) ; 
  \end{scope}
  
  \node (case-C-go-right) [rotate=-45] at (2.3cm, -8.75cm) {\Large $\contr$};

   \node (case-C-right-left) [cloud,draw,cloud puffs=10,cloud puff arc=100, aspect=1.5,
          fill opacity=0.5, text opacity=1,
          minimum width=2cm,minimum height=1.7cm,fill=mygray1] 
     at (2.5cm, -10cm) {};
   \node at (2.3cm, -10cm) {$S_1 \setminus S_2$};
% the following clips the cloud of S1     
   \node (case-C-right-left-cut) [draw,cloud,cloud puffs=10,cloud puff arc=100, aspect=1.5,
       minimum width=2cm,minimum height=1.7cm,fill=white,draw=white] 
     at (4cm, -10cm) {};
   \node (case-C-right-right) [place-transition,draw,fill=mygray2,fill opacity=0.5, text opacity=1] 
     at (case-C-right-left-cut.center) {$n_2$};

% the following puts back the edge of S2 that is missing because of the clip
  \begin{scope}[]
    \pgftransformshift{\pgfpoint{2.5cm}{-10cm}}
    \pgfset{cloud puffs=10,cloud puff arc=100, aspect=1.5,
       minimum width=2cm,minimum height=1.7cm}
    \pgfnode{cloud}{center}{}{nodename}{\pgfusepath{clip}}

   \node [draw,cloud,cloud puffs=10,cloud puff arc=100, aspect=1.5,
       minimum width=2cm,minimum height=1.7cm] 
     at (case-C-right-right.center) {};
  \end{scope}

% add edges from n2 to S1, but clip so it does not enter clipped S1
  \begin{scope}[]
    \pgftransformshift{\pgfpoint{4cm}{-10cm}}
    \pgfset{cloud puffs=10,cloud puff arc=100, aspect=1.5,
       minimum width=2cm,minimum height=1.7cm}
    \pgfnode{cloud}{center}{}{nodename}{\pgfusepath{clip}}
    
    \path (case-C-right-right) edge[-, bend left=30] (case-C-right-left.center) ; 
    \path (case-C-right-right) edge[-] (case-C-right-left.center) ;     
    \path (case-C-right-right) edge[-, bend right=30] (case-C-right-left.center) ; 
  \end{scope}

  \node (case-C-go-right2) [rotate=-55] at (0.3cm, -11.25cm) {\Large $\contr$};
  \node (case-C-go-left2) [rotate=-135] at (2.25cm, -11.25cm) {\Large $\contr$};

   \node (case-C-bottom-left) [place-transition,draw,fill=mygray1,fill opacity=0.5, text opacity=1] 
     at (1.3cm,-12cm) {$n_3$};

  \node (case-B) at (1.55cm,-13cm) {(C) Overlap, but not nested, same I/O type};

% Case D

   \node (case-D-top-left) [cloud,cloud puffs=15,cloud puff arc=100,
       minimum width=3.5cm,minimum height=1.7cm,fill=mygray1, fill opacity=0.5, text opacity=1,draw] 
     at (9.75cm, -7.5cm) {};
   \node at (9cm, -7.5cm) {$S_1$};
   \node (case-D-top-right) [cloud,draw,cloud puffs=10,cloud puff arc=100, aspect=1.5,
       fill opacity=0.5, text opacity=1,
       minimum width=1.5cm,minimum height=1.25cm,fill=mygray2] 
     at (10.2cm, -7.5cm) {$S_2$};

  \node (case-D-go-right) [rotate=-45] at (10.75cm, -8.75cm) {\Large $\contr$};

  \path (8.9cm, -8.6cm) edge[->, bend right=50, line width=0.55, line join=round, decorate, decoration={zigzag, segment length=6, amplitude=1.5, post=lineto, post length=3pt}] (8.9cm, -11.4cm);

   \node (case-D-right-left) [cloud,draw,cloud puffs=15,cloud puff arc=100, aspect=1.5,
          fill opacity=0.5, text opacity=1,
          minimum width=3.5cm,minimum height=1.7cm,fill=mygray1] 
     at (11cm, -10cm) {};
   \node at (10.02cm, -10cm) {$S_1 \setminus S_2$};
% the following clips the cloud of S1     
   \node (case-D-right-left-cut) [draw,cloud,cloud puffs=10,cloud puff arc=100, aspect=1.5,
       minimum width=1.5cm,minimum height=1.25cm,fill=white] 
     at (11.45cm, -10cm) {};
   \node (case-D-right-right) [place-transition,draw,fill=mygray2,fill opacity=0.5, text opacity=1] 
     at (case-D-right-left-cut.center) {$n_2$};

% add edges from n2 to S1, but clip so it does not enter clipped S1
  \begin{scope}[]
    \pgftransformshift{\pgfpoint{11.45cm}{-10cm}}
    \pgfset{cloud puffs=10,cloud puff arc=100, aspect=1.5,
       minimum width=1.5cm,minimum height=1.25cm}
    \pgfnode{cloud}{center}{}{nodename}{\pgfusepath{clip}}
    
    \path (case-D-right-right) edge[-, bend left=30] (-1.2cm, 0cm) ; 
    \path (case-D-right-right) edge[-] (0.9cm, 0cm) ;     
    \path (case-D-right-right) edge[-, bend right=30] (-1.2cm, 0cm) ; 
  \end{scope}

  \node (case-D-go-left2) [rotate=-135] at (10.7cm, -11.25cm) {\Large $\contr$};

   \node (case-D-bottom-left) [place-transition,draw,fill=mygray1,fill opacity=0.5, text opacity=1] 
     at (9.75cm,-12cm) {$n_1$};

  \node (case-B) at (10cm,-13cm) {(D) Nesting};

\end{tikzpicture}
}
\end{center}
\caption{\label{fig:cases-of-proof}The four cases considered in the confluence proof}
\end{figure}
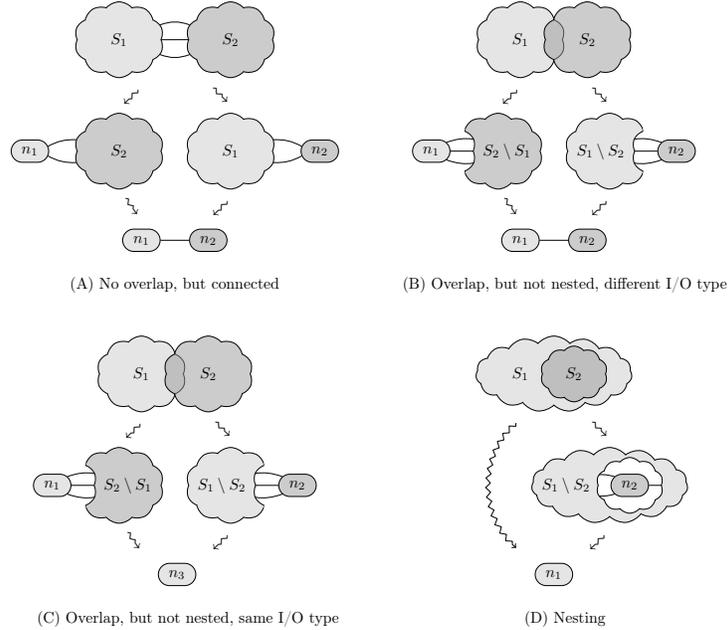

\begin{enumerate}[(A)]

  \item \emph{The subnets $S_1$ and $S_2$ do not share nodes, but nodes in $S_1$ are connected to nodes in $S_2$}: In this case the contraction of one subnet will not change the other subnet, but it might change the edges connected to it. However, as will be shown, this subnet will remain well-nested and therefore contractible. Note that there might be multiple edges between $S_1$ and $S_2$, but after contracting subnets defined by them into $n_1$ and $n_2$, there will be at most one edge from $n_1$ to $n_2$, and vice versa. 
  
  \item  \emph{The subnets $S_1$ and $S_2$ overlap, but are not nested in one another and the I/O type of the two nets is different}: In this case the contraction of one subnet, say $M[S_1]$, removes a part of $S_2$. It must therefore be shown that after the contraction the remainder of $M[S_2]$, i.e., the subnet $S_2 \setminus S_1$ is a contractible net, i.e., it is well-nested and belongs to a basic AND-OR class, or is already a one-node WF net, in which case no further contraction is required. In fact, it will be shown that it will belong to the same basic AND-OR class as $M[S_2]$. Note that since $M[S_1]$ and $M[S_2]$ have different I/O types, it follows that $n_1$ and $n_2$ have different types, which is required in a Petri net for connected nodes.
  
  It must also be shown that the final result with nodes $n_1$ and $n_2$ is the same, independent of whether first $S_1$ or first $S_2$ is contracted. This means that (a) the edges between $n_1$ and $n_2$ must be the same, (b) the edges between $n_1$ and nodes other than $n_2$ must be the same and (c) the edges between $n_2$ and nodes other than $n_1$ must be the same. 
  
  \item \emph{The subnets $S_1$ and $S_2$ overlap, are not nested in one another and the I/O type of the two nets is the same}: In this case we contract in the second step not just the subnet consisting of the remainder of the subnet that was not contracted, say $M'[S_2 \setminus S_1]$, but the total of that remainder plus the new node, which would be the subnet $(S_2 \setminus S_1) \cup  \{ n_1 \}$. Therefore it must be shown that this subnet is a contractible net. In fact, it will be shown that it will belong to the same basic AND-OR class as $M[S_2]$. Note that since $M[S_1]$ and $M[S_2]$ have the same I/O type, we cannot simply contract $S_2 \setminus S_1$ in $M'$, since that would result in a node of the same type as $n_1$.
  
  It must also be shown that the final result with node $n_3$ is the same, independent of whether first $S_1$ or first $S_2$ is contracted. This means that $n_3$ must be participating in the same edges.
  
  \item \emph{The subnet $S_1$ contains the subnet $S_2$}: The first option is to contract $S_1$ into $n_1$, in which case $S_2$ completely disappears. The second option is to contract $S_2$ into $n_2$. It will be shown that after this step the subnet defined by of the remainder of $S_1$ plus the new node $n_2$, i.e., $(S_2 \setminus S_1) \cup  \{ n_2 \}$, is contractible in $M'$. In fact, it will be shown that it will belong to the same basic AND-OR class as $M[S_1]$ in the original workflow net.
  
  It must also be shown that the final result with node $n_1$ is the same, independent of whether first $S_1$ or first $S_2$ is contracted. This means that $n_1$ must be participating in the same edges.

\end{enumerate}
In the following subsections we will discuss the necessary proofs for each of the mentioned cases in detail.

\subsection{Connected non-overlapping subnets}

In this subsection we consider case (A) in Figure~\ref{fig:cases-of-proof}.  Therefore, let us assume that $S_1$ and $S_2$ are contractible subnets in a WF net $M$, and do not share nodes, but are connected by edges. Let us also assume that $M'$ is the result of contracting $S_1$ into the node $n_1$ and $M''$ the result of contracting $S_2$ into the node $n_2$. Recall that we need to show that (1) the contraction of $S_1$ does not affect the contractibility of $S_2$ and (2) we get the same result by contracting $S_1$ and then $S_2$ as with contracting $S_2$ and then $S_1$. This is captured by the following two lemmas.

\begin{lem}[preservation of contractibility]
\label{lem:contractibility-A}
The subnet $S_2$ is contractible in $M'$.
\end{lem}

\begin{proof}
Recall that a subnet $S$ is contractible in $M$ if the following two properties hold: (C1) $N$ is well-nested and (C2) $M[S]$ is of a basic AND-OR type. Since the contraction of $S_1$ does not change the nodes of $S_2$ or edges between them, it is clear that $M'[S_2]$ has the same places, transitions and edges as $M[S_2]$. Also the set of input nodes is the same, because if in $M$ a node in $S_2$ had an incoming edge from outside $S_2$, it will still have such an edge after the contraction, and for internal subnets this defines whether a node is an input. By a similar argument it holds that the set of output nodes remains the same. Since all the components doe not change, it follows that (C2) is preserved for $M'[S_2]$. So it remains to be shown that property (C1) is preserved.

Since we consider only internal subnets, it is sufficient to show that (1) for all two edges $(n_2, n_3)$ and $(n_4, n_5)$ in $M'$ such that $n_2, n_4 \not\in S_2$ and $n_3, n_5 \in S_2$, there are also the edges $(n_2, n_5)$ and $(n_4, n_3)$ in $M'$, and (2) for all two edges $(n_2, n_3)$ and $(n_4, n_5)$ in $M'$ such that $n_3, n_5 \not\in S_2$ and $n_2, n_4 \in S_2$, there are also the edges $(n_2, n_5)$ and $(n_4, n_3)$ in $M'$.

To show (1) we start with assuming that there are two edges $(n_2, n_3)$ and $(n_4, n_5)$ in $M'$ such that $n_2, n_4 \not\in S_2$ and $n_3, n_5 \in S_2$. Then, we consider the three possible cases were (i) both $n_2$ and $n_4$ are equal to $n_1$, (ii) only one of $n_2$ and $n_4$ is equal to $n_1$ and (iii) both $n_2$ and $n_4$ are unequal to $n_1$.
\begin{enumerate}[(i)]
  \item It then follows that $(n_2, n_3) = (n_1, n_3) = (n_4, n_3)$ and $(n_4, n_5) = (n_1, n_5) = (n_2, n_5)$, and so the edges  $(n_4, n_3)$ and  $(n_2, n_5)$ indeed also exist in $M'$.
  \item We can assume without loss of generality that $n_2 = n_1$ and $n_4 \neq n_1$. The reasoning is then as illustrated in Figure~\ref{fig:well-nested-A} (a). By definition of contraction there is in $M$ an edge $(n_4, n_5)$ and an edge $(n'_2, n_3)$ with $n'_2$ an output node of $M[S_1]$. Since $S_2$ is by assumption well-nested in $M$, there are therefore also the edges $(n'_2, n_5)$ and $(n_4, n_3)$ in $M$. By the definition of contraction it then follows that there are edges $(n_1, n_5)$ and $(n_4, n_3)$ in $M'$.   
  \item  The reasoning is then as illustrated in Figure~\ref{fig:well-nested-A} (b). By definition of contraction there are edges $(n_2, n_3)$ and $(n_4, n_5)$ in $M$. Since $S_2$ is by assumption well-nested in $M$, there are therefore also the edges $(n_2, n_5)$ and $(n_4, n_3)$ in $M$. By the definition of contraction it then follows that there are edges $(n_1, n_5)$ and $(n_4, n_3)$ in $M'$.    
\end{enumerate}

The proof for (2) is similar, but with the direction of the edges reversed and the roles of input and output nodes interchanged.
\end{proof}

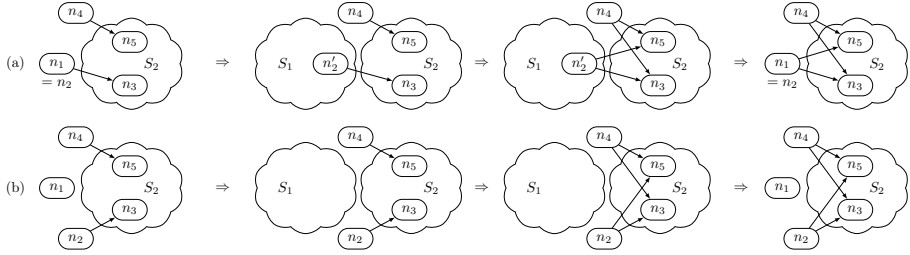
\begin{figure}[htb]
\begin{center}
\resizebox{\textwidth}{!}{%
\begin{tikzpicture}
    \tikzstyle{transition} = [rectangle,draw,minimum width=0.55cm, minimum height=0.55cm,fill=white]
    \tikzstyle{place} = [circle,draw,minimum width=0.55cm, minimum height=0.55cm,fill=white,inner sep=0.08cm]
    \tikzstyle{place-transition} = [rounded rectangle,draw,minimum width=1cm, minimum height=0.5cm,fill=white,inner sep=0.08cm]

\begin{scope} % start case (a), first line

\node at (0cm, 0cm) {(a)};

\begin{scope}[shift={(1cm, 0cm)}]
   \node (n1) [place-transition, draw, text opacity=1, fill=white] 
     at (0cm, 0cm) {$n_1$};     
   \node [below=0cm of n1] {$=n_2$};

   \node (S2) [cloud,draw,cloud puffs=10,cloud puff arc=100, aspect=1.5,
          text opacity=1,
          minimum width=2.5cm,minimum height=2.2cm,
          right=0.2cm of n1]
     {};
   \node [right=-1cm of S2] {$S_2$};

   \node (n3) [place-transition, draw, text opacity=1, 
     below left=-0.5cm and -1cm of S2] 
     {$n_3$};     
   \node (n5) [place-transition, draw, text opacity=1, 
     above left=-0.5cm and -1cm of S2] 
     {$n_5$};     

   \node (n4) [place-transition, draw, text opacity=1,
     above left=0.2cm and 0.3cm of S2] 
     {$n_4$};     
     
   \path (n1) edge[-latex] (n3);
   \path (n4) edge[-latex] (n5);
\end{scope}

\node at (5cm, 0cm) {$\Rightarrow$};

\begin{scope}[shift={(7cm, 0cm)}]
   \node (S1) [cloud,draw,cloud puffs=10,cloud puff arc=100, aspect=1.5,
          text opacity=1,
          minimum width=2.5cm,minimum height=2.2cm]
     {};
   \node [left=-1cm of S1] {$S_1$};
   
   \node (S2) [cloud,draw,cloud puffs=10,cloud puff arc=100, aspect=1.5,
          text opacity=1,
          minimum width=2.5cm,minimum height=2.2cm,
          right=0.2cm of S1]
     {};
   \node [right=-1cm of S2] {$S_2$};
   
   \node (n2p) [place-transition, draw, text opacity=1, fill=white,
     right=-1cm of S1] 
     {$n'_2$};     
   \node (n3) [place-transition, draw, text opacity=1, 
     below left=-0.5cm and -1cm of S2] 
     {$n_3$};     
   \node (n5) [place-transition, draw, text opacity=1, 
     above left=-0.5cm and -1cm of S2] 
     {$n_5$};     
   \node (n4) [place-transition, draw, text opacity=1,
     above left=0.2cm and 0.3cm of S2] 
     {$n_4$};     
     
   \path (n2p) edge[-latex] (n3);
   \path (n4) edge[-latex] (n5);
\end{scope}

\node at (11.25cm, 0cm) {$\Rightarrow$};

\begin{scope}[shift={(13cm, 0cm)}]
   \node (S1) [cloud,draw,cloud puffs=10,cloud puff arc=100, aspect=1.5,
          text opacity=1,
          minimum width=2.5cm,minimum height=2.2cm]
     at (0cm, 0cm)     
     {};
   \node [left=-1cm of S1] {$S_1$};

   \node (S2) [cloud,draw,cloud puffs=10,cloud puff arc=100, aspect=1.5,
          text opacity=1,
          minimum width=2.5cm,minimum height=2.2cm,
          right=0.2cm of S1]
     {};
   \node [right=-1cm of S2] {$S_2$};
   
   \node (n2p) [place-transition, draw, text opacity=1, fill=white,
     right=-1cm of S1] 
     {$n'_2$};     
 
   \node (n3) [place-transition, draw, text opacity=1, 
     below left=-0.5cm and -1cm of S2] 
     {$n_3$};     

   \node (n5) [place-transition, draw, text opacity=1, 
     above left=-0.5cm and -1cm of S2] 
     {$n_5$};     
   \node (n4) [place-transition, draw, text opacity=1,
     above left=0.2cm and 0.3cm of S2] 
     {$n_4$};     
     
   \path (n2p) edge[-latex] (n3);
   \path (n4) edge[-latex] (n5);
   \path (n2p) edge[-latex] (n5);
   \path (n4) edge[-latex] (n3);  
\end{scope}

\node at (17.5cm, 0cm) {$\Rightarrow$};

\begin{scope}[shift={(18.5cm, 0cm)}]
   \node (n1) [place-transition, draw, text opacity=1, fill=white] 
     at (0cm, 0cm) {$n_1$};     
   \node [below=0cm of n1] {$=n_2$};

   \node (S2) [cloud,draw,cloud puffs=10,cloud puff arc=100, aspect=1.5,
          text opacity=1,
          minimum width=2.5cm,minimum height=2.2cm,
          right=0.2cm of n1]
     {};
   \node [right=-1cm of S2] {$S_2$};
   
   \node (n3) [place-transition, draw, text opacity=1, 
     below left=-0.5cm and -1cm of S2] 
     {$n_3$};     
   \node (n5) [place-transition, draw, text opacity=1, 
     above left=-0.5cm and -1cm of S2] 
     {$n_5$};     
   \node (n4) [place-transition, draw, text opacity=1,
     above left=0.2cm and 0.3cm of S2] 
     {$n_4$};     
     
   \path (n1) edge[-latex] (n3);
   \path (n4) edge[-latex] (n5);
   \path (n1) edge[-latex] (n5);
   \path (n4) edge[-latex] (n3);   
\end{scope}

\end{scope} % end case (a), first line

\begin{scope}[shift={(0cm, -3cm)}] % start case (b), second line

\node at (0cm, 0cm) {(b)};

\begin{scope}[shift={(1cm, 0cm)}]
   \node (n1) [place-transition, draw, text opacity=1, fill=white] 
     at (0cm, 0cm) {$n_1$};

   \node (S2) [cloud,draw,cloud puffs=10,cloud puff arc=100, aspect=1.5,
          text opacity=1,
          minimum width=2.5cm,minimum height=2.2cm,
          right=0.2cm of n1]
     {};
   \node [right=-1cm of S2] {$S_2$};
   
  \node (n2) [place-transition, draw, text opacity=1, fill=white,
     below left=0.2cm and 0.3cm of S2] 
     {$n_2$};     
  \node (n3) [place-transition, draw, text opacity=1, 
     below left=-0.5cm and -1cm of S2] 
     {$n_3$};     
   \node (n5) [place-transition, draw, text opacity=1, 
     above left=-0.5cm and -1cm of S2] 
     {$n_5$};     

   \node (n4) [place-transition, draw, text opacity=1,
     above left=0.2cm and 0.3cm of S2] 
     {$n_4$};     
     
   \path (n2) edge[-latex] (n3);
   \path (n4) edge[-latex] (n5);
\end{scope}

\node at (5cm, 0cm) {$\Rightarrow$};

\begin{scope}[shift={(7cm, 0cm)}]
   \node (S1) [cloud,draw,cloud puffs=10,cloud puff arc=100, aspect=1.5,
          text opacity=1,
          minimum width=2.5cm,minimum height=2.2cm]
     at (0cm, 0cm) 
     {};
   \node [left=-1cm of S1] {$S_1$};
   
   \node (S2) [cloud,draw,cloud puffs=10,cloud puff arc=100, aspect=1.5,
          text opacity=1,
          minimum width=2.5cm,minimum height=2.2cm,
          right=0.2cm of S1]
     {};
   \node [right=-1cm of S2] {$S_2$};
   
  \node (n2) [place-transition, draw, text opacity=1, fill=white,
     below left=0.2cm and 0.3cm of S2] 
     {$n_2$};     
  \node (n3) [place-transition, draw, text opacity=1, 
     below left=-0.5cm and -1cm of S2] 
     {$n_3$};     

   \node (n5) [place-transition, draw, text opacity=1, 
     above left=-0.5cm and -1cm of S2] 
     {$n_5$};     

   \node (n4) [place-transition, draw, text opacity=1,
     above left=0.2cm and 0.3cm of S2] 
     {$n_4$};     
     
   \path (n2) edge[-latex] (n3);
   \path (n4) edge[-latex] (n5);
\end{scope}

\node at (11.25cm, 0cm) {$\Rightarrow$};

\begin{scope}[shift={(13cm, 0cm)}]
   \node (S1) [cloud,draw,cloud puffs=10,cloud puff arc=100, aspect=1.5,
          text opacity=1,
          minimum width=2.5cm,minimum height=2.2cm]
     at (0cm, 0cm) 
     {};
   \node [left=-1cm of S1] {$S_1$};
   
   \node (S2) [cloud,draw,cloud puffs=10,cloud puff arc=100, aspect=1.5,
          text opacity=1,
          minimum width=2.5cm,minimum height=2.2cm,
          right=0.2cm of S1]
     {};
   \node [right=-1cm of S2] {$S_2$};
   
  \node (n2) [place-transition, draw, text opacity=1, fill=white,
     below left=0.2cm and 0.3cm of S2] 
     {$n_2$};     
  \node (n3) [place-transition, draw, text opacity=1, 
     below left=-0.5cm and -1cm of S2] 
     {$n_3$};     
 
   \node (n5) [place-transition, draw, text opacity=1, 
     above left=-0.5cm and -1cm of S2] 
     {$n_5$};     

   \node (n4) [place-transition, draw, text opacity=1,
     above left=0.2cm and 0.3cm of S2] 
     {$n_4$};     
     
   \path (n2) edge[-latex] (n3);
   \path (n4) edge[-latex] (n5);
   \path (n2) edge[-latex] (n5);
   \path (n4) edge[-latex] (n3);  
\end{scope}

\node at (17.5cm, 0cm) {$\Rightarrow$};

\begin{scope}[shift={(18.5cm, 0cm)}]
   \node (n1) [place-transition, draw, text opacity=1, fill=white] 
     at (0cm, 0cm) {$n_1$};     

   \node (S2) [cloud,draw,cloud puffs=10,cloud puff arc=100, aspect=1.5,
          text opacity=1,
          minimum width=2.5cm,minimum height=2.2cm,
          right=0.2cm of n1]
     {};
   \node [right=-1cm of S2] {$S_2$};
   
  \node (n2) [place-transition, draw, text opacity=1, fill=white,
     below left=0.2cm and 0.3cm of S2] 
     {$n_2$};      

  \node (n3) [place-transition, draw, text opacity=1, 
     below left=-0.5cm and -1cm of S2] 
     {$n_3$};     

   \node (n5) [place-transition, draw, text opacity=1, 
     above left=-0.5cm and -1cm of S2] 
     {$n_5$};     

   \node (n4) [place-transition, draw, text opacity=1,
     above left=0.2cm and 0.3cm of S2] 
     {$n_4$};     
     
   \path (n2) edge[-latex] (n3);
   \path (n4) edge[-latex] (n5);
   \path (n2) edge[-latex] (n5);
   \path (n4) edge[-latex] (n3);   
\end{scope}

\end{scope} % end case (b), second line

\end{tikzpicture}
}
\end{center}
\caption{\label{fig:well-nested-A}Illustration of proof of preservation of well-nestedness in case (A)}
\end{figure}

\begin{lem}[commutativity of contraction]
\label{lem:commutativity-A}
Let $M''_1$ be the result of first contracting $S_1$ into node $n_1$ resulting in $M'_1$ followed by contracting $S_2$ into node $n_2$, and let $M''_2$ be the result of first contracting $S_2$ into node $n_2$ resulting in $M'_2$ followed by contracting $S_1$ into node $n_1$. Then $M''_1 = M''_2$.
\end{lem}

\begin{proof}
Clearly $M''_1$ and $M''_2$ will have the same set of nodes, namely $(N \setminus (S_1 \cup S_2)) \cup \{ n_1, n_2 \}$. The types of the nodes in $M$ will not change, and the types of $n_1$ and $n_2$ are also identical in $M''_1$ and $M''_2$ since they are determined by the I/O type of $M[S_1]$ and $M[S_2]$ respectively, and this is not changed for $M[S_1]$ by the contraction of $M[S_2]$ and vice versa.

That the edges are also the same in $M''_1$ and $M''_2$ can be shown as follows. For edges in these nets between nodes in $M$ it is clear that they are the same since contractions preserve these. We therefore consider the remaining three types of edges in these nets: (i) edges between $n_1$ and $n_2$, (ii) edges between $n_1$ and nodes in $M$ and (iii) edges between $n_2$ and nodes in $M$.
\begin{enumerate}[(i)]

\item For both $M''_1$ and $M''_2$ it holds that they contain an edge $(n_1, n_2)$ iff there is an edge in $M$ between an output node of $M[S_1]$ and an input node of $M[S_2]$ in $M$. So it holds that $M''_1$ contains the edge $(n_1, n_2)$ iff $M''_2$ contains it. By symmetry the same holds for the edge $(n_2, n_1)$.

\item let $n_3$ be a node in $M$. For both $M''_1$ and $M''_2$ it holds that they contain the edge $(n_1, n_3)$ iff there is an edge in $M$ between an output node of $M[S_1]$ and $n_3$. So it holds that $M''_1$ contains the edge $(n_1, n_3)$ iff $M''_2$ contains it.  By a symmetrical argument the same holds for the edge $(n_3, n_1)$.

\item The argument is similar to that of the previous case, but with $S_1$ and $n_1$ replaced with $S_2$ and $n_2$, respectively.

\end{enumerate}
\end{proof}

\subsection{Overlapping, but not nested, subnets with different I/O types}

In this subsection we consider case (B) in Figure~\ref{fig:cases-of-proof}.  Therefore, we assume that $S_1$ and $S_2$ are contractible subnets in $M$ such that (i) they share nodes but not so that one is entirely nested inside the other and (ii) the WF nets $M[S_1]$ and $M[S_2]$ different I/O types. Let us also assume that $M'$ is the result of contracting $S_1$ into the node $n_1$ and $M''$ the result of contracting $S_2$ into the node $n_2$. Recall that we need to show that (1) after contraction of $S_1$ into $n_1$, the subnet $S_2 \setminus S_1$ is contractible and (2) if we contract first $S_1$ into $n_1$ and then $S_2 \setminus S_1$ into $n_2$, the result is the same as when we contract first $S_2$ into $n_2$ and then $S_1 \setminus S_2$ into $n_1$. 

The proof of (1) will be separated into smaller lemmas. Recall that a subnet $S$ is contractible in $M$ if the following two properties hold: (C1) $S$ is well-nested in $M$ and (C2) $M[S]$ is of a basic AND-OR type. There will be a lemma that shows (C1) for $M'[S_2 \setminus S_1]$. The property (C2) will be shown by a set of lemmas that show that the defining properties of basic AND-OR nets are inherited by $M'[S_2 \setminus S_1]$ if they hold for $M[S_2]$. These properties are:
\begin{enumerate}[(i)]
  \item I/O consistency,
  \item the I/O type,
  \item AND property,
  \item OR property,
  \item acyclicity and
  \item one-input one-output property.
\end{enumerate}
Observe that (i) implies that $M'[S_2 \setminus S_1]$ is a WF net, as follows from Theorem~\ref{thm:subnet-wf-char}, which implies that every I/O consistent subnet of a WF net is a WF net, and the fact $M'$ is indeed a WF, as follows from Theorem~\ref{thm:contraction-correctness}, which states that the result of a contraction in a WF net is again a WF net. The inheritance of the properties (ii) through (vi) by $M'[S_2 \setminus S_1]$ from $M[S_2]$ shows that the WF net $M'[S_2 \setminus S_1]$ is of the same basic AND-OR type as $M[S_2]$. We now present the lemmas that state that well-nestedness and each of the defining properties of AND-OR nets are preserved. Note that for properties (i) and (ii) we have only one lemma which states that the the type of the input nodes and of the output nodes is preserved, and therefore implies that both (i) and (ii) are preserved.  For properties (iii) and (iv) we also have one lemma that proves the preservation of each of the properties.

\begin{lem}[preservation of well-nestedness] \label{lem:well-nested-B}
The subnet $S_2 \setminus S_1$ is well-nested in $M'$.
\end{lem}

\begin{proof}
We first show the well-nestedness for input nodes of $M'[S_2 \setminus S_1]$. Since $S_2$ is an internal subnet, a node in $S_2 \setminus S_1$ is an input node of $M'[S_2 \setminus S_1]$ iff there is in $M'$ an edge that ends in that node and starts from a node not in $S_2 \setminus S_1$. So we need to show that if there are two edges $(n_2, n_3)$ and $(n_4, n_5)$ in $M'$ such that $n_2, n_4 \not\in S_2 \setminus S_1$ and $n_3, n_5 \in S_2 \setminus S_1$, then there are edges $(n_2, n_5)$ and $(n_4, n_3)$ in $M'$. Let us assume that such  $(n_2, n_3)$ and $(n_4, n_5)$ exist. We consider the following four cases: (i) both $n_2$ and $n_4$ are equal to $n_1$, (ii)  only one of $n_2$ and $n_4$ is equal to $n_1$, (iii) neither $n_2$ nor $n_4$ are equal to $n_1$.
\begin{enumerate}[(i)]

\item Since $n_2 = n_1 = n_4$, it holds that $(n_4, n_5) = (n_1, n_5) = (n_2, n_5)$ and that $(n_2, n_3) = (n_1, n_3) = (n_4, n_3)$. So it indeed follows that $(n_2, n_5)$ and $(n_4, n_3)$ exist in $M'$.

\item We can assume, without loss of generality, that $n_2 = n_1$ and $n_4 \neq n_1$. The reasoning that follows is illustrated in Figure~\ref{fig:well-nested-B} (a). Then, by definition of contraction, there must be an output node $n_6$ of $M[S_1]$ and an edge $(n_6, n_3)$ in $M$. We now distinguish two sub-cases, namely that (1) all output nodes of $M[S_1]$ are in $S_1 \cap S_2$ and (2) there is at least one output nodes of $M[S_1]$ in $S_1 \setminus S_2$. We consider the two cases.

\begin{enumerate}[(1)]

\item If all output nodes of  $M[S_1]$ are in $S_1 \cap S_2$, then so must $n_6$. Moreover, it also then holds for all nodes in $S_1 \setminus S_2$ that there is a path in $M[S_1]$ from them to a node in $S_1 \cap S_2$. Take one such path, and take the first edge $(n_7, n_8)$ in that path where $n_7 \in S_1 \setminus S_2$ and $n_8 \in S_1 \cap S_2$. Note that such an edge must exist if the path goes from $S_1 \setminus S_2$ to $S_1 \cap S_2$ and contains only nodes from $S_1$. Then, by the well-nestedness of $S_2$ in $M$, there is an edge $(n_4, n_8)$ in $M$, since $n_8$ and $n_5$ are both input nodes of $M[N_2]$. Then, it follows that $n_8$ is an input node of both $M[S_1]$ and $M[S_2]$, because of the incoming edge $(n_4, n_8)$. This leads to a contradiction, since it is assume that the I/O types of $M[S_1]$ and $M[S_2]$ are different, and so they cannot share input nodes.

\item We can assume that $n_6$ that is the output node of $M[S_1]$ in $S_1 \setminus S_2$. Because of the well-nestedness of $M[S_1]$ there is an edge $(n_6, n_3)$ in $M$. It then follows from the well-nestedness of $M[S_2]$ that there are also edges $(n_6, n_5)$ and $(n_4, n_3)$ in $M$. By the definition of contraction it then follows that there are edges $(n_1, n_5)$ and $(n_4, n_3)$ in $M'$. Since $n_2 = n_1$, there is then also an edge $(n_2, n_5)$ in $M'$.

\end{enumerate}

\item So we assume that $n_2 \neq n_1$ and $n_4 \neq n_1$. The reasoning that follows is illustrated in Figure~\ref{fig:well-nested-B} (b). It follows there are edges $(n_2, n_3)$ and $(n_4, n_5)$ in $M$. Since $S_2$ is well-nested, there are edges $(n_2, n_5)$ and $(n_4, n_3)$ in $M$. Since $n_2, n_3, n_4, n_5 \not\in S_1$, such edges are also in $M'$.

\end{enumerate}  
The proof for well-nestedness of the output nodes of $M'[S_2 \setminus S_1]$ is similar to that of the input nodes, except that the direction of the edges is reversed and the set of input nodes is replaced with the set of output nodes.
\end{proof}

\begin{figure}[htb]
\begin{center}
\resizebox{\textwidth}{!}{%
\begin{tikzpicture}
    \tikzstyle{transition} = [rectangle,draw,minimum width=0.55cm, minimum height=0.55cm,fill=white]
    \tikzstyle{place} = [circle,draw,minimum width=0.55cm, minimum height=0.55cm,fill=white,inner sep=0.08cm]
    \tikzstyle{place-transition} = [rounded rectangle,draw,minimum width=1cm, minimum height=0.5cm,fill=white,inner sep=0.08cm]

\begin{scope}[shift={(0cm, -8cm)}] % second case, line (b)

\begin{scope}[shift={(1cm, 0cm)}]

   \node (S1) [cloud,cloud puffs=15,cloud puff arc=100, aspect=1.5,
          text opacity=1,
          minimum width=3cm,minimum height=3cm]
    {};

   \node (S2) [cloud,draw,cloud puffs=15,cloud puff arc=100, aspect=1.5,
          text opacity=1,
          minimum width=3cm,minimum height=3cm,
          right=-1.5cm of S1]
    {};
   \node [right=-1.5cm of S2] {$S_2 \setminus S_1$};

  % erase S1 from S2
   \node (ES1) [cloud, draw, fill=white, color=white, cloud puffs=15,cloud puff arc=100, aspect=1.5,
          text opacity=1,
          minimum width=3cm,minimum height=3cm] at (S1.center)
    {};
    
    % put back erased edge of S2
  \begin{scope}[]
    \pgftransformshift{\pgfpointanchor{S2}{center}}
    \pgfset{cloud puffs=15,cloud puff arc=100, aspect=1.5,
       minimum width=3cm,minimum height=3cm}
    \pgfnode{cloud}{center}{}{nodename}{\pgfusepath{clip}}

   \node [draw,cloud,cloud puffs=15,cloud puff arc=100, aspect=1.5,
       minimum width=3cm,minimum height=3cm] 
     at (S1.center) {};
  \end{scope}
     
   \node (n1) [place-transition, right=-1.5cm of S1] {$n_1$};  
   \node (n3) [place-transition, above right=-0.6cm and -0.6cm of S2] {$n_3$};  
   \node (n5) [place-transition, below right=-0.6cm and -0.6cm of S2] {$n_5$};  

   \node (n4) [place-transition, below right=0.2cm and 0.2cm of S2] {$n_4$};  
   \node (n2) [place-transition, above right=0.2cm and 0.2cm of S2] {$n_2$};  
   
   \path (n2) edge[-latex] (n3);
   \path (n4) edge[-latex] (n5);

\end{scope}

\node at (0cm, 0cm) {(b)};

\node at (5cm, 0cm) {\Large $\Rightarrow$};

\begin{scope}[shift={(7cm, 0cm)}]

   \node (S1) [cloud,draw,cloud puffs=15,cloud puff arc=100, aspect=1.5,
          text opacity=1,
          minimum width=3cm,minimum height=3cm]
    {};
   \node [left=-1cm of S1] {$S_1$};

   \node (S2) [cloud,draw,cloud puffs=15,cloud puff arc=100, aspect=1.5,
          text opacity=1,
          minimum width=3cm,minimum height=3cm,
          right=-1.5cm of S1]
    {};
   \node [right=-1cm of S2] {$S_2$};
     
   \node (n3) [place-transition, above right=-0.6cm and -0.6cm of S2] {$n_3$};  
   \node (n5) [place-transition, below right=-0.6cm and -0.6cm of S2] {$n_5$};  

   \node (n4) [place-transition, below right=0.2cm and 0.2cm of S2] {$n_4$};  
   \node (n2) [place-transition, above right=0.2cm and 0.2cm of S2] {$n_2$};  
   
   \path (n2) edge[-latex] (n3);
   \path (n4) edge[-latex] (n5);

\end{scope}

\node at (11cm, 0cm) {\Large $\Rightarrow$};

\begin{scope}[shift={(13cm, 0cm)}]

   \node (S1) [cloud,draw,cloud puffs=15,cloud puff arc=100, aspect=1.5,
          text opacity=1,
          minimum width=3cm,minimum height=3cm]
    {};
   \node [left=-1cm of S1] {$S_1$};

   \node (S2) [cloud,draw,cloud puffs=15,cloud puff arc=100, aspect=1.5,
          text opacity=1,
          minimum width=3cm,minimum height=3cm,
          right=-1.5cm of S1]
    {};
   \node [right=-1cm of S2] {$S_2$};
     
   \node (n3) [place-transition, above right=-0.6cm and -0.6cm of S2] {$n_3$};  
   \node (n5) [place-transition, below right=-0.6cm and -0.6cm of S2] {$n_5$};  

   \node (n4) [place-transition, below right=0.2cm and 0.2cm of S2] {$n_4$};  
   \node (n2) [place-transition, above right=0.2cm and 0.2cm of S2] {$n_2$};  
   
   \path (n2) edge[-latex] (n3);
   \path (n4) edge[-latex] (n5);
   \path (n2) edge[-latex, bend left=45] (n5);
   \path (n4) edge[-latex, bend right=45] (n3);

\end{scope}

\begin{scope}[shift={(18cm, 0cm)}]

   \node (S1) [cloud,cloud puffs=15,cloud puff arc=100, aspect=1.5,
          text opacity=1,
          minimum width=3cm,minimum height=3cm]
    {};

   \node (S2) [cloud,draw,cloud puffs=15,cloud puff arc=100, aspect=1.5,
          text opacity=1,
          minimum width=3cm,minimum height=3cm,
          right=-1.5cm of S1]
    {};
   \node [right=-1.5cm of S2] {$S_2 \setminus S_1$};

  % erase S1 from S2
   \node (ES1) [cloud, draw, fill=white, color=white, cloud puffs=15,cloud puff arc=100, aspect=1.5,
          text opacity=1,
          minimum width=3cm,minimum height=3cm] at (S1.center)
    {};
    
    % put back erased edge of S2
  \begin{scope}[]
    \pgftransformshift{\pgfpointanchor{S2}{center}}
    \pgfset{cloud puffs=15,cloud puff arc=100, aspect=1.5,
       minimum width=3cm,minimum height=3cm}
    \pgfnode{cloud}{center}{}{nodename}{\pgfusepath{clip}}

   \node [draw,cloud,cloud puffs=15,cloud puff arc=100, aspect=1.5,
       minimum width=3cm,minimum height=3cm] 
     at (S1.center) {};
  \end{scope}
     
   \node (n1) [place-transition, right=-1.5cm of S1] {$n_1$};  
   \node (n3) [place-transition, above right=-0.6cm and -0.6cm of S2] {$n_3$};  
   \node (n5) [place-transition, below right=-0.6cm and -0.6cm of S2] {$n_5$};  

   \node (n4) [place-transition, below right=0.2cm and 0.2cm of S2] {$n_4$};  
   \node (n2) [place-transition, above right=0.2cm and 0.2cm of S2] {$n_2$};  
   
   \path (n2) edge[-latex] (n3);
   \path (n4) edge[-latex] (n5);
   \path (n2) edge[-latex, bend left=45] (n5);
   \path (n4) edge[-latex, bend right=45] (n3);

\end{scope}

\node at (17cm, 0cm) {\Large $\Rightarrow$};

\end{scope} % end of second line

\begin{scope}[shift={(0cm, 0cm)}] % first case, line (a)

\begin{scope}[shift={(1cm, 0cm)}]

   \node (S1) [cloud,cloud puffs=15,cloud puff arc=100, aspect=1.5,
          text opacity=1,
          minimum width=3cm,minimum height=3cm]
    {};

   \node (S2) [cloud,draw,cloud puffs=15,cloud puff arc=100, aspect=1.5,
          text opacity=1,
          minimum width=3cm,minimum height=3cm,
          right=-1.5cm of S1]
    {};
   \node [right=-1.5cm of S2] {$S_2 \setminus S_1$};

  % erase S1 from S2
   \node (ES1) [cloud, draw, fill=white, color=white, cloud puffs=15,cloud puff arc=100, aspect=1.5,
          text opacity=1,
          minimum width=3cm,minimum height=3cm] at (S1.center)
    {};
    
    % put back erased edge of S2
  \begin{scope}[]
    \pgftransformshift{\pgfpointanchor{S2}{center}}
    \pgfset{cloud puffs=15,cloud puff arc=100, aspect=1.5,
       minimum width=3cm,minimum height=3cm}
    \pgfnode{cloud}{center}{}{nodename}{\pgfusepath{clip}}

   \node [draw,cloud,cloud puffs=15,cloud puff arc=100, aspect=1.5,
       minimum width=3cm,minimum height=3cm] 
     at (S1.center) {};
  \end{scope}
     
   \node (n1) [place-transition, right=-1.5cm of S1] {$n_1$};  
   \node [below=0cm of n1] {$=n_2$};

   \node (n3) [place-transition, above right=-0.6cm and -0.6cm of S2] {$n_3$};  
   \node (n5) [place-transition, below right=-0.6cm and -0.6cm of S2] {$n_5$};  

   \node (n4) [place-transition, below right=0.2cm and 0.2cm of S2] {$n_4$};  
   
   \path (n1) edge[-latex] (n3);
   \path (n4) edge[-latex] (n5);

\end{scope}

\node at (0cm, 0cm) {(a)};

\node [rotate=40] at (4.9cm, 1.5cm) {\Large $\Rightarrow$};

\begin{scope}[shift={(7cm, -2.5cm)}]

   \node (S1) [cloud,draw,cloud puffs=15,cloud puff arc=100, aspect=1.5,
          text opacity=1,
          minimum width=3cm,minimum height=3cm]
    {};
   \node [left=-1cm of S1] {$S_1$};

   \node (S2) [cloud,draw,cloud puffs=15,cloud puff arc=100, aspect=1.5,
          text opacity=1,
          minimum width=3cm,minimum height=3cm,
          right=-1.5cm of S1]
    {};
   \node [right=-1cm of S2] {$S_2$};
     
   \node (n6) [place-transition, above left=-0.6cm and -0.6cm of S1] {$n_6$};  
   \node (n3) [place-transition, above right=-0.6cm and -0.6cm of S2] {$n_3$};  
   \node (n5) [place-transition, below right=-0.6cm and -0.6cm of S2] {$n_5$};  

   \node (n4) [place-transition, below right=0.2cm and 0.2cm of S2] {$n_4$}; 

   \path (n6) edge[-latex] (n3);
   \path (n4) edge[-latex] (n5);
   
   \node [below right=0.7 cm and -3 cm of S1] {Not all output nodes of $M[S_1]$ in $S_1 \cap S_2$};
\end{scope}

\node at (11cm, -2.5cm) {\Large $\Rightarrow$};

\begin{scope}[shift={(13cm, -2.5cm)}]

   \node (S1) [cloud,draw,cloud puffs=15,cloud puff arc=100, aspect=1.5,
          text opacity=1,
          minimum width=3cm,minimum height=3cm]
    {};
   \node [left=-1cm of S1] {$S_1$};

   \node (S2) [cloud,draw,cloud puffs=15,cloud puff arc=100, aspect=1.5,
          text opacity=1,
          minimum width=3cm,minimum height=3cm,
          right=-1.5cm of S1]
    {};
   \node [right=-1cm of S2] {$S_2$};
     
   \node (n6) [place-transition, above left=-0.6cm and -0.6cm of S1] {$n_6$};  
   \node (n3) [place-transition, above right=-0.6cm and -0.6cm of S2] {$n_3$};  
   \node (n5) [place-transition, below right=-0.6cm and -0.6cm of S2] {$n_5$};  

   \node (n4) [place-transition, below right=0.2cm and 0.2cm of S2] {$n_4$};  

   \path (n6) edge[-latex] (n3);
   \path (n4) edge[-latex] (n5);
   \path (n6) edge[-latex] (n5);
   \path (n4) edge[-latex, bend right=20] (n3);   
\end{scope}

\begin{scope}[shift={(18cm, -2.5cm)}]

   \node (S1) [cloud,cloud puffs=15,cloud puff arc=100, aspect=1.5,
          text opacity=1,
          minimum width=3cm,minimum height=3cm]
    {};

   \node (S2) [cloud,draw,cloud puffs=15,cloud puff arc=100, aspect=1.5,
          text opacity=1,
          minimum width=3cm,minimum height=3cm,
          right=-1.5cm of S1]
    {};
   \node [right=-1.5cm of S2] {$S_2 \setminus S_1$};

  % erase S1 from S2
   \node (ES1) [cloud, draw, fill=white, color=white, cloud puffs=15,cloud puff arc=100, aspect=1.5,
          text opacity=1,
          minimum width=3cm,minimum height=3cm] at (S1.center)
    {};
    
    % put back erased edge of S2
  \begin{scope}[]
    \pgftransformshift{\pgfpointanchor{S2}{center}}
    \pgfset{cloud puffs=15,cloud puff arc=100, aspect=1.5,
       minimum width=3cm,minimum height=3cm}
    \pgfnode{cloud}{center}{}{nodename}{\pgfusepath{clip}}

   \node [draw,cloud,cloud puffs=15,cloud puff arc=100, aspect=1.5,
       minimum width=3cm,minimum height=3cm] 
     at (S1.center) {};
  \end{scope}
     
   \node (n1) [place-transition, right=-1.5cm of S1] {$n_1$};  
   \node [below=0cm of n1] {$=n_2$};

   \node (n3) [place-transition, above right=-0.6cm and -0.6cm of S2] {$n_3$};  
   \node (n5) [place-transition, below right=-0.6cm and -0.6cm of S2] {$n_5$};  

   \node (n4) [place-transition, below right=0.2cm and 0.2cm of S2] {$n_4$};  
   
   \path (n1) edge[-latex] (n3);
   \path (n4) edge[-latex] (n5);
   \path (n1) edge[-latex] (n5);
   \path (n4) edge[-latex, bend right=30] (n3);

\end{scope}

\node at (17cm, -2.5cm) {\Large $\Rightarrow$};

\node [rotate=-40] at (4.9cm, -1.5cm) {\Large $\Rightarrow$};

\begin{scope}[shift={(7cm, 2.5cm)}]

   \node (S1) [cloud,draw,cloud puffs=15,cloud puff arc=100, aspect=1.5,
          text opacity=1,
          minimum width=3cm,minimum height=3cm]
    {};
   \node [left=-1cm of S1] {$S_1$};

   \node (S2) [cloud,draw,cloud puffs=15,cloud puff arc=100, aspect=1.5,
          text opacity=1,
          minimum width=3cm,minimum height=3cm,
          right=-1.5cm of S1]
    {};
   \node [right=-1cm of S2] {$S_2$};
     
   \node (n6) [place-transition, above right=-0.9cm and -0.5cm of S1] {$n_6$};  

   \node (n3) [place-transition, above right=-0.6cm and -0.6cm of S2] {$n_3$};  
   \node (n5) [place-transition, below right=-0.6cm and -0.6cm of S2] {$n_5$};  

   \node (n4) [place-transition, below right=0.2cm and 0.2cm of S2] {$n_4$};  

   \path (n6) edge[-latex] (n3);
   \path (n4) edge[-latex] (n5);
   
   \node [below right=0.7 cm and -3 cm of S1] {All output nodes of $M[S_1]$ in $S_1 \cap S_2$};
\end{scope}

\node at (11cm, 2.5cm) {\Large $\Rightarrow$};

\begin{scope}[shift={(13cm, 2.5cm)}]

   \node (S1) [cloud,draw,cloud puffs=15,cloud puff arc=100, aspect=1.5,
          text opacity=1,
          minimum width=3cm,minimum height=3cm]
    {};
   \node [left=-1cm of S1] {$S_1$};

   \node (S2) [cloud,draw,cloud puffs=15,cloud puff arc=100, aspect=1.5,
          text opacity=1,
          minimum width=3cm,minimum height=3cm,
          right=-1.5cm of S1]
    {};
   \node [right=-1cm of S2] {$S_2$};
     
   \node (n6) [place-transition, above right=-0.9cm and -0.5cm of S1] {$n_6$};  

   \node (n3) [place-transition, above right=-0.6cm and -0.6cm of S2] {$n_3$};  
   \node (n5) [place-transition, below right=-0.6cm and -0.6cm of S2] {$n_5$};  

   \node (n4) [place-transition, below right=0.2cm and 0.2cm of S2] {$n_4$};  
   
   \node (n8) [place-transition, below right=-0.9cm and -0.5cm of S1] {$n_8$};  
   \node (n7) [place-transition, below left=-0.6cm and -0.6cm of S1] {$n_7$};  

   \path (n6) edge[-latex] (n3);
   \path (n4) edge[-latex] (n5);
   \path (n7) edge[-latex] (n8);
   
\end{scope}

\node at (17cm, 2.5cm) {\Large $\Rightarrow$};

\begin{scope}[shift={(19cm, 2.5cm)}]

   \node (S1) [cloud,draw,cloud puffs=15,cloud puff arc=100, aspect=1.5,
          text opacity=1,
          minimum width=3cm,minimum height=3cm]
    {};
   \node [left=-1cm of S1] {$S_1$};

   \node (S2) [cloud,draw,cloud puffs=15,cloud puff arc=100, aspect=1.5,
          text opacity=1,
          minimum width=3cm,minimum height=3cm,
          right=-1.5cm of S1]
    {};
   \node [right=-1cm of S2] {$S_2$};
   \node [right=1cm of S2] {\textbf{FALSE}};
     
   \node (n6) [place-transition, above right=-0.9cm and -0.5cm of S1] {$n_6$};  

   \node (n3) [place-transition, above right=-0.6cm and -0.6cm of S2] {$n_3$};  
   \node (n5) [place-transition, below right=-0.6cm and -0.6cm of S2] {$n_5$};  

   \node (n4) [place-transition, below right=0.2cm and 0.2cm of S2] {$n_4$};  
   
   \node (n8) [place-transition, below right=-0.9cm and -0.5cm of S1] {$n_8$};  
   \node (n7) [place-transition, below left=-0.6cm and -0.6cm of S1] {$n_7$};  

   \path (n6) edge[-latex] (n3);
   \path (n4) edge[-latex] (n5);
   \path (n7) edge[-latex] (n8);
   \path (n4) edge[-latex, bend left=20] (n8);   
   
\end{scope}

\end{scope} % end of first case, line (a)

\end{tikzpicture}
}
\end{center}
\caption{\label{fig:well-nested-B}Illustration of proof of preservation of well-nestedness in case (B)}
\end{figure}
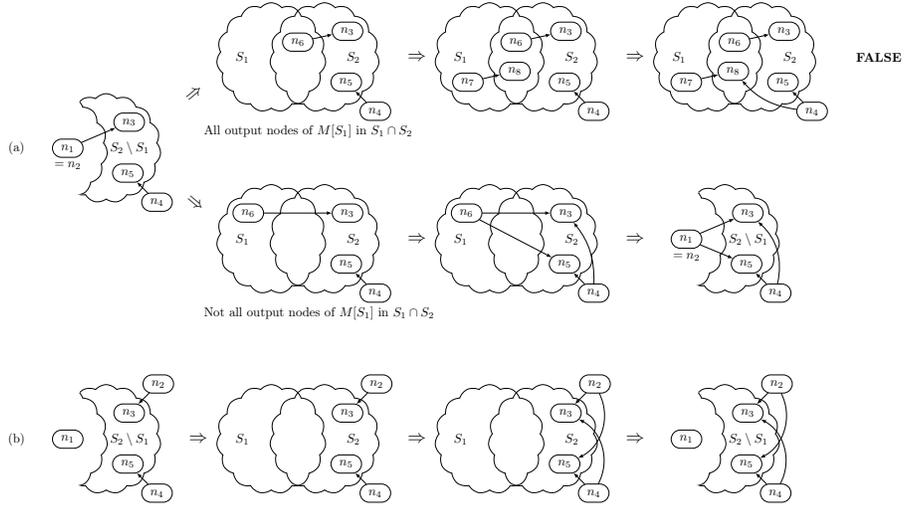

\begin{lem}[preservation of I/O type] \label{lem:IO-type-B}
The input nodes and output nodes of $M'[S_2 \setminus S_1]$ have the same type as the input nodes and output nodes of $M[S_2]$, respectively.
\end{lem}

\begin{proof}
We first consider input nodes. Assume $n_2$ is an input node of  $M'[S_2 \setminus S_1]$, then it will have in $M'$ an incoming edge $(n_3, n_2)$ with $n_3 \not\in S_2 \setminus S_1$. Then either (i) $n_3 = n_1$ or (ii) $n_3 \neq n_1$:
\begin{enumerate}[(i)]

  \item Then in $M[S_1]$ there is an output node $n_4$ and an edge $(n_4, n_2)$.  So the type of $n_2$ in $M$ is the opposite of the I/O type of $M[S_1]$. Since by assumption the I/O type of $M[S_1]$ is the opposite of the I/O type of $M[S_2]$, it follows that the type of $n_2$ in $M$ is the I/O type of $M[S_2]$, and since the type of a node is not changed by contraction, the type of $n_2$ in $M'$ is also the I/O type of $M[S_2]$.
  
  \item It follows that this edge exists in $M$ and therefore $n_2$ is an input node of $M[S_2]$. Since the contraction does not change the type of nodes, the type of $n_2$ in $M'$ is the same as the type of the input nodes of $M[S_2]$.

\end{enumerate}
The argument for output nodes is similar, except that edges are reversed and the set of input nodes is replaced with the set of output nodes.
\end{proof}

\begin{lem}[preservation of the AND and the OR property] \label{lem:AND-OR-B}
The net $M'[S_2 \setminus S_1]$ has the AND (OR) property if $M[S_2]$ has it.
\end{lem}

\begin{proof}
We first consider the AND property, and of this first the restriction on incoming edges. The proof proceeds by contradiction, so we start with assuming that $M[S_2]$ has the AND property and $M'[S_2 \setminus S_1]$ does not have the AND property. We will show that it then follows that $M[S_2]$ does \emph{not} have the AND property.

If $M'[S_2 \setminus S_1]$ does not have the AND property, then in $M'[S_2 \setminus S_1]$ there is a place $p$ such that (i) $p$ has at least two incoming edges and is not an input place or (ii) $p$ is an input place and has an incoming edge $(t, p)$. In case (i) these two edges already existed in $M$, and so $M[S_2]$ does not have the AND property. In case (ii) as an input place of an internally nested net $p$ has in $M'$ an incoming edge $(t', p)$ with $t' \not\in S_2 \setminus S_1$. It then holds that either (a) $t' \neq n_1$ or (b) $t' = n_1$. We consider these two cases:
  
    \begin{enumerate}[(a)]
    
      \item  Then $(t', p) $ exists in $M$ and $t' \not\in S_1 \cup S_2$. Since there was also an edge $(t, p)$ in $M'[S_2 \setminus S_1]$ with $t \in S_2 \setminus S_1$, which therefore also exists in $M[S_2]$, it follows that $p$ has in $M[S_2]$ at least two distinct incoming edges, and so $M[S_2]$ does not satisfy the AND property.
            
      \item Then there is in $M$ an edge $(t'', p)$ from an output transition $t''$ of $M[S_1]$ to $p$. It holds that (1) $t'' \in S_1 \cap S_2$ or (2) $t'' \in S_1 \setminus S_2$. If (1) then $(t'', p)$  exists in $M[S_2]$. However, since there is the edge $(t, p)$ in $M'[S_2 \setminus S_1]$ with $t \in S_2 \setminus S_1$, which necessarily also exists in $M[S_2]$, it follows that in $M[S_2]$ the place $p$ has two distinct incoming edges, and therefore does not have the AND property. If (2) then $p$ is an input node in $M[S_2]$, since it has an incoming edge in $M$ starting from a node outside $S_2$. However, as we just saw in case (1), there is also an edge $(t, p)$ in $M[S_2]$. So, in $M[S_2]$ the place $p$ is an input node with an incoming edge, and therefore $M[S_2]$ does not have the AND property.
      
    \end{enumerate}
    The proof for the restriction on the outgoing edges is similar, except that the direction of the edges is reversed and the sets of output places are replaced with the sets of input places, and vice versa.
    
    The proof for the preservation for the OR property is similar to that for the AND property, except that places are replaced with transitions and vice versa.
\end{proof}

\begin{lem}[preservation of acyclicity] \label{lem:acyclic-B}
The net $M'[S_2 \setminus S_1]$ is acyclic if $M[S_2]$ is acyclic.
\end{lem}

\begin{proof}
If $M'[S_2 \setminus S_1]$ contains a cycle, then the edges of this cycle are also present in $M[S_2]$, and so it would also contain a cycle.
\end{proof}

\begin{lem}[preservation of one-input one-output property] \label{lem:one-input-one-output-B}
The net $M'[S_2 \setminus S_1]$ is a one-input one-output net if $M[S_2]$ is a 11tAND or a 11pOR net.
\end{lem}

\begin{proof}
We start with proving the one-input property. Let $M[S_2]$ be a 11tAND or a 11pOR net. The proof proceeds by contradiction, so we assume that $M'[S_2 \setminus S_1]$ does not have the one-input property, i.e., there are two distinct nodes $n_2$ and $n_3$ in $M'[S_2 \setminus S_1]$ with edges $(n_4, n_2)$ and $(n_5, n_3)$ in $M'$ such that $n_4, n_5 \not\in S_2 \setminus S_1$. We then consider the following three cases for the nodes $n_4$ and $n_5$: (i) both are equal to $n_1$, (ii) one of them is equal to $n_1$ and (ii) neither is equal to $n_1$.
\begin{enumerate}[(i)]

  \item This case is illustrated in Figure~\ref{fig:one-input-one-output-B} (a). In this case there must be in $M[S_1]$ an output node $n_6$ with edges $(n_6, n_2)$ and $(n_6, n_3)$ in $M$. Note that both edges must exist, since $S_1$ is assumed to be well-nested in $M$. We consider the following two cases: (1) $n_6 \in S_1 \cap S_2$ and (2) $n_6 \in S_1 \setminus S_2$:
  
    \begin{enumerate}[(1)]
    
      \item The WF net $M'[S_2 \setminus S_1]$ is either a pWF net or a tWF net. In the first case, it follows that $n_2$ and $n_3$ are places, since they are input nodes of $M'[S_2 \setminus S_1]$ and $n_6$ is a transition. By Lemma~\ref{lem:IO-type-B}, the I/O type of $M'[S_2 \setminus S_1]$ is equal to that of $M[S_2]$, and so $M[S_2]$ must be a 11pOR net, since by assumption it is either a 11pOR net or a 11tAND net. But this is contradicted by the two distinct edges $(n_6, n_2)$ and $(n_6, n_3)$ in $M[S_2]$, which violate the OR property in regard to the transition $t_6$. Similarly, if we assume that $M'[S_2 \setminus S_1]$ is a tWF net, then we can derive that $M[S_2]$ is a 11tAND net, but at the same time the edges $(n_6, n_2)$ and $(n_6, n_3)$ violate the AND property.
      
      % JS: the FALSE is missing for this case on the drawing (also case a in the drawing could be moved a little up)
      % JH: it's a different FALSE, like in cases (b) and (c)
      
      \item In this case $n_2$ and $n_3$ have in $M$ incoming edges from outside $S_2$, and so are both input nodes of $M[S_2]$. However, this contradicts the assumption that $M[S_2]$ has the one-input property.
      
    \end{enumerate}

  \item This case is illustrated in Figure~\ref{fig:one-input-one-output-B} (b). Without loss of generality we can assume that $n_5 = n_1$ and $n_4 \neq n_1$. By Lemma~\ref{lem:well-nested-B} we know that $S_2 \setminus S_1$ is well-nested in $M'$, and so the edge $(n_4, n_2)$ implies the edge $(n_4, n_3)$ in $M'$. Since both edges will also exist in $M$, it follows that $n_2$ and $n_3$ are both input nodes of $M[S_2]$. However, this contradicts the assumption that $M[S_2]$ has the one-input property.
  
  \item This case is illustrated in Figure~\ref{fig:one-input-one-output-B} (c). It proceeds the same as in the previous case.
  
\end{enumerate}
 
The one-output property can be shown to be preserved in a similar way, but with the direction of the edges reversed and the sets of input places replaced by the set of output places and vice versa.
\end{proof}

\begin{figure}[htb]
\begin{center}
\resizebox{0.8\textwidth}{!}{%
\begin{tikzpicture}
    \tikzstyle{transition} = [rectangle,draw,minimum width=0.55cm, minimum height=0.55cm,fill=white]
    \tikzstyle{place} = [circle,draw,minimum width=0.55cm, minimum height=0.55cm,fill=white,inner sep=0.08cm]
    \tikzstyle{place-transition} = [rounded rectangle,draw,minimum width=1cm, minimum height=0.5cm,fill=white,inner sep=0.08cm]

\begin{scope}[shift={(0cm, -7cm)}] % second case, line (b)

\begin{scope}[shift={(1cm, -2cm)}]
   \node (S1) [cloud,cloud puffs=15,cloud puff arc=100, aspect=1.5,
          text opacity=1,
          minimum width=3cm,minimum height=3cm]
    {};

   \node (S2) [cloud,draw,cloud puffs=15,cloud puff arc=100, aspect=1.5,
          text opacity=1,
          minimum width=3cm,minimum height=3cm,
          right=-1.5cm of S1]
    {};
   \node [right=-1.5cm of S2] {$S_2 \setminus S_1$};

  % erase S1 from S2
   \node (ES1) [cloud, draw, fill=white, color=white, cloud puffs=15,cloud puff arc=100, aspect=1.5,
          text opacity=1,
          minimum width=3cm,minimum height=3cm] at (S1.center)
    {};
    
    % put back erased edge of S2
  \begin{scope}[]
    \pgftransformshift{\pgfpointanchor{S2}{center}}
    \pgfset{cloud puffs=15,cloud puff arc=100, aspect=1.5,
       minimum width=3cm,minimum height=3cm}
    \pgfnode{cloud}{center}{}{nodename}{\pgfusepath{clip}}

   \node [draw,cloud,cloud puffs=15,cloud puff arc=100, aspect=1.5,
       minimum width=3cm,minimum height=3cm] 
     at (S1.center) {};
  \end{scope}
     
   \node (n1) [place-transition, right=-1.5cm of S1] {$n_1$};  
   \node (n2) [place-transition, above right=-0.6cm and -0.6cm of S2] {$n_2$};  
   \node (n3) [place-transition, below right=-0.6cm and -0.6cm of S2] {$n_3$};  

   \node (n5) [place-transition, below right=0.2cm and 0.2cm of S2] {$n_5$};  
   \node (n4) [place-transition, above right=0.2cm and 0.2cm of S2] {$n_4$};  
   
   \path (n4) edge[-latex] (n2);
   \path (n5) edge[-latex] (n3);
\end{scope}

\begin{scope}[shift={(1cm, 2cm)}]
   \node (S1) [cloud,cloud puffs=15,cloud puff arc=100, aspect=1.5,
          text opacity=1,
          minimum width=3cm,minimum height=3cm]
    {};

   \node (S2) [cloud,draw,cloud puffs=15,cloud puff arc=100, aspect=1.5,
          text opacity=1,
          minimum width=3cm,minimum height=3cm,
          right=-1.5cm of S1]
    {};
   \node [right=-1.5cm of S2] {$S_2 \setminus S_1$};

  % erase S1 from S2
   \node (ES1) [cloud, draw, fill=white, color=white, cloud puffs=15,cloud puff arc=100, aspect=1.5,
          text opacity=1,
          minimum width=3cm,minimum height=3cm] at (S1.center)
    {};
    
    % put back erased edge of S2
  \begin{scope}[]
    \pgftransformshift{\pgfpointanchor{S2}{center}}
    \pgfset{cloud puffs=15,cloud puff arc=100, aspect=1.5,
       minimum width=3cm,minimum height=3cm}
    \pgfnode{cloud}{center}{}{nodename}{\pgfusepath{clip}}

   \node [draw,cloud,cloud puffs=15,cloud puff arc=100, aspect=1.5,
       minimum width=3cm,minimum height=3cm] 
     at (S1.center) {};
  \end{scope}
     
   \node (n1) [place-transition, right=-1.5cm of S1] {$n_1$};  
   \node [below=0cm of n1] {$=n_5$};

   \node (n2) [place-transition, above right=-0.6cm and -0.6cm of S2] {$n_2$};  
   \node (n3) [place-transition, below right=-0.6cm and -0.6cm of S2] {$n_3$};  

   \node (n4) [place-transition, above right=0.2cm and 0.2cm of S2] {$n_4$};  
   
   \path (n4) edge[-latex] (n2);
   \path (n1) edge[-latex] (n3);
\end{scope}

\node at (0cm, 2cm) {(b)};

\node at (0cm, -2cm) {(c)};

\node [rotate=-30] at (5.5cm, 1cm) {\Large $\Rightarrow$};

\node [rotate=30] at (5.5cm, -1cm) {\Large $\Rightarrow$};

\begin{scope}[shift={(7cm, 0cm)}]
   \node (S1) [cloud,cloud puffs=15,cloud puff arc=100, aspect=1.5,
          text opacity=1,
          minimum width=3cm,minimum height=3cm]
    {};

   \node (S2) [cloud,draw,cloud puffs=15,cloud puff arc=100, aspect=1.5,
          text opacity=1,
          minimum width=3cm,minimum height=3cm,
          right=-1.5cm of S1]
    {};
   \node [right=-1.5cm of S2] {$S_2 \setminus S_1$};

  % erase S1 from S2
   \node (ES1) [cloud, draw, fill=white, color=white, cloud puffs=15,cloud puff arc=100, aspect=1.5,
          text opacity=1,
          minimum width=3cm,minimum height=3cm] at (S1.center)
    {};
    
    % put back erased edge of S2
  \begin{scope}[]
    \pgftransformshift{\pgfpointanchor{S2}{center}}
    \pgfset{cloud puffs=15,cloud puff arc=100, aspect=1.5,
       minimum width=3cm,minimum height=3cm}
    \pgfnode{cloud}{center}{}{nodename}{\pgfusepath{clip}}

   \node [draw,cloud,cloud puffs=15,cloud puff arc=100, aspect=1.5,
       minimum width=3cm,minimum height=3cm] 
     at (S1.center) {};
  \end{scope}
     
   \node (n1) [place-transition, right=-1.5cm of S1] {$n_1$};  
   \node (n2) [place-transition, above right=-0.6cm and -0.6cm of S2] {$n_2$};  
   \node (n3) [place-transition, below right=-0.6cm and -0.6cm of S2] {$n_3$};  

   \node (n4) [place-transition, above right=0.2cm and 0.2cm of S2] {$n_4$};  
   
   \path (n4) edge[-latex, bend left=40] (n3);
   \path (n4) edge[-latex] (n2);
\end{scope}

\node at (11cm, 0cm) {\Large $\Rightarrow$};

\begin{scope}[shift={(13cm, 0cm)}]

   \node (S1) [cloud,draw,cloud puffs=15,cloud puff arc=100, aspect=1.5,
          text opacity=1,
          minimum width=3cm,minimum height=3cm]
    {};
   \node [left=-1cm of S1] {$S_1$};

   \node (S2) [cloud,draw,cloud puffs=15,cloud puff arc=100, aspect=1.5,
          text opacity=1,
          minimum width=3cm,minimum height=3cm,
          right=-1.5cm of S1]
    {};
   \node [right=-1cm of S2] {$S_2$};
     
   \node (n2) [place-transition, above right=-0.6cm and -0.6cm of S2] {$n_2$};  
   \node (n3) [place-transition, below right=-0.6cm and -0.6cm of S2] {$n_3$};  

   \node (n4) [place-transition, above right=0.2cm and 0.2cm of S2] {$n_4$};  
   
   \path (n4) edge[-latex, bend left=40] (n3);
   \path (n4) edge[-latex] (n2);
\end{scope}

\end{scope} % end of second line

\begin{scope}[shift={(0cm, 0cm)}] % first case, line (a)

\begin{scope}[shift={(1cm, 0cm)}]
   \node (S1) [cloud,cloud puffs=15,cloud puff arc=100, aspect=1.5,
          text opacity=1,
          minimum width=3cm,minimum height=3cm]
    {};

   \node (S2) [cloud,draw,cloud puffs=15,cloud puff arc=100, aspect=1.5,
          text opacity=1,
          minimum width=3cm,minimum height=3cm,
          right=-1.5cm of S1]
    {};
   \node [right=-1.5cm of S2] {$S_2 \setminus S_1$};

  % erase S1 from S2
   \node (ES1) [cloud, draw, fill=white, color=white, cloud puffs=15,cloud puff arc=100, aspect=1.5,
          text opacity=1,
          minimum width=3cm,minimum height=3cm] at (S1.center)
    {};
    
    % put back erased edge of S2
  \begin{scope}[]
    \pgftransformshift{\pgfpointanchor{S2}{center}}
    \pgfset{cloud puffs=15,cloud puff arc=100, aspect=1.5,
       minimum width=3cm,minimum height=3cm}
    \pgfnode{cloud}{center}{}{nodename}{\pgfusepath{clip}}

   \node [draw,cloud,cloud puffs=15,cloud puff arc=100, aspect=1.5,
       minimum width=3cm,minimum height=3cm] 
     at (S1.center) {};
  \end{scope}
     
   \node (n1) [place-transition, right=-1.5cm of S1] {$n_1$};  
   \node [below=0cm of n1] {$=n_4$};
   \node [below=0.3cm of n1] {$=n_5$};

   \node (n2) [place-transition, above right=-0.6cm and -0.6cm of S2] {$n_2$};  
   \node (n3) [place-transition, below right=-0.6cm and -0.6cm of S2] {$n_3$};  
   
   \path (n1) edge[-latex] (n2);
   \path (n1) edge[-latex] (n3);   
\end{scope}

\node at (0cm, 0cm) {(a)};

\node [rotate=40] at (4.9cm, 1.5cm) {\Large $\Rightarrow$};

\begin{scope}[shift={(7cm, -2cm)}]
   \node (S1) [cloud,draw,cloud puffs=15,cloud puff arc=100, aspect=1.5,
          text opacity=1,
          minimum width=3cm,minimum height=3cm]
    {};
   \node [left=-1cm of S1] {$S_1$};

   \node (S2) [cloud,draw,cloud puffs=15,cloud puff arc=100, aspect=1.5,
          text opacity=1,
          minimum width=3cm,minimum height=3cm,
          right=-1.5cm of S1]
    {};
   \node [right=-1cm of S2] {$S_2$};
     
   \node (n4) [place-transition, above left=-0.6cm and -0.6cm of S1] {$n_6$};  
   \node (n2) [place-transition, above right=-0.6cm and -0.6cm of S2] {$n_2$};  
   \node (n3) [place-transition, below right=-0.6cm and -0.6cm of S2] {$n_3$};  

   \path (n4) edge[-latex] (n2);
   \path (n4) edge[-latex] (n3);
 \end{scope}

\node [rotate=-40] at (4.9cm, -1.5cm) {\Large $\Rightarrow$};

\begin{scope}[shift={(7cm, 2cm)}]
   \node (S1) [cloud,draw,cloud puffs=15,cloud puff arc=100, aspect=1.5,
          text opacity=1,
          minimum width=3cm,minimum height=3cm]
    {};
   \node [left=-1cm of S1] {$S_1$};

   \node (S2) [cloud,draw,cloud puffs=15,cloud puff arc=100, aspect=1.5,
          text opacity=1,
          minimum width=3cm,minimum height=3cm,
          right=-1.5cm of S1]
    {};
   \node [right=-1cm of S2] {$S_2$};
     
   \node (n4) [place-transition, right=-1.2cm of S1] {$n_6$};  
   \node (n2) [place-transition, above right=-0.6cm and -0.6cm of S2] {$n_2$};  
   \node (n3) [place-transition, below right=-0.6cm and -0.6cm of S2] {$n_3$};  

   \path (n4) edge[-latex] (n2);
   \path (n4) edge[-latex] (n3);
\end{scope}

\node [rotate=30] at (11cm, 2.9cm) {\Large $\Rightarrow$};

\begin{scope}[shift={(13cm, 3.9cm)}]
   \node (S1) [cloud,draw,cloud puffs=15,cloud puff arc=100, aspect=1.5,
          text opacity=1,
          minimum width=3cm,minimum height=3cm]
    {};
   \node [left=-1cm of S1] {$S_1$};

   \node (S2) [cloud,draw,cloud puffs=15,cloud puff arc=100, aspect=1.5,
          text opacity=1,
          minimum width=3cm,minimum height=3cm,
          right=-1.5cm of S1]
    {};
   \node [right=-1cm of S2] {$S_2$};
     
   \node (n4) [transition, right=-1.2cm of S1] {$n_6$};  
   \node (n2) [place, above right=-0.6cm and -0.6cm of S2] {$n_2$};  
   \node (n3) [place, below right=-0.6cm and -0.6cm of S2] {$n_3$};  

   \path (n4) edge[-latex] (n2);
   \path (n4) edge[-latex] (n3);
       
   \node [below right=0.5cm and -2.3cm of S1] {$M'[S_2 \setminus S_1]$ is a pWF net};
    \node [right=1cm of S2] {\textbf{FALSE}};
\end{scope}

\node [rotate=-30] at (11cm, 1cm) {\Large $\Rightarrow$};

\begin{scope}[shift={(13cm, 0.1cm)}]
  \node (S1) [cloud,draw,cloud puffs=15,cloud puff arc=100, aspect=1.5,
          text opacity=1,
          minimum width=3cm,minimum height=3cm]
    {};
   \node [left=-1cm of S1] {$S_1$};

   \node (S2) [cloud,draw,cloud puffs=15,cloud puff arc=100, aspect=1.5,
          text opacity=1,
          minimum width=3cm,minimum height=3cm,
          right=-1.5cm of S1]
    {};
   \node [right=-1cm of S2] {$S_2$};
     
   \node (n4) [place, right=-1.2cm of S1] {$n_6$};  
   \node (n2) [transition, above right=-0.6cm and -0.6cm of S2] {$n_2$};  
   \node (n3) [transition, below right=-0.6cm and -0.6cm of S2] {$n_3$};  

   \path (n4) edge[-latex] (n2);
   \path (n4) edge[-latex] (n3);
   
    \node [below right=0.5cm and -2.3cm of S1] {$M'[S_2 \setminus S_1]$ is a tWF net};
    \node [right=1cm of S2] {\textbf{FALSE}};
\end{scope}

\end{scope} % end of first case, line (a)

\end{tikzpicture}
}
\end{center}
\caption{\label{fig:one-input-one-output-B}Illustration of proof of preservation of one-input and one-output properties in case (B)}
\end{figure}
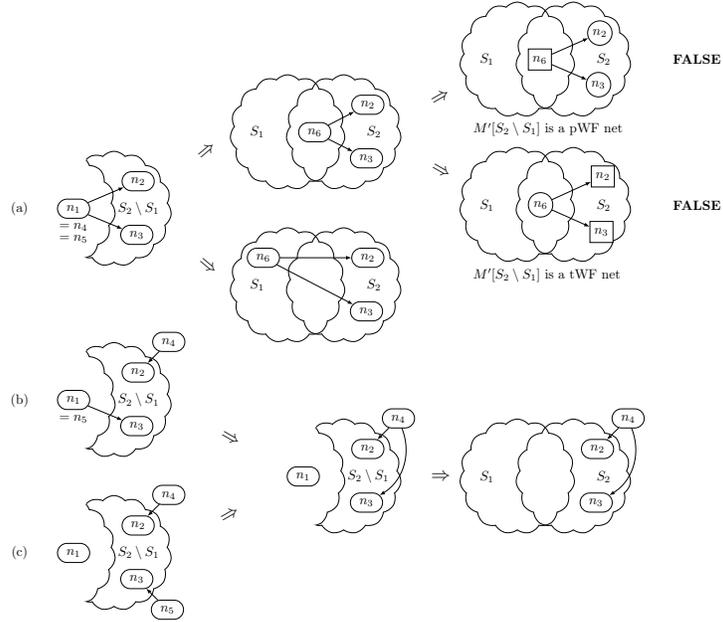

This concludes the lemmas that show that (1) the well-nestedness of $S_2$ in $M$ is preserved by $S_2 \setminus S_1$ in $M'$ and (2) all the defining properties of basic AND-OR nets are also preserved in $M'[S_2 \setminus S_1]$ if $M[S_2]$ had them. This leads us to the following lemma.

\begin{lem}[preservation of contractibility] \label{lem:contractibility-B}
The subnet $S_2 \setminus S_1$ is contractible in $M'$.
\end{lem}

\begin{proof}
We can show the following three claims: (1) $S_2 \setminus S_1$ is well-nested in $M'$, (2) $M'[S_2 \setminus S_1]$ is a WF net and (3) $M'[S_2 \setminus S_1]$ belongs to the same basic AND-OR class as $M[S_2]$.
\begin{enumerate}[(1)]

\item This follows from Lemma~\ref{lem:well-nested-B}. 

\item It follows from Lemma~\ref{lem:IO-type-B} (preservation of I/O type) that $M'[S_2 \setminus S_1]$ is I/O consistent. It then follows that $M'[S_2 \setminus S_1]$ is a WF net, by Theorem~\ref{thm:subnet-wf-char}, which states that every I/O-consistent subnet of a WF net is a WF net, and the fact $M'$ is indeed a WF net, as follows from Theorem~\ref{thm:contraction-correctness}, which states that the result of a contraction in a WF net is again a WF net. 

\item This claim follows from the observation that all defining properties of basic AND-OR nets are preserved by $M'[S_2 \setminus S_1]$ if they were satisfied by $M[S_2]$, as is shown by Lemma~\ref{lem:IO-type-B} (preservation of I/O type), Lemma~\ref{lem:AND-OR-B} (preservation of the AND and OR properties), Lemma~\ref{lem:acyclic-B} (preservation of acyclicity) and Lemma~\ref{lem:one-input-one-output-B} (preservation of one-input and one-output properties). 

\end{enumerate}
\end{proof}

We now turn to the question if the order of contractions of $S_1$ and $S_2$ influences the final result. Before we proceed to the lemma that states this, we present two auxiliary lemmas.

\begin{lem}[no edges from external nodes to nodes in intersection] \label{lem:no-edge-to-intersect}
 In $M$ there are no edges from a node outside $S_1 \cup S_2$ to a node in $S_1 \cap S_2$.
\end{lem}

\begin{proof}
If such an edge exists, then its end node is an input node of both $M[S_1]$ and $M[S_2]$, but this contradicts the assumption that the I/O types of $M[S_1]$ and $M[S_2]$ are different.
\end{proof}

\begin{lem}[shifting of edges] \label{lem:edge-shift}
 There is in $M$ an edge from a node in $S_1 \setminus S_2$ to a node in $S_2$ iff there is in $M$ an edge from a node in $S_1$ to a node in $S_2 \setminus S_1$.
\end{lem}

\begin{proof}
We first consider the \emph{if}-part of the claim. We start with assuming there is in $M$ an edge $(n_3, n_4)$ with $n_3 \in S_1$ and $n_4 \in S_2 \setminus S_1$. If $n_3 \not\in S_2$ then the edge $(n_3, n_4)$ is indeed an edge from $S_1 \setminus S_2$ to $S_2$. So we proceed under the assumption that $n_3 \in S_2$, and therefore $n_3 \in S_1 \cap S_2$. We now consider the cases (i) all output nodes of $M[S_1]$ are in $S_1 \cap S_2$ and (ii) there is an output node of $M[S_1]$ in $S_1 \setminus S_2$:  
\begin{enumerate}[(i)]

  \item Since $M[S_1]$ is well-connected, and because there is at least one node in $S_1 \setminus S_2$, as $S_1$ is not nested in $S_2$, there must be in $M[S_1]$ a path from that node to an output node in $S_1 \cap S_2$. It follows that there is in that path an edge from a node in $S_1 \setminus S_2$ to a node in $S_1 \cap S_2$.
  
  \item Let this output node of $M[S_1]$ in $S_1 \setminus S_2$ be $n_5$. Since $n_3$ is also an output node of $S_1$ with the edge $(n_3, n_4)$ to a node outside $S_1$, it follows from the well-nestedness of $S_1$ in $M$ that there is also the edge $(n_5, n_4)$, which is from a node in $S_1 \setminus S_2$ to $S_2$.

\end{enumerate}

The proof for the \emph{only-if}-part of the claim is similar, but with the directions of the edges reversed and the roles of input nodes and output nodes interchanged.
\end{proof}

\begin{lem}[commutativity of contraction] \label{lem:commutativity-B}
Let $M''_1$ be the result of first contracting $S_1$ into node $n_1$ resulting in $M'_1$ followed by contracting $S_2 \setminus S_1$ into node $n_2$, and let $M''_2$ be the result of first contracting $S_2$ into node $n_2$ resulting in $M'_2$ followed by contracting $S_1 \setminus S_2$ into node $n_1$. Then $M''_1 = M''_2$.
\end{lem}

\begin{proof}
The WF nets $M''_1$ and $M''_2$ have the same sets of nodes, since in both cases all the nodes of $S_1$ and $S_2$ are removed, and the new nodes $n_1$ and $n_2$ are added. The types and interconnections of the preserved nodes are not changed by contractions. In addition, the types of $n_1$ and $n_2$ will in both cases be the same, which can be shown as follows. Consider the type of $n_1$, which is either equal to the I/O type of $M[S_1]$, if this is contracted first, or the I/O type of $M'[S_1 \setminus S_2]$, if $M[S_2]$ is contracted first, but by Lemma~\ref{lem:IO-type-B} it follows that these are the same. By a similar argument, with $S_1$ and $S_2$ interchanged, it can be shown that the type of $n_2$ is in both cases the same.

Since we are only considering internal contractions, the sets of input and output nodes will not change and therefore be in both $M''_1$ and $M''_2$ identical to those in $M$.

So what remains to be shown is that in $M''_1$ and $M''_2$ the nodes $n_1$ and $n_2$ have (1) the same incoming and outgoing edge and (2) the same edges between them.

\begin{enumerate}[(1)]

\item We first consider incoming edges. Assume that in $M''_1$ there is an edge $(n_3, n_1)$ with $n_3 \neq n_2$. It follows there is in $M$ an edge $(n_3, n_4)$ with $n_4$ an input node of $M[S_1]$. By Lemma~\ref{lem:no-edge-to-intersect} we know that $n_4 \not\in S_1 \cap S_2$. It follows that $n_4 \in S_1 \setminus S_2$, and so the edge $(n_3, n_1)$ will be created when $S_1 \setminus S_2$ is contracted. It follows that this edge exists in $M''_2$. By symmetry it also holds that every incoming edge of $n_2$ in $M''_2$ also exists in $M''_1$.

The proof for outgoing edges is similar, except that the edges are reversed an the sets of input nodes are interchanged with sets of output nodes.

\item We first consider edges from $n_1$ to $n_2$. Assume there is an edge $(n_1, n_2)$ in $M''_1$. Then there is in $M$ an edge from a node in $S_1$ to a node in $S_2 \setminus S_1$. By Lemma~\ref{lem:edge-shift} it follows that there is in $M$ an edge from a node in $S_1 \setminus S_2$ to a node in $S_2$. Then there will be an edge $(n_1, n_2)$ in $M''_2$. By symmetry it also holds that if there is an edge  $(n_1, n_2)$ in $M''_2$, then there is also an edge  $(n_1, n_2)$ in $M''_1$.

The proof for edges from $n_2$ to $n_1$ proceeds analogously.

\end{enumerate}
\end{proof}

\subsection{Overlapping, but not nested, subnets with identical I/O types}

In this subsection we consider case (C) in Figure~\ref{fig:cases-of-proof}.  Therefore, we assume that $S_1$ and $S_2$ are contractible subnets in a WF net $M$ such that (i) they share nodes but not so that one is entirely nested inside the other and (ii) their associated WF nets $M[S_1]$ and $M[S_2]$ have the same I/O type. Let us also assume that $M'$ is the result of contracting $S_1$ into the node $n_1$ and $M''$ the result of contracting $S_2$ into the node $n_2$. Recall that we need to show that (1) after contraction of $S_1$ into $n_1$, the subnet $(S_2 \setminus S_1) \cup \{ n_1 \}$ is contractible and (2) if we contract first $S_1$ into $n_1$ and then $(S_2 \setminus S_1) \cup \{ n_1 \}$ into $n_3$, the result is the same as when we contract first $S_2$ into $n_2$ and then $(S_1 \setminus S_2  ) \cup \{ n_2 \}$ into $n_3$. 

Like in the previous subsection, the proof of (1) will consist of a list of lemmas that show that the different properties that define contractibility are all preserved for $(S_2 \setminus S_1) \cup \{ n_1 \}$ by the contraction of $S_1$. We start with the property of well-nestedness, and then move on to the defining properties of WF nets and basic AND-OR nets, and close of with a lemma showing (2). 

\begin{lem}[preservation of well-nestedness] \label{lem:well-nested-C}
The subnet $(S_2 \setminus S_1) \cup \{ n_1 \}$ is well-nested in $M'$.
\end{lem}

\begin{proof}
This proof is similar to that of Lemma~\ref{lem:well-nested-B}. We first show the well-nestedness for input nodes of $M'[(S_2 \setminus S_1) \cup \{ n_1 \}]$. For that we need to show that if there are two edges $(n_2, n_3)$ and $(n_4, n_5)$ in $M'$ such that $n_2, n_4 \not\in (S_2 \setminus S_1) \cup \{ n_1 \}$ and $n_3, n_5 \in (S_2 \setminus S_1) \cup \{ n_1 \}$, then there are edges $(n_2, n_5)$ and $(n_4, n_3)$ in $M'$. Let us assume that such  $(n_2, n_3)$ and $(n_4, n_5)$ exist. We consider the following four cases: (i) both $n_3$ and $n_5$ are equal to $n_1$, (ii)  only one of $n_3$ and $n_5$ is equal to $n_1$, (iii) neither $n_3$ nor $n_5$ is equal to $n_1$.
\begin{enumerate}[(i)]

\item Since $n_3 = n_1 = n_5$, it holds that $(n_4, n_5) = (n_4, n_1) = (n_4, n_3)$ and that $(n_2, n_3) = (n_2, n_1) = (n_2, n_5)$. So it indeed follows that $(n_2, n_5)$ and $(n_4, n_3)$ exist in $M'$.

\item We can assume, without loss of generality, that $n_3 = n_1$ and $n_5 \neq n_1$. The reasoning that follows is illustrated in Figure~\ref{fig:well-nested-C}. Then, by definition of contraction, there must be an input node $n_6$ of $M[S_1]$ and an edge $(n_2, n_6)$ in $M$. We now distinguish two sub-cases, namely that (1) all input nodes of $M[S_1]$ are in $S_1 \setminus S_2$ and (2) there is at least one input node of $M[S_1]$ in $S_1 \cap S_2$. We consider the two cases.

\begin{enumerate}[(1)]

\item If all input nodes of  $M[S_1]$ are in $S_1 \setminus S_2$, then so must $n_6$. Moreover, it also then holds for all nodes in $S_1 \cap S_2$ that there is a path in $M[S_1]$ to them from a node in $S_1 \setminus S_2$. Take one such path, and take the first edge $(n_7, n_8)$ in that path where $n_7 \in S_1 \setminus S_2$ and $n_8 \in S_1 \cap S_2$. Note that such an edge must exist if the path goes from $S_1 \setminus S_2$ to $S_1 \cap S_2$ and contains only nodes from $S_1$. Then, by the well-nestedness of $S_2$ in $M$, there is an edge $(n_4, n_8)$ in $M$, since $n_8$ and $n5$ are both input nodes of $M[S_2]$. Then, it follows that $n_8$ is an input node of $M[S_1]$, but this contradicts the assumption that all input nodes of $M[S_1]$ are in $S_1 \setminus S_2$, so leads to a contradiction.

\item Because of the well-nestedness of $S_2$ in $M$, there are also edges $(n_2, n_5)$ and $(n_4, n_6)$ in $M$. By the definition of contraction it then follows that there are edges $(n_2, n_5)$ and $(n_4, n_1)$ in $M'$. Since $n_3 = n_1$, there is then also an edge $(n_4, n_3)$ in $M'$.

\end{enumerate}

\item So we assume that $n_3 \neq n_1$ and $n_5 \neq n_1$. The reasoning that follows is identical to that for the corresponding case (iii) in the proof for Lemma~\ref{lem:well-nested-B} and is illustrated equally in Figure~\ref{fig:well-nested-B} (b).

\end{enumerate}  
The proof for well-nestedness of the output nodes of $M'[(S_2 \setminus S_1) \cup \{ n_1 \}]$ is similar to that of the input nodes, except that the direction of the edges is reversed and the set of input nodes is replaced with the set of output nodes.
\end{proof}

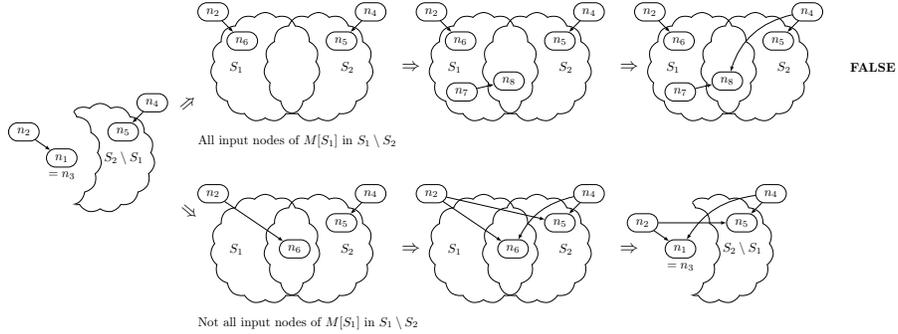
\begin{figure}[htb]
\begin{center}
\resizebox{\textwidth}{!}{%
\begin{tikzpicture}
    \tikzstyle{transition} = [rectangle,draw,minimum width=0.55cm, minimum height=0.55cm,fill=white]
    \tikzstyle{place} = [circle,draw,minimum width=0.55cm, minimum height=0.55cm,fill=white,inner sep=0.08cm]
    \tikzstyle{place-transition} = [rounded rectangle,draw,minimum width=1cm, minimum height=0.5cm,fill=white,inner sep=0.08cm]

\begin{scope}[shift={(1cm, 0cm)}]
   \node (S1) [cloud,cloud puffs=15,cloud puff arc=100, aspect=1.5,
          text opacity=1,
          minimum width=3cm,minimum height=3cm]
    {};

   \node (S2) [cloud,draw,cloud puffs=15,cloud puff arc=100, aspect=1.5,
          text opacity=1,
          minimum width=3cm,minimum height=3cm,
          right=-1.5cm of S1]
    {};
   \node [right=-1.5cm of S2] {$S_2 \setminus S_1$};

  % erase S1 from S2
   \node (ES1) [cloud, draw, fill=white, color=white, cloud puffs=15,cloud puff arc=100, aspect=1.5,
          text opacity=1,
          minimum width=3cm,minimum height=3cm] at (S1.center)
    {};
    
    % put back erased edge of S2
  \begin{scope}[]
    \pgftransformshift{\pgfpointanchor{S2}{center}}
    \pgfset{cloud puffs=15,cloud puff arc=100, aspect=1.5,
       minimum width=3cm,minimum height=3cm}
    \pgfnode{cloud}{center}{}{nodename}{\pgfusepath{clip}}

   \node [draw,cloud,cloud puffs=15,cloud puff arc=100, aspect=1.5,
       minimum width=3cm,minimum height=3cm] 
     at (S1.center) {};
  \end{scope}
     
   \node (n1) [place-transition, right=-1.5cm of S1] {$n_1$};  
   \node [below=0cm of n1] {$=n_3$};

   \node (n5) [place-transition, above right=-0.6cm and -0.6cm of S2] {$n_5$};  
   \node (n2) [place-transition, above left=-0.6cm and -0.6cm of S1] {$n_2$};  

   \node (n4) [place-transition, above right=0.2cm and 0.2cm of S2] {$n_4$};  
   
   \path (n2) edge[-latex] (n1);
   \path (n4) edge[-latex] (n5);
\end{scope}

\node [rotate=40] at (4.9cm, 1.5cm) {\Large $\Rightarrow$};

\begin{scope}[shift={(7cm, -2.5cm)}]
   \node (S1) [cloud,draw,cloud puffs=15,cloud puff arc=100, aspect=1.5,
          text opacity=1,
          minimum width=3cm,minimum height=3cm]
    {};
   \node [left=-1cm of S1] {$S_1$};

   \node (S2) [cloud,draw,cloud puffs=15,cloud puff arc=100, aspect=1.5,
          text opacity=1,
          minimum width=3cm,minimum height=3cm,
          right=-1.5cm of S1]
    {};
   \node [right=-1cm of S2] {$S_2$};
     
   \node (n6) [place-transition, right=-1.1cm of S1] {$n_6$};  
   \node (n2) [place-transition, above left=0.2cm and 0.2cm of S1] {$n_2$};  
   \node (n5) [place-transition, above right=-0.6cm and -0.6cm of S2] {$n_5$};  

   \node (n4) [place-transition, above right=0.2cm and 0.2cm of S2] {$n_4$}; 

   \path (n2) edge[-latex] (n6);
   \path (n4) edge[-latex] (n5);
   
   \node [below right=0.7 cm and -3 cm of S1] {Not all input nodes of $M[S_1]$ in $S_1 \setminus S_2$};
\end{scope}

\node at (11cm, -2.5cm) {\Large $\Rightarrow$};

\begin{scope}[shift={(13cm, -2.5cm)}]
   \node (S1) [cloud,draw,cloud puffs=15,cloud puff arc=100, aspect=1.5,
          text opacity=1,
          minimum width=3cm,minimum height=3cm]
    {};
   \node [left=-1cm of S1] {$S_1$};

   \node (S2) [cloud,draw,cloud puffs=15,cloud puff arc=100, aspect=1.5,
          text opacity=1,
          minimum width=3cm,minimum height=3cm,
          right=-1.5cm of S1]
    {};
   \node [right=-1cm of S2] {$S_2$};
     
   \node (n6) [place-transition, right=-1.1cm of S1] {$n_6$};  
   \node (n2) [place-transition, above left=0.2cm and 0.2cm of S1] {$n_2$};  
   \node (n5) [place-transition, above right=-0.6cm and -0.6cm of S2] {$n_5$};  

   \node (n4) [place-transition, above right=0.2cm and 0.2cm of S2] {$n_4$}; 

   \path (n2) edge[-latex] (n6);
   \path (n4) edge[-latex] (n5);
   \path (n2) edge[-latex] (n5);
   \path (n4) edge[-latex, bend right=25] (n6);   
\end{scope}

\begin{scope}[shift={(18cm, -2.5cm)}]
   \node (S1) [cloud,cloud puffs=15,cloud puff arc=100, aspect=1.5,
          text opacity=1,
          minimum width=3cm,minimum height=3cm]
    {};

   \node (S2) [cloud,draw,cloud puffs=15,cloud puff arc=100, aspect=1.5,
          text opacity=1,
          minimum width=3cm,minimum height=3cm,
          right=-1.5cm of S1]
    {};
   \node [right=-1.5cm of S2] {$S_2 \setminus S_1$};

  % erase S1 from S2
   \node (ES1) [cloud, draw, fill=white, color=white, cloud puffs=15,cloud puff arc=100, aspect=1.5,
          text opacity=1,
          minimum width=3cm,minimum height=3cm] at (S1.center)
    {};
    
    % put back erased edge of S2
  \begin{scope}[]
    \pgftransformshift{\pgfpointanchor{S2}{center}}
    \pgfset{cloud puffs=15,cloud puff arc=100, aspect=1.5,
       minimum width=3cm,minimum height=3cm}
    \pgfnode{cloud}{center}{}{nodename}{\pgfusepath{clip}}

   \node [draw,cloud,cloud puffs=15,cloud puff arc=100, aspect=1.5,
       minimum width=3cm,minimum height=3cm] 
     at (S1.center) {};
  \end{scope}
     
   \node (n1) [place-transition, right=-1.5cm of S1] {$n_1$};  
   \node [below=0cm of n1] {$=n_3$};

   \node (n5) [place-transition, above right=-0.6cm and -0.6cm of S2] {$n_5$};  
   \node (n2) [place-transition, above left=-0.6cm and -0.6cm of S1] {$n_2$};  

   \node (n4) [place-transition, above right=0.2cm and 0.2cm of S2] {$n_4$};  
   
   \path (n2) edge[-latex] (n1);
   \path (n4) edge[-latex] (n5);
   \path (n2) edge[-latex] (n5);
   \path (n4) edge[-latex, bend right=25] (n1);   
\end{scope}

\node at (17cm, -2.5cm) {\Large $\Rightarrow$};

\node [rotate=-40] at (4.9cm, -1.5cm) {\Large $\Rightarrow$};

\begin{scope}[shift={(7cm, 2.5cm)}]
   \node (S1) [cloud,draw,cloud puffs=15,cloud puff arc=100, aspect=1.5,
          text opacity=1,
          minimum width=3cm,minimum height=3cm]
    {};
   \node [left=-1cm of S1] {$S_1$};

   \node (S2) [cloud,draw,cloud puffs=15,cloud puff arc=100, aspect=1.5,
          text opacity=1,
          minimum width=3cm,minimum height=3cm,
          right=-1.5cm of S1]
    {};
   \node [right=-1cm of S2] {$S_2$};
     
   \node (n6) [place-transition, above left=-0.6cm and -0.6cm of S1] {$n_6$};  

   \node (n2) [place-transition, above left=0.2cm and 0.2cm of S1] {$n_2$};  
   \node (n5) [place-transition, above right=-0.6cm and -0.6cm of S2] {$n_5$};  

   \node (n4) [place-transition, above right=0.2cm and 0.2cm of S2] {$n_4$};  

   \path (n2) edge[-latex] (n6);
   \path (n4) edge[-latex] (n5);
   
   \node [below right=0.7 cm and -3 cm of S1] {All input nodes of $M[S_1]$ in $S_1 \setminus S_2$};
\end{scope}

\node at (11cm, 2.5cm) {\Large $\Rightarrow$};

\begin{scope}[shift={(13cm, 2.5cm)}]
   \node (S1) [cloud,draw,cloud puffs=15,cloud puff arc=100, aspect=1.5,
          text opacity=1,
          minimum width=3cm,minimum height=3cm]
    {};
   \node [left=-1cm of S1] {$S_1$};

   \node (S2) [cloud,draw,cloud puffs=15,cloud puff arc=100, aspect=1.5,
          text opacity=1,
          minimum width=3cm,minimum height=3cm,
          right=-1.5cm of S1]
    {};
   \node [right=-1cm of S2] {$S_2$};
     
   \node (n6) [place-transition, above left=-0.6cm and -0.6cm of S1] {$n_6$};  

   \node (n2) [place-transition, above left=0.2cm and 0.2cm of S1] {$n_2$};  
   \node (n5) [place-transition, above right=-0.6cm and -0.6cm of S2] {$n_5$};  

   \node (n4) [place-transition, above right=0.2cm and 0.2cm of S2] {$n_4$};  

   \path (n2) edge[-latex] (n6);
   \path (n4) edge[-latex] (n5);
      
   \node (n8) [place-transition, below right=-0.9cm and -0.5cm of S1] {$n_8$};  
   \node (n7) [place-transition, below left=-0.6cm and -0.6cm of S1] {$n_7$};  

   \path (n7) edge[-latex] (n8);   
\end{scope}

\node at (17cm, 2.5cm) {\Large $\Rightarrow$};

\begin{scope}[shift={(19cm, 2.5cm)}]
   \node (S1) [cloud,draw,cloud puffs=15,cloud puff arc=100, aspect=1.5,
          text opacity=1,
          minimum width=3cm,minimum height=3cm]
    {};
   \node [left=-1cm of S1] {$S_1$};

   \node (S2) [cloud,draw,cloud puffs=15,cloud puff arc=100, aspect=1.5,
          text opacity=1,
          minimum width=3cm,minimum height=3cm,
          right=-1.5cm of S1]
    {};
   \node [right=-1cm of S2] {$S_2$};
     
   \node (n6) [place-transition, above left=-0.6cm and -0.6cm of S1] {$n_6$};  

   \node (n2) [place-transition, above left=0.2cm and 0.2cm of S1] {$n_2$};  
   \node (n5) [place-transition, above right=-0.6cm and -0.6cm of S2] {$n_5$};  

   \node (n4) [place-transition, above right=0.2cm and 0.2cm of S2] {$n_4$};  

   \path (n2) edge[-latex] (n6);
   \path (n4) edge[-latex] (n5);
      
   \node (n8) [place-transition, below right=-0.9cm and -0.5cm of S1] {$n_8$};  
   \node (n7) [place-transition, below left=-0.6cm and -0.6cm of S1] {$n_7$};  

   \path (n7) edge[-latex] (n8);
   \path (n4) edge[-latex, bend right=30] (n8);
   
   \node [right=1cm of S2] {\textbf{FALSE}};
\end{scope}

\end{tikzpicture}
}
\end{center}
\caption{\label{fig:well-nested-C}Illustration of proof of preservation of well-nestedness in case (C)}
\end{figure}

\begin{lem}[preservation of I/O type] \label{lem:IO-type-C}
The input nodes and output nodes of $M'[(S_2 \setminus S_1) \cup \{ n_1 \}]$ have the same type as the input nodes and output nodes of $M[S_2]$, respectively.
\end{lem}

\begin{proof}
We first consider input nodes. Assume $n_2$ is an input node of  $M'[(S_2 \setminus S_1) \cup \{ n_1 \}]$, then it will have in $M'$ an incoming edge $(n_3, n_2)$ with $n_3 \not\in (S_2 \setminus S_1) \cup \{ n_1 \}$. Then either (i) $n_2 \in S_2 \setminus S_1$ or (ii) $n_2 = n_1$:
\begin{enumerate}[(i)]

  \item Then this edge already existed in $M$, and so $n_2$ is an input node in $M[S_2]$. Since the type of a node is not changed by contraction, it follows that the type of $n_2$ is equal to the I/O type of $M[S_2]$.
  
   \item By the definition of contraction the type of $n_1$ is the I/O type of $M[S_1]$, which is by assumption equal to the I/O type of $M[S_2]$. It follows that the type of $n_2 = n_1$ is equal to the I/O type of $M[S_2]$.
   
\end{enumerate}
The argument for output nodes is similar, except that edges are reversed and the set of input nodes is replaced with the set of output nodes.
\end{proof}

Before we proceed with the lemma for the preservation of the AND and OR properties, we introduce an auxiliary lemma about the relationship between $S_1$ and $S_2$.

\begin{lem}[edges only via intersection] \label{lem:edge-via-intersection}
There cannot be edges in $M$ between nodes in $S_1 \setminus S_2$ and nodes in $S_2 \setminus S_1$.
\end{lem}

\begin{proof}
If there is such an edge, then its starting node is an output node of one net, and the end node an input node of the other net. Since these nodes must have different types in a Petri net, it follows that the I/O types of two nets are different, but this contradicts the assumption that they are the same.
\end{proof}

\begin{lem}[preservation of the AND and the OR property] \label{lem:AND-OR-C}
The net $M'[(S_2 \setminus S_1) \cup \{ n_1 \}]$ has the AND (OR) property if $M[S_2]$ has the AND (OR) property.
\end{lem}

\begin{proof}
We first consider the AND property, and of this first the restriction on incoming edges. The proof proceeds by contradiction, so we start with assuming that $M[S_2]$ has the AND property and $M'[(S_2 \setminus S_1) \cup \{ n_1 \}]$ does not have the AND property. We will show that it then follows that $M[S_2]$ does \emph{not} have the AND property.

If $M'[(S_2 \setminus S_1) \cup \{ n_1 \}]$ does not have the AND property, then in $M'[(S_2 \setminus S_1) \cup \{ n_1 \}]$ there is a place $p$ that violates the AND property. We distinguish two cases: (i)  $p \in S_2 \setminus S_1$ or (ii) $p = n_1$.

\begin{enumerate}[(i)]

\item The node $p$ can violate the AND property in four ways: (a) it has two distinct incoming edges from nodes in $S_2 \setminus S_1$, (b) it has in $M'[(S_2 \setminus S_1) \cup \{ n_1 \}]$ an incoming edge from a node in $S_2 \setminus S_1$ and an incoming edge from  $n_1$,  (c) it is an input node of $M'[(S_2 \setminus S_1) \cup \{ n_1 \}]$ and has in $M'[(S_2 \setminus S_1) \cup \{ n_1 \}]$ an incoming edge from a node in $S_2 \setminus S_1$ and (d) it is an input node of $M'[(S_2 \setminus S_1) \cup \{ n_1 \}]$ and has in $M'[(S_2 \setminus S_1) \cup \{ n_1 \}]$ an incoming edge from $n_1$.

\begin{enumerate}[(a)]

\item In this case these two edges also exist in $M[S_2]$, and so $M[S_2]$ does not have the AND property.

\item Let $(t', p)$ be the incoming edge from $S_2 \setminus S_1$. Since there is an edge $(n_1, p)$ in $M'$, there is in $M$ an edge $(t'', p)$ from an output transition $t''$ of $M[S_1]$ to $p$. By Lemma~\ref{lem:edge-via-intersection} it holds that $t'' \in S_1 \cap S_2$. It follows that $(t'', p)$  exists in $M[S_2]$. However, since there is the edge $(t', p)$ in $M'[S_2 \setminus S_1]$ with $t' \in S_2 \setminus S_1$, which necessarily also exists in $M[S_2]$, it follows that in $M[S_2]$ the place $p$ has two distinct incoming edges, and therefore $M[S_2]$ does not have the AND property.

\item If $p$ is an input node of $M'[(S_2 \setminus S_1) \cup \{ n_1 \}]$ and $p in S_2 \setminus S_1$, then it is also an input node of $M[S_2]$, since we consider internal contractions and an incoming edge in $M'$ from outside $(S_2 \setminus S_1) \cup \{ n_1 \}$ is necessarily also present in $M$. The postulated incoming edge in $M'$ from a node in $S_2 \setminus S_1$ is also necessarily present in $M[S_2]$. It follows that in $M[S_2]$ the node $p$ is an input node and has an incoming edge and so the subnet does not have the AND property.

\item Since there is an edge $(n_1, p)$ in $M'$, there is in $M$ an edge $(t'', p)$ from an output transition $t''$ of $M[S_1]$ to $p$. By Lemma~ \ref{lem:edge-via-intersection} the node $t''$ must be in $S_1 \cap S_2$. It follows that in $M[S_2]$ the node $p$ is an input node and has an incoming edge and so the subnet does not have the AND property.

\end{enumerate}

\item The node $p = n_1$ can violate the AND property in two ways: (a) it has two distinct incoming edges from nodes in $S_2 \setminus S_1$ and (b) it is an input node of $M'[(S_2 \setminus S_1) \cup \{ n_1 \}]$ and has in $M'[(S_2 \setminus S_1) \cup \{ n_1 \}]$ an incoming edge from a node in $S_2 \setminus S_1$.

\begin{enumerate}[(a)]

\item Let the two distinct incoming edges from nodes in $S_2 \setminus S_1$ be $(n_2, p)$ and $(n_3, p)$. Then, by the definition of contraction, and since $S_1$ is well-nested in $M$, there is an input node $n_4$ of $M[S_1]$ with two edges $(n_2, n_4)$ and $(n_3, n_4)$ in $M$. By Lemma~ \ref{lem:edge-via-intersection} the node $n_4$ must be in $S_1 \cap S_2$. But then $n_4$ has two distinct incoming edges in $M[S_2]$, and so $M[S_2]$ would not have the AND property.    

\item Let $(n_2, n_1)$ be an edge that makes $n_1$ an input node of $M'[(S_2 \setminus S_1) \cup \{ n_1 \}]$, which mus exist since we consider internal contractions, and so $n_2 \not\in (S_2 \setminus S_1) \cup \{ n_1 \}$. Let $(n_3, n_1)$ be the edge in $M'[(S_2 \setminus S_1) \cup \{ n_1 \}]$ such that $n_3 \in S_2 \setminus S_1$. From the definition of contraction it follows that there is an input node $n_4$ of $M[S_1]$, with an edge $(n_3, n_4)$ in $M$. Because $S_1$ is well-nested in $M$, there is also an edge $(n_2, n_4)$ in $M$. By Lemma~ \ref{lem:edge-via-intersection}, the node $n_4$ must be in $S_1 \cap S_2$. But then $n_4$ is both an input node of $M[S_2]$, because of $(n_2, n_4)$, and has in $M[S_2]$ an incoming edge, namely $(n_3,n_4)$. Therefore $M[S_2]$ does not have the AND property.

\end{enumerate}

\end{enumerate}
The proof for the restriction on the outgoing edges is similar, except that the direction of the edges is reversed and the sets of output places are replaced with the sets of input places, and vice versa.
    
The proof for the preservation for the OR property is similar to that for the AND property, except that places are replaced with transitions and vice versa.
\end{proof}

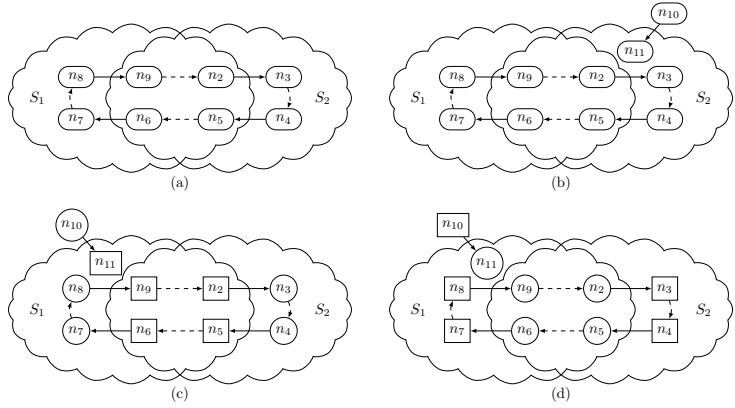
\begin{figure}[htb]
\begin{center}
\resizebox{0.8\textwidth}{!}{%
\begin{tikzpicture}
    \tikzstyle{transition} = [rectangle,draw,minimum width=0.55cm, minimum height=0.55cm,fill=white]
    \tikzstyle{place} = [circle,draw,minimum width=0.55cm, minimum height=0.55cm,fill=white,inner sep=0.08cm]
    \tikzstyle{place-transition} = [rounded rectangle,draw,minimum width=1cm, minimum height=0.5cm,fill=white,inner sep=0.08cm]

% -- case: general pattern

\begin{scope}[shift={(0cm, 0cm)}]
\node (a-cl-left) [cloud, draw, cloud puffs=20, cloud puff arc=100, 
  minimum width=5.8cm, minimum height=3.5cm] at (2.1, 1) {};
\node (a-cl-right) [cloud, draw, cloud puffs=20, cloud puff arc=100, 
  minimum width=5.8cm, minimum height=3.5cm] at (4.4, 1) {};

\node (n3) [place-transition] at (5.7, 1.5) {$n_{3}$};
\node (n2) [place-transition] at (4.1, 1.5) {$n_{2}$};
\node (n9) [place-transition] at (2.4, 1.5) {$n_{9}$};
\node (n8) [place-transition] at (0.8, 1.5) {$n_{8}$};
\node (n7) [place-transition] at (0.8, 0.5) {$n_{7}$};
\node (n6) [place-transition] at (2.4, 0.5) {$n_{6}$};
\node (n5) [place-transition] at (4.1, 0.5) {$n_{5}$};
\node (n4) [place-transition] at (5.7, 0.5) {$n_{4}$};

\node [left=-1cm of a-cl-left] {$S_1$};
\node [right=-1cm of a-cl-right] {$S_2$};

\path (n9) edge[-latex, dashed] (n2);
\path (n5) edge[-latex, dashed] (n6);
\path (n4) edge[-latex] (n5);
\path (n6) edge[-latex] (n7);
\path (n7) edge[-latex, dashed, bend left=20] (n8);
\path (n3) edge[-latex, dashed, bend left=20] (n4);
\path (n8) edge[-latex] (n9);
\path (n2) edge[-latex] (n3);

\node at (3.25, -1) {(a)};
\end{scope}

% -- case: incoming edge in S2

\begin{scope}[shift={(9cm, 0cm)}]
\node (a-cl-left) [cloud, draw, cloud puffs=20, cloud puff arc=100, 
  minimum width=5.8cm, minimum height=3.5cm] at (2.1, 1) {};
\node (a-cl-right) [cloud, draw, cloud puffs=20, cloud puff arc=100, 
  minimum width=5.8cm, minimum height=3.5cm] at (4.4, 1) {};

\node (n3) [place-transition] at (5.7, 1.5) {$n_{3}$};
\node (n2) [place-transition] at (4.1, 1.5) {$n_{2}$};
\node (n9) [place-transition] at (2.4, 1.5) {$n_{9}$};
\node (n8) [place-transition] at (0.8, 1.5) {$n_{8}$};
\node (n7) [place-transition] at (0.8, 0.5) {$n_{7}$};
\node (n6) [place-transition] at (2.4, 0.5) {$n_{6}$};
\node (n5) [place-transition] at (4.1, 0.5) {$n_{5}$};
\node (n4) [place-transition] at (5.7, 0.5) {$n_{4}$};

\node (n10) [place-transition] at (5.8, 3.0) {$n_{10}$};
\node (n11) [place-transition] at (5.0, 2.1) {$n_{11}$};

\node [left=-1cm of a-cl-left] {$S_1$};
\node [right=-1cm of a-cl-right] {$S_2$};

\path (n9) edge[-latex, dashed] (n2);
\path (n5) edge[-latex, dashed] (n6);
\path (n4) edge[-latex] (n5);
\path (n6) edge[-latex] (n7);
\path (n7) edge[-latex, dashed, bend left=20] (n8);
\path (n3) edge[-latex, dashed, bend left=20] (n4);
\path (n8) edge[-latex] (n9);
\path (n2) edge[-latex] (n3);
\path (n10) edge[-latex] (n11);

\node at (3.25, -1) {(b)};
\end{scope}

% -- case: incoming edge in S1, S2 is 11tAND

\begin{scope}[shift={(0cm, -5cm)}]
\node (a-cl-left) [cloud, draw, cloud puffs=20, cloud puff arc=100, 
  minimum width=5.8cm, minimum height=3.5cm] at (2.1, 1) {};
\node (a-cl-right) [cloud, draw, cloud puffs=20, cloud puff arc=100, 
  minimum width=5.8cm, minimum height=3.5cm] at (4.4, 1) {};
  
\node (n3) [place] at (5.7, 1.5) {$n_{3}$};
\node (n2) [transition] at (4.1, 1.5) {$n_{2}$};
\node (n9) [transition] at (2.4, 1.5) {$n_{9}$};
\node (n8) [place] at (0.8, 1.5) {$n_{8}$};
\node (n7) [place] at (0.8, 0.5) {$n_{7}$};
\node (n6) [transition] at (2.4, 0.5) {$n_{6}$};
\node (n5) [transition] at (4.1, 0.5) {$n_{5}$};
\node (n4) [place] at (5.7, 0.5) {$n_{4}$};

\node (n10) [place] at (0.7, 3.0) {$n_{10}$};
\node (n11) [transition] at (1.5, 2.1) {$n_{11}$};

\node [left=-1cm of a-cl-left] {$S_1$};
\node [right=-1cm of a-cl-right] {$S_2$};

\path (n9) edge[-latex, dashed] (n2);
\path (n5) edge[-latex, dashed] (n6);
\path (n4) edge[-latex] (n5);
\path (n6) edge[-latex] (n7);
\path (n7) edge[-latex, dashed, bend left=20] (n8);
\path (n3) edge[-latex, dashed, bend left=20] (n4);
\path (n8) edge[-latex] (n9);
\path (n2) edge[-latex] (n3);
\path (n10) edge[-latex] (n11);

\node at (3.25, -1) {(c)};
\end{scope}

% -- case: incoming edge in S1, S2 is tOR

\begin{scope}[shift={(9cm, -5cm)}]
\node (a-cl-left) [cloud, draw, cloud puffs=20, cloud puff arc=100, 
  minimum width=5.8cm, minimum height=3.5cm] at (2.1, 1) {};
\node (a-cl-right) [cloud, draw, cloud puffs=20, cloud puff arc=100, 
  minimum width=5.8cm, minimum height=3.5cm] at (4.4, 1) {};

\node (n3) [transition] at (5.7, 1.5) {$n_{3}$};
\node (n2) [place] at (4.1, 1.5) {$n_{2}$};
\node (n9) [place] at (2.4, 1.5) {$n_{9}$};
\node (n8) [transition] at (0.8, 1.5) {$n_{8}$};
\node (n7) [transition] at (0.8, 0.5) {$n_{7}$};
\node (n6) [place] at (2.4, 0.5) {$n_{6}$};
\node (n5) [place] at (4.1, 0.5) {$n_{5}$};
\node (n4) [transition] at (5.7, 0.5) {$n_{4}$};

\node (n10) [transition] at (0.7, 3.0) {$n_{10}$};
\node (n11) [place] at (1.5, 2.1) {$n_{11}$};

\node [left=-1cm of a-cl-left] {$S_1$};
\node [right=-1cm of a-cl-right] {$S_2$};

\path (n9) edge[-latex, dashed] (n2);
\path (n5) edge[-latex, dashed] (n6);
\path (n4) edge[-latex] (n5);
\path (n6) edge[-latex] (n7);
\path (n7) edge[-latex, dashed, bend left=20] (n8);
\path (n3) edge[-latex, dashed, bend left=20] (n4);
\path (n8) edge[-latex] (n9);
\path (n2) edge[-latex] (n3);
\path (n10) edge[-latex] (n11);

\node at (3.25, -1) {(d)};
\end{scope}

\end{tikzpicture}
}
\end{center}
\caption{\label{fig:acyclic-C}Illustration of proof of preservation of acyclicity}
\end{figure}

We continue with the lemma that shows the preservation of acyclicity. Note that the premise of this lemma is stronger then in its analogue in case (B), Lemma~\ref{lem:acyclic-B}. It not only requires $M[S_2]$ to be acyclic, but to be either a 11tAND net or a pAND net.

\begin{lem}[preservation of acyclicity] \label{lem:acyclic-C}
The net $M'[(S_2 \setminus S_1) \cup \{ n_1 \}]$ is acyclic if $M[S_2]$ is a 11tAND net or a pAND net.
\end{lem}

\begin{proof}
The proof proceeds by contradiction, so we assume that there is a cycle in $M'[(S_2 \setminus S_1) \cup \{ n_1 \}]$, and then show that it follows that there is a cycle in $M[S_2]$. So let us assume that there is a cycle in $M'[(S_2 \setminus S_1) \cup \{ n_1 \}]$. It is clear that if the cycle does not contain $n_1$, then this cycle is also present in $M[S_2]$, so we continue under the assumption that $n_1$ is contained in the cycle. We can assume that $n_1$ appears exactly once in the cycle, since if a cycle with multiple occurrences of $n_1$ exist, then there is also one with exactly one occurrence. In this cycle there must be an edge that leaves $n_1$, say $(n_1, n_3)$, and an edge that arrives in $n_1$, say $(n_4, n_1)$. Note that there is in $M'[S_2 \setminus S_1]$ a path from $n_3$ to $n_4$, or that these nodes are the same node.

Let us now consider the situation in $M$. Since there are edges $(n_1, n_3)$ and $(n_4, n_1)$ in $M'$, there must in $M$ be edges $(n_2, n_3)$ and $(n_4, n_5)$, with $n_2$ an output node of $M[S_1]$ and $n_5$ an input node of $M[S_1]$. By Lemma~\ref{lem:edge-via-intersection} it holds that $n_2, n_5 \in S_1 \cap S_2$. Moreover, since $S_1$ is well-nested in $M$ there are edges to $n_3$ from every output node of $M[S_1]$ and from $n_4$ to every input node of $M[S_1]$. We can pick $n_2$ and $n_5$ such that there is in $M[S_1]$ a path from $n_5$ to $n_2$, which must possible since $M[S_1]$ is well-connected. If this path contains only nodes from $S_1 \cap S_2$, then there is a cycle in $M[S_2]$ that consists of the edge $(n_2, n_3)$, the path from $n_3$ to $n_4$, the edge from $n_4$ to $n_5$ and finally the path from $n_5$ to $n_2$. So we continue under the assumption that the path from $n_5$ to $n_2$ contains at least one node from $S_1 \setminus S_2$. This situation is illustrated in Figure~\ref{fig:acyclic-C} (a). In the path from $n_5$ to $n_2$ we consider two special edges: the edge $(n_6, n_7)$ that is the first edge that goes from a node in $S_1 \cap S_2$ to a node in $S_1 \setminus S_2$, and the edge $(n_8, n_9)$ that is the last edge that goes from a node in $S_1 \setminus S_2$ to a node in $S_1 \cap S_2$. These edges must exist, since by Lemma~\ref{lem:edge-via-intersection} the nodes $n_2$ and $n_5$ are in $S_1 \cap S_2$. Note that the section of the path between $n_7$ and $n_8$ contains only nodes from $S_1$, but some of those might be in $S_1 \cap S_2$. Also note that the paths from $n_9$ to $n_2$, and from $n_5$ to $n_6$, contain only nodes in $S_1 \cap S_2$.

Since $M$ is a well-connected, there must for every node in $S_1 \cup S_2$ be a path from an input node of $M$ to that node. Moreover, since $S_1$ and $S_2$ are internal subnets, that path contains at least one edge. Taken one such path and let $(n_{10}, n_{11})$ be the first edge in that path such that $n_{10} \not\in S_1 \cup S_2$ and $n_{11} \in S_1 \cup S_2$. It then must hold that (i) $n_{11} \in S_2$ or (ii) $n_{11} \in S_1$. We consider the two cases:

\begin{enumerate}[(i)]

\item This case is illustrated in Figure~\ref{fig:acyclic-C} (b). Since $S_2$ is well-nested in $M$, there is an edge $(n_{10}, n_9)$ in $M$ as $n_9$ is an input node of $M[S_2]$. But then $n_9$ is also an input node of $M[S_1]$ since it has an incoming edge from outside $S_1$. Since $S_1$ is well-nested in $M$, and $n_5$ is an input node of $M[S_1]$, there must then be an edge $(n_4, n_9)$ in $M$. It follows that there must be cycle in $M[S_2]$ through the nodes $n_2$, $n_3$, $n_4$, $n_5$ and $n_9$.

\item This case is  illustrated in Figure~\ref{fig:acyclic-C} (c) and (d), where (c) shows the sub-case where $M[S_2]$ is a 11tAND net and (d) shows the sub-case where $M[S_2]$ is a pAND net. We consider these two sub-cases:

\begin{description}

\item[{$M[S_2]$} is a 11tAND net:] Since $S_1$ is well-nested in $M$, and $n_5$ is an input node of $M[S_1]$, there is an edge $(n_{10}, n_5)$ in $M$. Because of this edge, $n_5$ is also an input node of $S_2$. Since $n_9$ is also an input node of $S_2$, and $M[S_2]$ has exactly one input node, it follows that $n_5 = n_9$. Then, the path in $S_2$ that goes through the nodes $n_2$, $n_3$, $n_4$, $n_5$ and $n_9$, is in fact a cycle.

\item[{$M[S_2]$} is a pAND net:] As in the previous sub-case, we derive that $n_5$ is an input node of $S_2$. But then $n_5$ violates the AND property since it is both an input node of $M[S_2]$ and has in $M[S_2]$ an incoming edge, which contradicts the assumption that $M[S_2]$ is a pAND net.

\end{description}

\end{enumerate}
\end{proof}

\begin{lem}[preservation of one-input one-output property] \label{lem:one-input-one-output-C}
The net $M'[(S_2 \setminus S_1) \cup \{ n_1 \}]$ is a one-input one-output net if $M[S_2]$ is a one-input one-output net.
\end{lem}

\begin{proof}
We start with proving the one-input property. The proof proceeds by contradiction, so we start with assuming that $M'[(S_2 \setminus S_1) \cup \{ n_1 \}]$ does not have the one-input property, i.e., there are two distinct nodes $n_3$ and $n_5$ in $M'[(S_2 \setminus S_1) \cup \{ n_1 \}]$ with edges $(n_2, n_3)$ and $(n_4, n_5)$ in $M'$ such that $n_2, n_4 \not\in (S_2 \setminus S_1) \cup \{ n_1 \}$. We then consider the following two cases for the distinct nodes $n_3$ and $n_5$: (i) one of them is equal to $n_1$ and (ii) neither is equal to $n_1$.
\begin{enumerate}[(i)]
  
  \item This case is illustrated in Figures~\ref{fig:one-input-one-output-C} (a). Without loss of generality, we can assume that $n_5 = n_1$ and $n_3 \neq n_1$. By definition of contraction, there is then in $M$ a node $n_6$ that is an input node of $M[S_1]$ and an edge $(n_4, n_6)$. We consider the cases (1) all input nodes of $M[S_1]$ are in $S_1 \setminus S_2$ and (2) there is at least one input node of $M[S_1]$ that is in $S_1 \cap S_2$:  
  
    \begin{enumerate}[(1)]

      \item Because $M[S_1]$ is well-connected, the nodes in $S_1 \cap S_2$ must be reachable via a path from an input node of $M[S_1]$ in $S_1 \setminus S_2$. Take one such path, and consider the first edge $(n_7, n_8)$ on that path such that $n_7 \in S_1 \setminus S_2$ and $n_8 \in S_1 \cap S_2$. It follows that $n_8$ is an input node of $M[S_2]$, because it has an incoming edge from outside $S_2$. The node $n_3$ is also an input node of $M[S_2]$ because of the edge $(n_2, n_3)$. Since $n_8 \in S_1 \cap S_2$ and $n_3 \in S_2 \setminus S_1$, they are distinct nodes, and so $M[S_2]$ does not have the one-input property.

      \item Because $S_1$ is well-nested in $M$, we can reassign $n_6$ such that it is in $S_1 \cap S_2$ and still has an incoming edge from $n_4$. It then follows that $n_6$ is an input node of $M[S_2]$, as well as $n_3$. Since $n_6 \in S_1 \cap S_2$ and $n_3 \in S_2 \setminus S_1$, they are distinct nodes, and so $M[S_2]$ does not have the one-input property.    
      
    \end{enumerate}
  
    \item This case is illustrated in Figures~\ref{fig:one-input-one-output-C} (b). It follows that $n_3$ and $n_5$ are input nodes of $M[S_2]$, but then $M[S_2]$ does not have the one-input property.

\end{enumerate}

The one-output property can be shown to be preserved in a similar way, but with the direction of the edges reversed and the sets of input places replaced by the sets of output places and vice versa.
\end{proof}

\begin{figure}[htb]
\begin{center}
\resizebox{0.7\textwidth}{!}{%
\begin{tikzpicture}
    \tikzstyle{transition} = [rectangle,draw,minimum width=0.55cm, minimum height=0.55cm,fill=white]
    \tikzstyle{place} = [circle,draw,minimum width=0.55cm, minimum height=0.55cm,fill=white,inner sep=0.08cm]
    \tikzstyle{place-transition} = [rounded rectangle,draw,minimum width=1cm, minimum height=0.5cm,fill=white,inner sep=0.08cm]

\begin{scope}[shift={(0cm, -8cm)}] % second case, line (b)

\begin{scope}[shift={(1cm, 0cm)}]
   \node (S1) [cloud,cloud puffs=15,cloud puff arc=100, aspect=1.5,
          text opacity=1,
          minimum width=3cm,minimum height=3cm]
    {};

   \node (S2) [cloud,draw,cloud puffs=15,cloud puff arc=100, aspect=1.5,
          text opacity=1,
          minimum width=3cm,minimum height=3cm,
          right=-1.5cm of S1]
    {};
   \node [right=-1.5cm of S2] {$S_2 \setminus S_1$};

  % erase S1 from S2
   \node (ES1) [cloud, draw, fill=white, color=white, cloud puffs=15,cloud puff arc=100, aspect=1.5,
          text opacity=1,
          minimum width=3cm,minimum height=3cm] at (S1.center)
    {};
    
    % put back erased edge of S2
  \begin{scope}[]
    \pgftransformshift{\pgfpointanchor{S2}{center}}
    \pgfset{cloud puffs=15,cloud puff arc=100, aspect=1.5,
       minimum width=3cm,minimum height=3cm}
    \pgfnode{cloud}{center}{}{nodename}{\pgfusepath{clip}}

   \node [draw,cloud,cloud puffs=15,cloud puff arc=100, aspect=1.5,
       minimum width=3cm,minimum height=3cm] 
     at (S1.center) {};
  \end{scope}
     
   \node (n1) [place-transition, right=-1.5cm of S1] {$n_1$};  
   \node (n3) [place-transition, above right=-0.6cm and -0.6cm of S2] {$n_3$};  
   \node (n5) [place-transition, below right=-0.6cm and -0.6cm of S2] {$n_5$};  

   \node (n4) [place-transition, below right=0.2cm and 0.2cm of S2] {$n_4$};  
   \node (n2) [place-transition, above right=0.2cm and 0.2cm of S2] {$n_2$};  
   
   \path (n2) edge[-latex] (n3);
   \path (n4) edge[-latex] (n5);
\end{scope}

\node at (0cm, 0cm) {(b)};

\node at (5cm, 0cm) {\Large $\Rightarrow$};

\begin{scope}[shift={(7cm, 0cm)}]
   \node (S1) [cloud,draw,cloud puffs=15,cloud puff arc=100, aspect=1.5,
          text opacity=1,
          minimum width=3cm,minimum height=3cm]
    {};
   \node [left=-1cm of S1] {$S_1$};

   \node (S2) [cloud,draw,cloud puffs=15,cloud puff arc=100, aspect=1.5,
          text opacity=1,
          minimum width=3cm,minimum height=3cm,
          right=-1.5cm of S1]
    {};
   \node [right=-1cm of S2] {$S_2$};
     
   \node (n3) [place-transition, above right=-0.6cm and -0.6cm of S2] {$n_3$};  
   \node (n5) [place-transition, below right=-0.6cm and -0.6cm of S2] {$n_5$};  

   \node (n4) [place-transition, below right=0.2cm and 0.2cm of S2] {$n_4$};  
   \node (n2) [place-transition, above right=0.2cm and 0.2cm of S2] {$n_2$};  
   
   \path (n2) edge[-latex] (n3);
   \path (n4) edge[-latex] (n5);
\end{scope}

\end{scope} % end of second line

\begin{scope}[shift={(0cm, 0cm)}] % first case, line (a)

\begin{scope}[shift={(1cm, 0cm)}]

   \node (S1) [cloud,cloud puffs=15,cloud puff arc=100, aspect=1.5,
          text opacity=1,
          minimum width=3cm,minimum height=3cm]
    {};

   \node (S2) [cloud,draw,cloud puffs=15,cloud puff arc=100, aspect=1.5,
          text opacity=1,
          minimum width=3cm,minimum height=3cm,
          right=-1.5cm of S1]
    {};
   \node [right=-1.5cm of S2] {$S_2 \setminus S_1$};

  % erase S1 from S2
   \node (ES1) [cloud, draw, fill=white, color=white, cloud puffs=15,cloud puff arc=100, aspect=1.5,
          text opacity=1,
          minimum width=3cm,minimum height=3cm] at (S1.center)
    {};
    
    % put back erased edge of S2
  \begin{scope}[]
    \pgftransformshift{\pgfpointanchor{S2}{center}}
    \pgfset{cloud puffs=15,cloud puff arc=100, aspect=1.5,
       minimum width=3cm,minimum height=3cm}
    \pgfnode{cloud}{center}{}{nodename}{\pgfusepath{clip}}

   \node [draw,cloud,cloud puffs=15,cloud puff arc=100, aspect=1.5,
       minimum width=3cm,minimum height=3cm] 
     at (S1.center) {};
  \end{scope}
     
   \node (n1) [place-transition, right=-1.5cm of S1] {$n_1$};  
   \node [below=0cm of n1] {$=n_5$};

   \node (n4) [place-transition, above left=0.6cm and 0.6cm of n1] {$n_4$};  
   \node (n3) [place-transition, above right=-0.6cm and -0.6cm of S2] {$n_3$};  

   \node (n2) [place-transition, above right=0.2cm and 0.2cm of S2] {$n_2$};  
   
   \path (n2) edge[-latex] (n3);
   \path (n4) edge[-latex] (n1);

\end{scope}

\node at (0cm, 0cm) {(a)};

\node [rotate=40] at (4.9cm, 1.5cm) {\Large $\Rightarrow$};

\begin{scope}[shift={(7cm, -2.5cm)}]
   \node (S1) [cloud,draw,cloud puffs=15,cloud puff arc=100, aspect=1.5,
          text opacity=1,
          minimum width=3cm,minimum height=3cm]
    {};
   \node [left=-1cm of S1] {$S_1$};

   \node (S2) [cloud,draw,cloud puffs=15,cloud puff arc=100, aspect=1.5,
          text opacity=1,
          minimum width=3cm,minimum height=3cm,
          right=-1.5cm of S1]
    {};
   \node [right=-1cm of S2] {$S_2$};
     
   \node (n6) [place-transition, right=-1.2cm of S1] {$n_6$};  

   \node (n4) [place-transition, above left=0.2cm and 0.2cm of S1] {$n_4$};  
   \node (n3) [place-transition, above right=-0.6cm and -0.6cm of S2] {$n_3$};  

   \node (n2) [place-transition, above right=0.2cm and 0.2cm of S2] {$n_2$};  
   
   \path (n2) edge[-latex] (n3);
   \path (n4) edge[-latex] (n6);
   
   \node [below right=0.7 cm and -3 cm of S1] {Not all input nodes of $M[S_1]$ in $S_1 \setminus S_2$};
\end{scope}

\node [rotate=-40] at (4.9cm, -1.5cm) {\Large $\Rightarrow$};

\begin{scope}[shift={(7cm, 2.5cm)}]
   \node (S1) [cloud,draw,cloud puffs=15,cloud puff arc=100, aspect=1.5,
          text opacity=1,
          minimum width=3cm,minimum height=3cm]
    {};
   \node [left=-1cm of S1] {$S_1$};

   \node (S2) [cloud,draw,cloud puffs=15,cloud puff arc=100, aspect=1.5,
          text opacity=1,
          minimum width=3cm,minimum height=3cm,
          right=-1.5cm of S1]
    {};
   \node [right=-1cm of S2] {$S_2$};
     
   \node (n6) [place-transition, above left=-0.6cm and -0.6cm of S1] {$n_6$};  

   \node (n3) [place-transition, above right=-0.6cm and -0.6cm of S2] {$n_3$};  
   \node (n2) [place-transition, above right=0.2cm and 0.2cm of S2] {$n_2$};  

   \node (n4) [place-transition, above left=0.2cm and 0.2cm of S1] {$n_4$};  

   \path (n2) edge[-latex] (n3);
   \path (n4) edge[-latex] (n6);
   
   \node [below right=0.7 cm and -3 cm of S1] {All input nodes of $M[S_1]$ in $S_1 \setminus S_2$};
\end{scope}

\node at (11cm, 2.5cm) {\Large $\Rightarrow$};

\begin{scope}[shift={(13cm, 2.5cm)}]
   \node (S1) [cloud,draw,cloud puffs=15,cloud puff arc=100, aspect=1.5,
          text opacity=1,
          minimum width=3cm,minimum height=3cm]
    {};
   \node [left=-1cm of S1] {$S_1$};

   \node (S2) [cloud,draw,cloud puffs=15,cloud puff arc=100, aspect=1.5,
          text opacity=1,
          minimum width=3cm,minimum height=3cm,
          right=-1.5cm of S1]
    {};
   \node [right=-1cm of S2] {$S_2$};
     
   \node (n6) [place-transition, above left=-0.6cm and -0.6cm of S1] {$n_6$};  
   \node (n3) [place-transition, above right=-0.6cm and -0.6cm of S2] {$n_3$};  
   \node (n2) [place-transition, above right=0.2cm and 0.2cm of S2] {$n_2$};  
   \node (n4) [place-transition, above left=0.2cm and 0.2cm of S1] {$n_4$};  

   \node (n7) [place-transition, below left=-0.6cm and -0.6cm of S1] {$n_7$};  
   \node (n8) [place-transition, below right=-0.9cm and -0.5cm of S1] {$n_8$};     

   \path (n2) edge[-latex] (n3);
   \path (n4) edge[-latex] (n6);
   \path (n7) edge[-latex] (n8);
\end{scope}

\end{scope} % end of first case, line (a)

\end{tikzpicture}
}
\end{center}
\caption{\label{fig:one-input-one-output-C}Illustration of proof of preservation of the one-input and one-output properties in case (C)}
\end{figure}
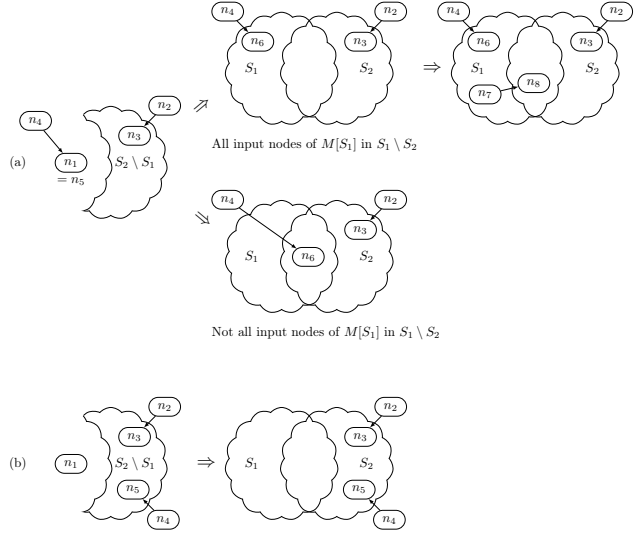

This concludes the lemmas that show that (1) the well-nestedness of $S_2$ in $M$ is preserved by $S_2 \setminus S_1$ in $M'$ and (2) all the defining properties of basic AND-OR nets are also preserved in $M'[S_2 \setminus S_1]$ if $M[S_2]$ had them. This leads us to the following lemma.

\begin{lem}[preservation of contractibility] \label{lem:contractibility-C}
The subnet $(S_2 \setminus S_1) \cup \{ n_1 \}$ is contractible in $M'$.
\end{lem}

\begin{proof}
The proof proceeds analogous to that of Lemma~\ref{lem:contractibility-B}, which show preservation of contractibility for case (B).  We show the following three claims: (1) $M'[(S_2 \setminus S_1) \cup \{ n_1 \}]$ is well-nested, (2) $M'[(S_2 \setminus S_1) \cup \{ n_1 \}]$ is a WF net and (3) $M'[(S_2 \setminus S_1) \cup \{ n_1 \}]$ belongs to the same basic AND-OR class as $M[S_2]$. Claim (1) follows from Lemma~\ref{lem:well-nested-C}. Claim (2) follows from Lemma~\ref{lem:IO-type-C} (preservation of I/O type), Theorem~\ref{thm:subnet-wf-char} (every I/O-consistent subnet of a WF net is a WF net) and Theorem~\ref{thm:contraction-correctness} (the result of a contraction in a WF net is a WF net). Claim (3) follows from the preservation of the AND-OR defining properties of $M[S_2]$ by $M'[(S_2 \setminus S_1) \cup \{ n_1 \}]$, as stated by Lemma~\ref{lem:IO-type-C} (preservation of I/O type), Lemma~\ref{lem:AND-OR-C} (preservation of the AND and OR properties), Lemma~\ref{lem:acyclic-C} (preservation of acyclicity) and Lemma~\ref{lem:one-input-one-output-C} (preservation of one-input and one-output properties). 
\end{proof}

We now turn to the question if the order of contractions of $S_1$ and $S_2$ influences the final result.

\begin{lem}[commutativity of contraction] \label{lem:commutativity-C}
Let $M''_1$ be the result of first contracting $S_1$ into node $n_1$ resulting in $M'_1$ followed by contracting $(S_2 \setminus S_1) \cup \{ n_1 \}$ into node $n_3$, and let $M''_2$ be the result of first contracting $S_2$ into node $n_2$ resulting in $M'_2$ followed-by contracting $(S_1 \setminus S_2) \cup \{ n_2 \}$ into node $n_3$. Then $M''_1 = M''_2$.
\end{lem}

\begin{proof}
The WF nets $M''_1$ and $M''_2$ have the same sets of nodes, since in both cases all the nodes of $S_1$ and $S_2$ are removed, and the new node $n_3$ is added. The types of the preserved nodes are not changed by contractions. In addition, the type of $n_3$ will in both cases be the same, which can be shown as follows. In $M''_1$ the type of $n_3$ is the I/O type of $M'_1[(S_2 \setminus S_1) \cup \{ n_1 \}]$, which by 
Lemma~\ref{lem:IO-type-C} is the I/O type of $S_2$ in $M$. Similarly, in $M''_2$ the type of $n_3$ is equal to the I/O type of $S_1$. Since the I/O types of $S_1$ and $S_2$ are assumed to be the same, the type of $n_3$ in $M''_1$ and in $M''_2$ is also the same. 

Since we are only considering internal contractions, the sets of input and output nodes will not change, and therefore be in both $M''_1$ and $M''_2$ identical to those in $M$.

So what remains to be shown is that in $M''_1$ and $M''_2$ the node $n_3$ has the same incoming and outgoing edges. For incoming edges it is easy to see that in both $M''_1$ and $M''_2$ it holds that an edge $(n_4, n_3)$ exists iff $n_4 \not\in S_1 \cup S_2$ and there is a node $n_5 \in S_1 \cup S_2$ such that the edge $(n_4, n_5)$ exists in $M$. A similar argument, but with the direction of the edges reversed, can be made for outgoing edges.
\end{proof}

\subsection{Nested subnets}

In this subsection we consider case (D) in Figure~\ref{fig:cases-of-proof}.  Therefore, we assume that $S_1$ and $S_2$ are contractible subnets in a WF net $M$ such that $S_2 \subseteq S_1$. Let us also assume that $M'$ is the result of contracting $S_2$ into the node $n_2$ and $M''$ is the result of contracting $S_1$ into the node $n_1$. Recall that we need to show that (1) after contraction of $S_2$ into $n_2$, the subnet $(S_1 \setminus S_2) \cup \{ n_2 \}$ is contractible in $M'$ and (2) if we contract first $S_2$ into $n_2$ and then $(S_1 \setminus S_2) \cup \{ n_2 \}$ into $n_1$, the result is the same as when we contract first $S_1$ into $n_1$. 

Like in the previous subsection, the proof of (1) will consists of a list of lemmas that show that the different properties that define contractibility are all preserved for subnet $(S_1 \setminus S_2) \cup \{ n_2 \}$ by the contraction of $S_2$. We start with the property of well-nestedness, and then move on to the defining properties of WF nets and basic AND-OR nets, and close of with a lemma showing (2).

\begin{lem}[preservation of well-nestedness] \label{lem:well-nested-D}
The subnet $(S_1 \setminus S_2) \cup \{ n_2 \}$ is well-nested in $M'$.
\end{lem}

\begin{proof}
This proof proceeds the same as for the analogous Lemma~\ref{lem:well-nested-C} for case (C), except that $S_1$ and $S_2$ reverse roles and $n_1$ is replaced with $n_2$. This is because that proof does in fact not use the assumptions that the nets are not nested and that their I/O types are the same.
\end{proof}

\begin{lem}[preservation of I/O type] \label{lem:IO-type-D}
The input nodes and output nodes of $M'[(S_1 \setminus S_2) \cup \{ n_2 \}]$ have the same type as the input nodes and output nodes of $M[S_2]$, respectively.
\end{lem}

\begin{proof}
This proof proceeds the same as for the analogous Lemma~\ref{lem:IO-type-C} for case (C), except that $S_1$ and $S_2$ reverse roles and $n_1$ is replaced with $n_2$. This is because that proof does not use the assumptions that the nets are not nested and that their I/O types are the same.
\end{proof}

\begin{lem}[preservation of the AND and the OR property] \label{lem:AND-OR-D}
The net $M'[(S_1 \setminus S_2) \cup \{ n_2 \}]$ has the AND (OR) property if $M[S_1]$ has the AND (OR) property.
\end{lem}

\begin{proof}
This proof proceeds the same as for the analogous Lemma~\ref{lem:AND-OR-C} for case (C), except that (1) $S_1$ and $S_2$ reverse roles and $n_1$ is replaced with $n_2$ and (2) wherever it uses Lemma~\ref{lem:edge-via-intersection}, which states that there are no edges between nodes in $S_1 \setminus S_2$ and nodes in $S_2 \setminus S_1$, this is replaces with the observation that by assumption there are no nodes in $S_2 \setminus S_1$. This works because that proof does not use the assumptions that the nets are not nested and that their I/O types are the same.
\end{proof}

\begin{lem}[preservation of acyclicity] \label{lem:acyclic-D}
The net $M'[(S_1 \setminus S_2) \cup \{ n_2 \}]$ is acyclic if $M[S_1]$ is acyclic.
\end{lem}

\begin{proof}
The proof proceeds by contradiction, so we assume that there is a cycle in $M'[(S_1 \setminus S_2) \cup \{ n_2 \}]$, and then show that it follows that there is a cycle in $M[S_1]$. So let us assume that there is a cycle in $M'[(S_1 \setminus S_2) \cup \{ n_2 \}]$. It is clear that if the cycle does not contain $n_2$, this cycle is also present in $M[S_1]$, so we continue under the assumption that $n_2$ is contained in the cycle. We can assume that $n_2$ appears exactly once in the cycle, since if a cycle with multiple occurrences of $n_2$ exist, then there is also one with exactly one occurrence. In this cycle there must be an edge that leaves $n_2$, say $(n_2, n_3)$, and an edge that arrives in $n_2$, say $(n_4, n_2)$. Note that there is in $M'[S_1 \setminus S_2]$ a path from $n_3$ to $n_4$, possibly consisting of only one node, in which case these nodes are one and the same.

Since $M[S_2]$ is a WF net, it will contain a path from an input node, say $n_5$, to an output node, say $n_6$. By the definition of contraction, and since $S_2$ is well-nested in $M$, it follows that there are edges $(n_4, n_5)$ and $(n_6, n_3)$ in $M$. It then follows that there is a cycle in $M[S_1]$ through the nodes $n_3$, $n_4$, $n_5$ and $n_6$.
\end{proof}

\begin{lem}[preservation of one-input one-output property] \label{lem:one-input-one-output-D}
The net $M'[(S_1 \setminus S_2) \cup \{ n_2 \}]$ is a one-input one-output net if $M[S_1]$ is a one-input one-output net.
\end{lem}

\begin{proof}
This proof proceeds the same as for the analogous Lemma~\ref{lem:one-input-one-output-C} for case (C), except that $S_1$ and $S_2$ reverse roles and $n_1$ is replaced with $n_2$. This is because that proof does not use the assumptions that the nets are not nested and that their I/O types are the same.
\end{proof}

This concludes the lemmas that show that (1) the well-nestedness of $M[S_2]$ is preserved by $M'[S_2 \setminus S_1]$ and (2) all the defining properties of basic AND-OR nets are also preserved in $M'[S_2 \setminus S_1]$ if $M[S_2]$ had them. This leads us to the following lemma.

\begin{lem}[preservation of contractibility] \label{lem:contractibility-D}
The subnet $(S_1 \setminus S_2) \cup \{ n_2 \}$ is contractible in $M'$.
\end{lem}

\begin{proof}
The proof proceeds analogous to that of Lemma~\ref{lem:contractibility-D}, which shows preservation of contractibility for case (C).  We show the following three claims: (1) $M'[(S_1 \setminus S_2) \cup \{ n_2 \}]$ is well-nested, (2) $M'[(S_1 \setminus S_2) \cup \{ n_2 \}]$ is a WF net and (3) $M'[(S_1 \setminus S_2) \cup \{ n_2 \}]$ belongs to the same basic AND-OR class as $M[S_1]$. Claim (1) follows from Lemma~\ref{lem:well-nested-D}. Claim (2) follows from Lemma~\ref{lem:IO-type-D} (preservation of I/O type), Theorem~\ref{thm:subnet-wf-char} (every I/O-consistent subnet of a WF net is a WF net) and Theorem~\ref{thm:contraction-correctness} (the result of a contraction in a WF net is a WF net). Claim (3) follows from the preservation of the AND-OR defining properties of $M[S_2]$ by $M'[(S_1 \setminus S_2) \cup \{ n_2 \}]$, as stated by Lemma~\ref{lem:IO-type-D} (preservation of I/O type), Lemma~\ref{lem:AND-OR-D} (preservation of the AND and OR properties), Lemma~\ref{lem:acyclic-D} (preservation of acyclicity) and Lemma~\ref{lem:one-input-one-output-D} (preservation of one-input and one-output properties). 
\end{proof}

We now turn to the question if the order of contractions of $S_1$ and $S_2$ influences the final result.

\begin{lem}[commutativity of contraction] \label{lem:commutativity-D}
Let $M''_1$ be the result of contracting $S_1$ into node $n_1$, and let $M''_2$ be the result of first contracting $S_2$ into node $n_2$ resulting in $M'$ followed by contracting $(S_1 \setminus S_2) \cup \{ n_2 \}$ into node $n_1$. Then $M''_1 = M''_2$.
\end{lem}

\begin{proof}
The WF nets $M''_1$ and $M''_2$ have the same sets of nodes, since in both cases all the nodes of $S_1$ and $S_2$ are removed, and the new node $n_1$ is added. The types of the preserved nodes are not changed by contractions. In addition, the type of $n_1$ will in both cases be the same, which can be shown as follows. In $M''_1$ the type of $n_1$ is the I/O type of $S_1$. Similarly, in $M''_2$ the type of $n_1$ is equal to the I/O type of $M'[(S_1 \setminus S_2) \cup \{ n_2\}]$, which by Lemma~\ref{lem:IO-type-D} has the same I/O type as $M[S_1]$.

Since we are only considering internal contractions, the sets of input and output nodes will not change, and therefore be in both $M''_1$ and $M''_2$ identical to those in $M$.

So what remains to be shown is that in $M''_1$ and $M''_2$ the node $n_1$ has the same incoming and outgoing edges. For incoming edges it is easy to show that in both $M''_1$ and $M''_2$ it holds that an edge $(n_3, n_1)$ exists iff $n_3 \not\in S_1 \cup S_2$ and there is a node $n_4 \in S_1 \cup S_2$ such that the edge $(n_3, n_4)$ exists in $M$. A similar argument, but with the direction of the edges reversed, can be made for outgoing edges.
\end{proof}

\subsection{Combining all cases}

In the preceding subsections it was shown for the cases (A), (B), (C) and (D) from Figure~\ref{fig:cases-of-proof} that for internal contractions we have local confluence. We now combine these results to show that this holds in general for every two internal AND-OR contractions.

% JS: is this clear that everywhere we mean AND-OR contractions not contractions?
% JH: was implicit in using \contr and \intcontr, but made a little more explicit now

% JS: is this clear in which theorems S1 and S2 are of basic AND-OR types?
% JH: in all of them, as stated in beginning of section

\begin{thm}[local confluence of the internal AND-OR contraction relation] \label{thm:local-confluence-internal}
Given a WF net $M$ and two AND-OR contractions $M \intcontr M'_1$ and $M \intcontr M'_2$, there is a WF net $M''$ such that $M'_1 \intcontr^* M''$ and $M'_2 \intcontr^* M''$.
\end{thm}

\begin{proof}
Let $M[S_1]$ and $M[S_2]$ be the subnets of $M$ that were contracted for $M \intcontr M'_1$ and $M \intcontr M'_2$, respectively. The easiest case to consider is the case where these nets to do not share nodes and are not connected by edges. In that case, the contraction of one net does not influence the other subnet or its connected edges, and so each of $M[S_1]$ and $M[S_2]$ will stay contractible if the other is contracted, and the final result after contracting both nets does not depend upon the order of contraction.

The remaining cases where defined in the beginning of this section, and for convenience we recall them here. In each case the following holds for $M[S_1]$ and $M[S_2]$:
\begin{enumerate}[(A)]

\item \emph{They do not share nodes, but there are edges from nodes in one subnet to nodes in another.} By Lemma~\ref{lem:contractibility-A} it holds that after contracting $S_1$ into $n_1$ resulting in $M'_1$ the subnet $S_2$ is contractible in $M'_1$. By symmetry the same holds for the subnet $S_1$ if we contract $S_2$ into $n_2$ resulting in $M'_2$. Moreover, by Lemma~\ref{lem:commutativity-A} it holds that the result of contracting first $S_1$ into $n_1$ and then $S_2$ into $n_2$ is the same as when we first contract $S_2$ into $n_2$ and then $S_1$ into $n_1$.

\item \emph{The subnets share nodes, the corresponding WF nets have different I/O types, and each subnet has nodes that are not contained in the other.} By Lemma~\ref{lem:contractibility-B} it holds that after contracting $S_1$ into $n_1$ resulting in $M'_1$ the subnet $S_2 \setminus S_1$ is contractible in $M'_1$. By symmetry the same holds for the subnet $S_1 \setminus S_2$ in $M'_2$ if we contract $S_2$ into $n_2$ resulting in $M'_2$. Moreover, by Lemma~\ref{lem:commutativity-B} it holds that the result of contracting first $S_1$ into $n_1$ and then $S_2 \setminus S_1$ into $n_2$, is the same as when we first contract $S_2$ into $n_2$ and then $S_1 \setminus S_2$ into $n_1$.

\item \emph{The subnets share nodes, the corresponding WF nets have the same I/O type, and each subnet has nodes that are not contained in the other.} By Lemma~\ref{lem:contractibility-C} it holds that after contracting $S_1$ into $n_1$ resulting in $M'_1$ the subnet $(S_2 \setminus S_1) \cup \{ n_1 \}$ is contractible in $M'_1$. By symmetry the same holds for the subnet $(S_1 \setminus S_2) \cup \{ n_2 \}$ in $M'_2$ if we contract $S_2$ into $n_2$ resulting in $M'_2$. Moreover, by Lemma~\ref{lem:commutativity-C} it holds that the result of contracting first $S_1$ into $n_1$ and then $(S_2 \setminus S_1) \cup \{ n_1 \}$ into $n_3$ is the same as when we first contract $S_2$ into $n_2$ and then $(S_1 \setminus S_2) \cup \{ n_2 \}$ into $n_3$.

\item \emph{One subnet is contained in the other.} Without loss of generality we assume that $S_2 \subseteq S_1$. By Lemma~\ref{lem:contractibility-D} it holds that after contracting $S_2$ into $n_2$ resulting in $M'_2$ the subnet $(S_1 \setminus S_2) \cup \{ n_2 \}$ is contractible in $M'_2$. Moreover, by Lemma~\ref{lem:commutativity-D} it holds that the result of contracting $S_1$ into $n_1$ is the same as when we first contract $S_2$ into $n_2$ and then $(S_1 \setminus S_2) \cup \{ n_2 \}$ into $n_1$.

\end{enumerate}
Summarising, we have shown that in cases (A), (B) and (C) there is a WF net $M''$ such that $M'_1 \intcontr M''$ and $M'_2 \intcontr M''$, and so it follows that $M'_1 \intcontr^* M''$ and $M'_2 \intcontr^* M''$. Moreover, in case (D) we have shown that $M'_2 \intcontr M'_1$, and so there is indeed a WF net $M'' = M'_1$ such that $M'_1 \intcontr^* M''$ and $M'_2 \intcontr^* M''$.

\end{proof}

% JS: we use here the AND-OR nes of the nets that we contract, is this clear from the assumptions?
% JH: its in the definition of the contraction relation \contr that we use

% JS: is this clear?
% JH: Don't see an easy way to make it more clear without much work.

\begin{lem}[simulation with internal AND-OR contractions] \label{lem:simulation_with_internal}
	If $M_1$ and $M_2$ are tWF nets and $M'_1$ and $M'_2$ their place completions, it holds that $M_1 \contr M_2$ iff $M'_1 \intcontr M'_2$, and the same holds if $M_1$ and $M_2$ are pWF nets and $M'_1$ and $M'_2$ their transition completions.
\end{lem}

\begin{proof}
	We start by showing that if $M_1$ and $M_2$ are tWF nets and $M'_1$ and $M'_2$ their place completions, it holds that $M_1 \contr M_2$ iff $M'_1 \intcontr M'_2$. Let the special input and output nodes in a place completion be $p_i$ and $p_o$.

	It is easy to observe that the following two statements are equivalent: (a) $n$ is an input node in $M_1$ and (b) there is an edge $(p_i, n)$ in $M'_1$. Similarly, the following two are also equivalent: (c) $n$ is an output node in $M_1$ and (d) there is an edge $(n, p_o)$ in $M'_1$.

	It follows that a subnet $S$ is well-nested in $M_1$ iff $S$ is internal well-nested in $M'_1$. After all, a set of nodes in $M_1$ have all the same incoming edges in $M'_1$ from nodes outside $S$ iff (1) they have all the same incoming edges in $M_1$ from outside $S$ and (2) either are all input nodes of $M_1$ or are all not input nodes of $M_1$. Similarly, they have all the same outgoing edges in $M'_1$ to nodes outside $S$ iff (1) they have all the same outgoing edges in $M_1$ to nodes outside $S$ and (2) either are all output nodes of $M_1$ or are all not output nodes of $M_1$.

	Moreover, a subnet $S$ defines a basic AND-OR net in $M_1$ iff it does in $M'_1$. This follows from the observation that in the definitions of the AND and the OR property, having an incoming edge from outside $S$ and being an input node of $M_1$ are treated equivalently. So we can conclude that subnet $S$ is contractible in $M_1$ iff $S$ is internally contractible in $M'_1$.

	If we then consider what happens if we contract $S$ in $M_1$ and $S$ in $M'_1$, then the same happens in terms of changes to edges and nodes, except that where in $M_1$ the new node $n_1$ becomes an input or output node, it holds in $M'_1$ that the edge $(p_i, n_1)$ or the edge $(n_1, p_o)$ is added, respectively. This indeed commutes with taking the place completion.

	The proof for the case where $M_1$ and $M_2$ are pWF nets, with $M'_1$ and $M'_2$ their transition completions, is similar, but with the roles of transitions and places interchanged.
\end{proof}

This leads us to the proof of Theorem~\ref{thm:local-confluence}, which states that the AND-OR contraction relation has the local confluence property. We recall its content here before giving its proof.

\medskip

% JS: we don't need the * here, i.e., we can prove a stronger version. why was there * where we defined this thm?
% JH: actually we do, since the lemma's show that sometimes we are already there, so after 0 steps

\noindent
\textbf{Theorem \ref{thm:local-confluence}. (local confluence of the AND-OR contraction relation)} \\
For all WF nets $M$, $M_1$ and $M_2$ it holds that if $M \contr M_1$ and $M \contr M_2$, then there is a WF net $M_3$ such that $M_1 \contr^* M_3$ and $M_2 \contr^* M_3$.

\medskip

\begin{proof}
	Let $M$, $M_1$ and $M_2$ be WF nets such that it holds $M \contr M_1$ and $M \contr M_2$. We are going to first deal with the case where $M$ is a tWF net. By Theorem~\ref{thm:contraction-correctness} it follows that $M_1$ and $M_2$ are also tWF nets. By Lemma~\ref{lem:simulation_with_internal} it follows that $\textbf{pc}(M) \intcontr \textbf{pc}(M_1)$ and $\textbf{pc}(M) \intcontr \textbf{pc}(M_2)$. Now by Theorem~\ref{thm:local-confluence-internal} we know that there is a WF net $M''$ such that $\textbf{pc}(M_1) \intcontr^* M''$ and $\textbf{pc}(M_2) \intcontr^* M''$. Again by Theorem~\ref{thm:contraction-correctness} and induction on the number of contractions, we know that $M''$ is a pWF net. Furthermore, as it is a result of internal AND-OR contractions applied to a one-input one-output net, it is a one-input one-output net and as there were distinct input and output nodes and some other nodes ($M$ cannot be empty) before the contraction, this is still the case in $M''$. Thus $M'' = \textbf{pc}(M_3)$ for some net $M_3$. Using the Lemma~\ref{lem:simulation_with_internal} again, from $\textbf{pc}(M_1) \intcontr^* M'' = \textbf{pc}(M_3)$ and $\textbf{pc}(M_2) \intcontr^* M'' = \textbf{pc}(M_3)$, we get with induction on the number of contraction steps that $M_1 \contr^* M_3$ and $M_2 \contr^* M_3$.

	The proof for the case were $M$ is a pWF net proceeds similar, but with taking transition completion rather that place completion.
\end{proof}

\section{A polynomial-time reduction algorithm\label{sect:Algorithm-AND-OR-nets}}

In this section we present a concrete AND-OR reduction algorithm for computing 
the net-based AND-OR reduction procedure. In addition we show
how this algorithm can be used to verify if a certain WF net is an AND-OR net. The algorithm is based
on the result of Theorem~\ref{thm:unique-result-net-reduction} which states
that the result of this reduction is unique up to isomorphism, no matter how
we select the subnets to contract. This allows the algorithm to proceed without
backtracking and essentially just repeat a process where it continues to 
look for contractible non-trivial subnets, and contract them, until
no more such subnets can be found. The efficiency of this algorithm is improved
further by an algorithm for identifying such subnets within polynomial time. As a result
the whole reduction algorithm runs within polynomial time.

\begin{algorithm}[htb]
\label{alg:reduce}
\SetAlgoVlined
\LinesNumbered
\KwIn{a net WF net $M = ( P, T, F, I, O )$}
\KwOut{a net-based AND-OR reduction of $M$}
\BlankLine
\ForEach{$(n_1, n_2) \in (P \times P) \cup (T \times T)$}{ \label{alg1:forloop}
     \If{$n_1=n_2$ and $n_1 \in P$ and there exists $t \in T$ s.t. $t \bullet_M = \bullet_M t = \{n_1\}$} { \label{alg1:loop-test}
       \Return reduce(contractSubnet($M$, $\{n_1, t\}$))\;
     }
	 \If{$n_1 \neq n_2$ and $\bullet_M n_1 = \bullet_M n_2$ and $n_1 \bullet_M = n_2 \bullet_M$} {\label{alg1:parallel-test}
       \Return reduce(contractSubnet($M$, $\{n_1, n_2 \}$))\;
     }
     \If{$n_1 \neq n_2$ and $n_2$ is reachable from $n_1$} { \label{alg1:reach-test}
       $N \leftarrow$ expand($M$, $n_1$, $n_2$)\;
       \If{$N \neq$ null}{
         \Return reduce(contractSubnet($M$, $N$))\;
       }
     }
}
\Return $M$\;
\caption{reduce($M$)}
\end{algorithm}

\subsection{The reduction algorithm}

Non-trivial subnets, i.e., subnets with more than one node, for contraction can be found in the following way. We iterate in line \ref{alg1:forloop} of Algorithm~\ref{alg:reduce} over all the pairs of nodes $(n_1,n_2)$ of the same type, i.e., pairs of places or pairs of transitions. We distinguish three cases. We start in
line \ref{alg1:loop-test} with contracting loops, where by a loop we mean a place with
a transition attached to it such that this transition has exactly two
edges, one incoming from the place and one outgoing to the place. Then in line \ref{alg1:parallel-test} we contract pairs of distinct nodes if these share
the same incoming and outgoing edges. Finally, in line \ref{alg1:reach-test}, for the remaining
pairs of nodes we treat one node as an input and the other as an output.
We expand forward from the input node and backward from the output node. While
expanding we remember that we can discover new interface nodes. If
we end up with a subnet that is of a basic AND-OR class then we contract it.

It can be shown that if any non-trivial subnet can be contracted,
then at least one of our three cases will also find a subnet. This follows
from the observation that for a contractible subnet it holds either
(1) it has an input node and an output node that is distinct from
the input node but can be reached from it or (2) all input nodes are
also output nodes and vice versa. Case (1) is covered by the test
on line \ref{alg1:reach-test} where we assume that $n_1$ is an input node and $n_2$ an output node of the subnet we are looking for. For case (2) we can distinguish the sub-cases (a) there
are two or more input/output nodes or (b) there is exactly one input/output
node. In case (a) the test on line \ref{alg1:parallel-test} will apply. In case (b) it can
be observed that the nested net must be an 11pOR net if it is non-trivial. This is because it then contains a cycle,
and so cannot be an 11tAND or pAND net, and it also
cannot be a tOR net since the input/output node has incoming and outgoing edges. If it is an 11pOR net then it either satisfies the test on line \ref{alg1:loop-test}, or the net minus
the input/output place defines a well-nested tOR net containing more then one node and which is covered by either
the tests on line \ref{alg1:parallel-test} or \ref{alg1:reach-test}. So we can conclude that if there is a contractible
subnet then there is a pair $(n_1,n_2)$ that satisfies at lease one of the tests
on line \ref{alg1:loop-test}, \ref{alg1:parallel-test}, or \ref{alg1:reach-test}.

The loop runs in quadratic time. Each time we do a contraction
the number of nodes decreases, so it suffices to show that we can do
the expansion in polynomial time. Observe first that we can easily
check in polynomial time if the output node is reachable from an input
node or even enumerate the nodes in such a way that for each input
node we iterate only over the output nodes reachable from it.

The expansion can be done as in Algorithm~\ref{alg:expand}. Given
an input and output node pair we traverse the net as in a breadth
or depth-first graph search algorithm. The initial nodes for the traversal
are the input and output nodes provided as parameters, yet we make
sure that we traverse forward only if the currently inspected node
is not an output node (see lines \ref{alg2:input-test} through \ref{alg2:add-pred}) and that we traverse backward
only if the currently inspected node is not an input node (see lines \ref{alg2:output-test} through \ref{alg2:add-foll}). The newly encountered nodes are tested to check if they can be input (output) nodes of the subnet we are looking for. 
This is done by testing if (1) they have incoming (outgoing) edges
from (to) exactly the same nodes as the initial input (output) node and (2) are input node of $M$ iff the initial input node is an input node of $M$.
When this test succeeds, the node is added to the input (output) set of the constructed net (see lines
\ref{alg2:add-inputs} and \ref{alg2:add-outputs}).

\begin{algorithm}[htb]
\label{alg:expand}
\small
\SetCommentSty{mycommfont} 
\SetAlgoVlined
\LinesNumbered
\SetSideCommentRight
\KwIn{a WF net $M = ( P, T, F, I, O )$ and nodes $i, o \in P$ or $i, o \in T$}
\KwOut{a set $S$ of nodes of $M$ such that $M[S]$ is a well-nested subnet in $M$ that is a WF net of a basic AND-OR class with input node $i$ and output node $o$, or $null$ if such a subnet does not exist}
\BlankLine
$( S, I_S, O_S ) \leftarrow ( \{ i, o \},  \{ i \}, \{ o \} )$\tcp*{init.\ $S$ and input/output sets of $M[S]$}
$TBA \leftarrow \{ i,  o \}$\tcp*{nodes to be analysed}
% $PC \leftarrow \emptyset$\tcp*{initialise set possible basic AND-OR classes} \label{alg2:init-pc}
\lIf(\tcp*[f]{init. set of possible}){$ i, o \in P $}{$PC \leftarrow \{ \textrm{11pOR}, \textrm{pAND} \}$}\label{alg2:init-pc}
\lIf(\tcp*[f]{basic AND-OR classes}){$i, o \in T$}{$PC \leftarrow \{ \textrm{11tAND}, \textrm{tOR} \}$}\label{alg2:end-init-pc}
\While{$TBA \neq \emptyset$ and $PC \neq \emptyset$} {
       $n' \leftarrow$ pick and remove from $TBA$\;
%       	\lIf(\tcp*[f]{interf. nodes $M$ reached?}){$n' \in I \cup O$}{\Return null} \label{alg2:interface-reached}
			\uIf(\tcp*[f]{is $n'$ input node?}){${\bullet_M n'} = {\bullet_M i}$ and $n' \in I \Leftrightarrow i \in I$}{ \label{alg2:input-test}
				\lIf(\tcp*[f]{new input}){$n' \not\in I_S$}{$PC \leftarrow PC \setminus \{ \textrm{11pOR}, \textrm{11tAND} \}$} \label{alg2:extra-input}
				$I_S \leftarrow I_S \cup \{ n' \}$\; \label{alg2:add-inputs}
			}
			\Else{
				$TBA \leftarrow TBA \cup ({\bullet_M n'} \setminus S)$\tcp*{add predecessors}
				$S \leftarrow S \cup {\bullet_M n'}$\; \label{alg2:add-pred}
			}
			\uIf(\tcp*[f]{is $n'$ output node?}){${n' \bullet_M} = {o \bullet_M}$ and $n' \in O \Leftrightarrow o \in O$}{\label{alg2:output-test}
				\lIf(\tcp*[f]{new output}){$n' \not\in O_S$}{$PC \leftarrow PC \setminus \{ \textrm{11pOR}, \textrm{11tAND} \}$} \label{alg2:extra-output}
				$O_S \leftarrow O_S \cup \{ n' \}$\; \label{alg2:add-outputs} 
			}
			\Else{
				$TBA \leftarrow TBA \cup ({n' \bullet_M} \setminus S)$\tcp*{add followers}
				$S \leftarrow  S \cup {n' \bullet_M}$\;  \label{alg2:add-foll}
			}
			\If{	$(n'  \not\in I_S \wedge |{\bullet_{M[S]} n'}| \neq 1)$ or 
					$(n' \in I_S \wedge |{\bullet_{M[S]} n'}| \neq 0)$ or \\
					$(n' \not\in O_S \wedge |{n' \bullet_{M[S]}}| \neq 1)$ or 
					$(n' \in O_S \wedge |{n' \bullet_{M[S]}}| \neq 0)$}{
				\lIf(\tcp*[f]{not AND net}){$n' \in P$}{ $PC \leftarrow PC \setminus \{ \textrm{pAND}, \textrm{11tAND} \}$} \label{alg2:no-and-prop}
				\lIf(\tcp*[f]{not OR net}){$n' \in T$}{ $PC \leftarrow PC \setminus \{ \textrm{tOR}, \textrm{11pOR} \}$} \label{alg2:no-or-prop}
			}
}
\lIf(\tcp*[f]{AND is acyclic}){$M[S]$ cyclic}{$PC \leftarrow PC \setminus \{ \textrm{pAND}, \textrm{11tAND}\}$} \label{alg2:no-cycle}
\lIf(\tcp*[f]{no more possible basic classes?}){$PC = \emptyset$}{\Return null} \label{alg2:no-class-left}
\Return $S$\; \label{alg2:return}
\caption{expand($M$,$i$, $o$)}
\end{algorithm}

It is easy to see, that during the traversal we discover
only nodes that must be in any well-nested subnet with the provided
interface nodes. This holds because, if there is in $M$ an
edge $(n_1, n_2)$, and $n_1$ is in a subnet that is a WF net,
but cannot be an output node of that WF net, then $n_2$ is also in that
subnet. So we argue that the algorithm finds only non-trivial subnets
that are well-nested and of a basic AND-OR class. Otherwise it returns \emph{null}.

First, observe that in lines \ref{alg2:input-test} and \ref{alg2:output-test} we check if a node has incoming
(outgoing) edges from (to) exactly the same nodes as the initial input
(output) node. However, for a well-nested subnet $S$ it is required that
all nodes in $I_{S}$ ($O_{S}$) must have incoming (outgoing) edges
from (to) the same nodes \emph{outside} $S$. We argue that our check
is correct, which can be reasoned by examining all the basic AND-OR
classes. For pAND and 11tAND nets no cycles are allowed,
so for input (output) nodes all incoming (outgoing) edges must come
from outside. In tOR nets transitions have at most one incoming
(outgoing) edges internally and none if they are input (output) node,
so all incoming edges come from outside. Finally, for 11pOR
nets only the initial single input and single output nodes are possible
and no other internal node can have edges from outside. Now we know
that if a subnet is found, it must be a WF net, since (1) $M$ is assumed
to be a WF net, (2) all the input (output) nodes are guaranteed to be connected
to the same nodes outside and either all input (output) node of $M$ or none, and (3) we include all the nodes reachable
forwards (backwards) from any of the input (output) nodes.

Next, during traversal, we maintain the set of possible basic AND-OR
classes that are still compatible with the discovered net. It is initiated
in the lines \ref{alg2:init-pc} through \ref{alg2:end-init-pc}, can be updated
in lines \ref{alg2:extra-input} and \ref{alg2:extra-output} if the expanded
net is observed to not to be one input one output net, respectively,
and in lines \ref{alg2:no-and-prop} and \ref{alg2:no-or-prop} if the AND or OR properties, respectively,
stop to be satisfied. Note that if $n'$ here satisfies the AND
and OR properties it will continue to do so, since all its predecessors
and successors have already been added to $S$ in the preceding steps
on lines \ref{alg2:add-pred} and \ref{alg2:add-foll} unless it is an input or output node. 
Finally, on line \ref{alg2:no-cycle} we examine
the acyclicity property, which is required for AND nets. This guaranties
that we know if the subnet is of a basic AND-OR class. If no compatible
basic AND-OR classes are left (see line \ref{alg2:no-class-left}), a $null$ is returned.

At the end of the algorithm, on line \ref{alg2:return} the found subnet is returned, 
which is spanned by $S$ and is a well-nested, internal subnet of $M$ that is in
at least one basic AND-OR class.
This concludes the argument for correctness of Algorithm~\ref{alg:expand}.

\subsection{Using the reduction algorithm for verifying AND-OR nets}

The presented algorithm can also be used to determine if a WF net is an AND-OR net or not. This is because the result of the algorithm is a one-node WF net iff the input is an AND-OR net. The show this, we first present a characterisation of AND-OR nets in terms of AND-OR reductions.

\begin{thm}[AND-OR net bottom-up characterisation] \label{thm:and-or_bottom-up} 
A WF net $M$ is and AND-OR net iff there is a one-node WF net $M'$ such that $M \contr^* M'$.
\end{thm}

\begin{proof}
The \emph{if}-part follows straightforwardly from the definition of AND-OR nets
and the observation that contractions are the inverse of substitutions.

The \emph{only-if}-part can be shown as follows. If $M$ is an AND-OR net there is substitution expression, e.g., $(A \otimes_{n_1} B) \otimes_{n_2} (C \otimes_{n_3} D)$, with all mentioned nets being basic AND-OR nets. Since substitution is associative this will be equivalent with a substitution expression where all brackets are moved to the left, e.g., $((A \otimes_{n_1} B) \otimes_{n_2} C) \otimes_{n_3} D$. It follows that we can reverse the substitutions  and perform them as contractions in reverse, e.g., we contract $D$, $C$, $B$ and $A$ in that order, and then obtain a one-node WF net.
\end{proof}

Finally, we show that this characterisation coincides with the net-based AND-OR reduction, which is what the algorithm computes, returning a one-node WF net.

\begin{thm}[equivalence of the reductions for AND-OR net verification]
Given a WF net $M$, the following statements are equivalent:
\begin{enumerate}[(i)]

  \item The application of the net-based AND-OR reduction to $M$ results in a one-node WF net.
  
  \item The application of the class-based AND-OR reduction to $M$ results in a class of one-node WF nets.
  
  \item The WF net $M$ can be reduced to a one-node WF net by zero or more contractions.
\end{enumerate}
\end{thm}

\begin{proof}
It is easy to see that (i) implies (iii), since a net-based AND-OR reduction is in fact a series of contractions. That (ii) implies (i) follows from   Theorem~\ref{thm:similarity-reduction}. So it remains to be shown that (iii) implies (ii).

Assume there is a list of contractions $M \contr M_1 \contr \ldots \contr M_n$ such that $M_n$ is a one-node WF net. For each contraction $M_i \contr M_{i+1}$ it holds that either it contracts a subnet of one node, in which case $M_i \sim M_{i+1}$ and therefore $[M_i] = [M_{i+1}]$, or it contracts a non-trivial subnet, in which case $[M_i] \contr [M_{i+1}]$. It follows that there is a series of class-based contractions $[M] \clcontr [M'_1] \clcontr \ldots \clcontr [M'_m]$ such that $M'_m \sim M_n$. It then follows that $M'_m$ is a one-node net and therefore contains no contractible non-trivial subnet, and so $[M'_m]$ is indeed the result of the class-based AND-OR reduction.
\end{proof}

\section{Summary and Future Research}
\label{sect:summary}

In this paper we introduce a notion of AND-OR reduction,
which reduces a WF net to a smaller net by iteratively contracting  certain
well-formed subnets into single nodes until no more such contractions
are possible. This reduction is interesting for several reasons. The first reason
is that it preserves the soundness and unsoundness of the net, so can be
used to help users understand why a WF net is problematic. It might
also give some valuable
insights for determining the different possible decompositions, e.g., how parts
of the workflow can be best distributed to independent organisational units or, in
case of workflows representing computations, e.g., as in scientific
workflows, to different servers.
The second reason is that it can provide as a side-effect a hierarchical structure of
parts, or the whole, of the WF net, which can help user to understand
the structure or large WF nets. 

Finally, the reduction can be used to show
that a certain WF net is an AND-OR net, because in that case the net
is reduced to a one-node WF net. This class of WF nets was introduced
in earlier work, and arguably describes nets that follow good
hierarchical design principles which can be compared to structured design of programs
and using well nested procedures and functions, rather than unrestricted
goto statements. As was shown in earlier work, these nets have the desired
soundness property.

It is shown that  the AND-OR reduction is confluent up to isomorphism, which
means that despite the inherent non-determinism that comes from
the choice of subnets that are contracted, the final result of the reduction is always
the same up to the choice of the identity of the nodes. Based on this result, an algorithm
is presented that computes the AND-OR reduction, and runs in polynomial time.

As a byproduct of the reduction procedure, a refinement tree for
the hierarchical structure of the net can be constructed, like has been done for similar
classes of nets~\cite{wachtel2003,wachtel2006,PChPGBL13}. It is worth to
investigate how such refinement trees can be used to determine efficiently sound
markings and to model recovery regions. Moreover, it is also interesting to
investigate to what extent this hierarchy is unique, or could be made unique by
normalising it in a certain way, since that could make
it more effective as a tool for understanding the structure of a net. For example,
in a refinement tree, if a parent and a child both represent contractions of AND nets,
then they can be merged into a single contraction of a
larger AND net.  Another source of ambiguity comes from the observation that linear nets
are simultaneously AND and OR nets. It would be interesting to investigate
if these types of ambiguity can describe all the ambiguity, and if a suitable normal form
can be defined that always gives a unique and meaningful refinement hierarchy.
A possible normalisation procedure could, for example, consist of the following
rules: (1) if a parent and a child in the refinement tree represent both the contraction of an AND net or both
and OR net, then they should be merged, (2) if a parent represents the contraction of an
AND net and it has a child that represents the contraction of an OR net that is 
is equal to the result of the substitution of several OR nets into an AND net
(which must be a linear net, since it is also part of an OR net), then
the AND net of the child is substituted into the AND net of the parent and the OR nets
become the new children that replace the old child, 
and (3) the same as rule (2) but with the roles of AND and OR interchanged.
We conjecture that this leads to a unique normal form for the refinement tree that will
be the same if the result of the refinement tree is the same. 

Finally, another interesting direction for future research is to extend the class
of sound free-choice nets that can be generated as AND-OR nets by
having additional substitution rules for edges, and to check if a similar verification
procedure by reduction would be possible. We could define ptAND, tpAND, ptOR and
tpOR nets where the small letters indicate the type of input and output
nodes, respectively. So an edge from a place to a transition could
be replaced with a tpAND or tpOR net. Note that the original place
and transition remains present. However, this could turn a free-choice
net into a non-free-choice net, and indeed create a non-sound net.
However, for edges from transitions to places this does not happen
as long as we require the ptAND net to be a one-output net, and the
ptOR net to be a one-input net. It is also not hard to see that such
substitutions preserve soundness. It would be interesting to investigate
to what extent this would come close to generating all choice-free substitution-sound
WF nets.

\section*{Acknowledgements}

This research was sponsored by National Science Centre based on decision
DEC-2012/07/D/ST6/02492.

\bibliographystyle{abbrv}
\bibliography{hierarchical}

\end{document}